\documentclass[12pt]{thesis}%
\usepackage{titlesec}
\usepackage{graphicx}
\usepackage{cite}
\usepackage{lscape}
\usepackage{indentfirst}
\usepackage{latexsym}
\usepackage{multirow}
\usepackage{tabls}
\usepackage{wrapfig}
\usepackage{slashbox}
\usepackage{longtable}
\usepackage{supertabular}
\usepackage{subfigure}
\usepackage{amsmath}
\usepackage{amsthm}
\usepackage{amsfonts}
\usepackage{amssymb}%
\providecommand{\U}[1]{\protect\rule{.1in}{.1in}}
\titleformat{\chapter}
{\normalfont\large}{Chapter \thechapter:}{1em}{}

\renewcommand{\baselinestretch}{2}
\newtheorem{theorem}{Theorem}

\newtheorem{corollary}[theorem]{Corollary}

\newtheorem{lemma}[theorem]{Lemma}

\allowdisplaybreaks
\begin{document}

\newpage\
\newpage

\begin{center}
APPROVAL SHEET \newline
\end{center}

\noindent Title of Dissertation: Amplified Quantum Transforms\newline\newline
Name of Candidate: David Jonathan Cornwell\newline\hspace*{3.8cm}{Doctor of
Philosophy, 2014} \newline\newline Dissertation and Abstract Approved:
\rule{4cm}{.1mm}\newline\newline\hspace*{7cm}{Dr. S. Lomonaco} \newline%
\hspace*{7cm}{Professor} \newline\hspace*{7cm}{CSEE}\newline\newline\newline
Date Approved: \rule{4cm}{.1mm}\newline

\newpage

\begin{center}
CURRICULUM VITAE \newline\newline
\end{center}

\noindent1980 Bachelor of Science, Mathematics, University of York, UK
\newline1995 Master of Science, Statistics, UMBC, USA \newline2014 Doctor of
Philosophy, Applied Mathematics, UMBC, USA \newline

\newpage

\hbox{\ }

\renewcommand{\baselinestretch}{1}

\begin{center}
{\normalsize {\large {ABSTRACT} }}

{\normalsize {\large \vspace{3em} }}
\end{center}

{\normalsize \hspace{-0.15in}
\begin{tabular}
[c]{ll}%
Title of dissertation: & {\large Amplified Quantum Transforms}\\
& \\
& \\
\  & \\
& {\large David J. Cornwell, Doctor of Philosophy}\\
\  & 2014\\
Dissertation directed by: & {\large Professor Samuel J. Lomonaco}\\
& {\large Department of Computer Science}\\
& UMBC
\end{tabular}
}

{\normalsize \vspace{3em} }%

\renewcommand{\baselinestretch}{2}\small\normalsize

In this thesis we investigate two new Amplified Quantum Transforms. In
particular we create and analyze the Amplified Quantum Fourier Transform
(Amplified-QFT) and the Amplified-Haar Wavelet Transform. The Amplified-QFT
algorithm is used to solve the following problem:

\textbf{The Local Period Problem: }\textit{Let }$L=\{0,1,...,N-1\}$\textit{ be
a set of }$N$\textit{ labels and let }$A$\textit{ be a subset of }$M$\textit{
labels of period }$P$\textit{, i.e. a subset of the form}%
\[
A=\{j:j=s+rP,r=0,1,...,M-1\}
\]
\textit{ where }$P\leq\sqrt{N\text{ }}$\textit{and }$M<<N$\textit{, and where
}$M$\textit{ is assumed known. Given an oracle}%
\[
f:L\rightarrow\{0,1\}
\]
\textit{ which is }$1$\textit{ on }$A$\textit{ and }$0$\textit{ elsewhere,
find the local period }$P$\textit{ and the offset }$s$\textit{.}

First, we provide a brief history of quantum mechanics and quantum computing.

Second, we examine the Amplified-QFT in detail and compare it against the
Quantum Fourier Transform (QFT) and Quantum Hidden Subgroup (QHS) algorithms
for solving the Local Period Problem. We calculate the probabilities of
success of each algorithm and show the Amplified-QFT\ is quadratically faster
than the QFT and QHS algorithms.

Third, we examine the Amplified-QFT algorithm for solving The Local Period
Problem with an Error Stream.

Fourth, we produce an uncertainty relation for the Amplified-QFT algorithm.

Fifth, we show how the Amplified-Haar Wavelet Transform can solve the Local
Constant or Balanced Signal Decision Problem which is a generalization of the
Deutsch-Jozsa algorithm.


\newpage

\thispagestyle{empty} \hbox{\ }
\renewcommand{\baselinestretch}{1}\small\normalsize

\begin{center}
AMPLIFIED\ QUANTUM\ TRANSFORMS \newline\ \newline\ \newline by \newline%
\ \newline David J. Cornwell\newline\newline2014\ \newline\ \newline%
\ \newline\ \newline Dissertation submitted to the Faculty of the Graduate
School of the \newline University of Maryland, Baltimore County in partial
fulfillment \newline of the requirements for the degree of \newline Doctor of Philosophy

2014
\end{center}

\vspace{7.5em}

\noindent Advisory Committee: \newline Professor Samuel J. Lomonaco,
Chair/Advisor \newline Professor Thomas Armstrong, Co-Advisor \newline
Professor Muddappa Gowda \newline Professor Florian Potra\newline Professor
Yanhua Shih\newline



\newpage

\thispagestyle{empty} \hbox{\ }

\vfill%
\renewcommand{\baselinestretch}{1}\small\normalsize
{\normalsize \vspace{-0.65in} }

\begin{center}
\copyright \hbox{ }Copyright by David J. Cornwell 2014
\end{center}

\vfill


\newpage\ 

\newpage
\pagestyle{plain} \pagenumbering{roman} \setcounter{page}{1}

%

\renewcommand{\baselinestretch}{2}\small\normalsize
\hbox{\ }


\begin{center}
Acknowledgements
\end{center}

\vspace{1ex}

I would like to thank my parents for supporting my university education. I
would like to thank my fiancee, Ivone de Lima, for being so considerate while
I\ worked on my PhD since 2006. I would like to wish my three sons, Tim, Zac
and Nic a terrific future.

I would also like to express a deep thank you to my thesis advisor Dr
Lomonaco. It has been a great pleasure discussing ideas and working closely
with him over the years. Also I would like to thank the UMBC\ Math Department
for enabling my cross disciplinary PhD to occur and allowing me to work on my
PhD over so many years while I\ worked full time. This has been a terrific experience.

%

\renewcommand{\baselinestretch}{1}
\small\normalsize

\newpage%

\tableofcontents

\newpage
\addcontentsline{toc}{chapter}{List of Abbreviations}

%

\renewcommand{\baselinestretch}{1}\small\normalsize
\hbox{\ }

\vspace{-4em}

\begin{center}
List of Abbreviations
\end{center}

\vspace{3pt}

{\normalsize
\begin{tabular}
[c]{ll}%
QFT & Quantum Fourier Transform\\
QHS & Quantum Hidden Subgroup\\
Amplified-QFT & Amplified Quantum Fourier Transform\\
& \\
&
\end{tabular}
}

\newpage\setlength{\parskip}{0em}
\renewcommand{\baselinestretch}{2}\small\normalsize

\setcounter{page}{1} \pagenumbering{arabic}


\newpage\renewcommand{\thechapter}{1}

\chapter{Introduction}

\section{Executive Summary}

In this thesis we analyze two new quantum algorithms. The first algorithm is
called the Amplified Quantum Fourier Transform (Amplified-QFT) which solves
the Local Period Problem (see Chapter 2) and the Local Period Problem with
Error Stream (see Chapter 3). We also produce an Uncertainty Principle for
this algorithm (see Chapter 4). The second algorithm is called the Amplified
Haar Wavelet Transform which solves the \textbf{Local Constant or Balanced
Signal Decision Problem} which is a generalization of the Deutsch-Josza
problem (See Chapter 5).

What is the Local Period Problem? This is best explained by an example.
Suppose we have a 1024 long signal of zeros and ones in positions 0 to 1023
which is nearly all zeros, except for 7 ones, which are located at positions
\{208,213,218,223,228,233,238\}. We can see that this sequence of ones has
period 5.

The Local Period Problem, is given the signal (which we call an Oracle) and
given the number of ones (7), find the period (5) and the starting position of
the sequence (the offset 208). In the notation of Chapter 2 we have
$N=1024,M=7,P=5,s=208$ and the periodic set of ones
$A=\{208,213,218,223,228,233,238\}.$

How does the Amplified-QFT\ solve this problem? We begin with a superposition
which has amplitudes of $+1/\sqrt{N}$ where the oracle is zero and
$-1/\sqrt{N}$, where the oracle is a one. We then run Grover's algorithm which
increases the amplitudes on the positions of the ones given by the set $A$ to
very close to $1/\sqrt{M},$ and decreases the amplitudes on the positions of
the zeros to very close to $0$. We then run the QFT algorithm on this state
and make a measurement to try to recover the period $P.$

\textbf{Result 1:} We show that the Amplified-QFT algorithm is, on average,
\textbf{quadratically} \textbf{faster} than two other algorithms, the Quantum
Fourier Transform (QFT) and the Quantum Hidden Subgroup (QHS) algorthms for
solving this problem. This result is obtained in section 2.8. The reason for
this is that the QFT and QHS algorithms do not amplify the amplitudes on the
set $A~$\ whereas the Amplified-QFT algorithm does. The results of both the
QFT and QHS algorithms are dominated by the number of zeros and so find it
difficult to find the period of the small set of ones. For these two
algorithms the probability of measuring the value zero is close to 1.

\textbf{Result 2:} We find the probabilities of success of each of these three
algorithms. These results are summarized for each of the algorithms in section
2.7 but are obtained in sections 2.9 for the Amplified-QFT, 2.10 for the QFT
and 2.11 for the QHS algorithms. We show that the ratio of the probabilities
of success of the Amplified-QFT to the QFT algorithms is approximately $N/4M$
whereas the ratio of the probabilities of success of the Amplified-QFT to the
QHS algorithms is approximately $N/2M$.

\textbf{Result 3:} In section 2.12 we produce two quantum algorithms for
finding the offset $s.$

\bigskip

\textbf{Result 4:} In section 2.13 we produce a general result where we
replace the QFT by a general unitary operator $U$ in the Amplified-QFT
algorithm. We find a property on $U$ such that the ratio of the probabilities
of the Amplified-U divided by U case is the same as the ratio of the
Amplified-QFT divided by QFT case.

\textbf{Result 5:} In section 2.14, we replace Grover's algorithm with a
general amplification algorithm in the Amplified-QFT algorithm and find an
upper bound on the probabilities of success in this case.

What is the Local Period Problem with Error Stream? We extend the example
described at the beginning of this executive summary.

Suppose, in addition to the 7 ones in the periodic set $A,$ there are $L=6$
additional ones introduced in random positions due to errors in the oracle. We
now have a random set $G=\{17,111,234,433,727,813\}$. How does this affect the
probability of success for the same three algorithms defined in chapter 2 and
the ability to recover the period 5?

\textbf{Result 6:} In chapter 3 section 3.1, we summarize the exact
probabilities of success which now include components of sums over a random
set. These values are obtained in sections 3.4, 3.5 and 3.6.

\textbf{Result 7:} In section 3.2, we calculate the corresponding expected
values and variances of the sums over the random set given by (where $T=L+M$)%
\[
\left\vert \frac{M}{T}\omega^{sy}+\frac{1}{T}\sum_{z\in G}\omega
^{zy}\right\vert ^{2}%
\]

and%

\[
\left\vert \frac{1}{T}\omega^{sy}\left[  \frac{1-\omega^{MPy}}{1-\omega^{Py}%
}\right]  +\frac{1}{T}\sum_{z\in G}\omega^{zy}\right\vert ^{2}%
\]

and%
\[
\left\vert \frac{1}{T}\sum_{z\in G}\omega^{zy}\right\vert ^{2}%
\]

\textbf{Result 8:} Also in section 3.2, for the Amplified-QFT algorithm, we
show that an upper bound of the expected probability of success has a minimum
value when $L=MinL$ given by%
\[
MinL=-M^{2}+\sqrt{M(M-1)(M(M-1)+N)}%
\]

This indicates that as $L$ increases to $MinL$, the upper bound of the
expected probability of success decreases, but then as $L$ increases in value
past $MinL$, the expected probability of success can increase again due to randomness.

\textbf{Result 9:} In chapter 4 we obtain an uncertainty principle for the
Amplified-QFT algorithm. Let $N=$ total number of elements, $M=$ number of
elements whose amplitudes are close to $1/\sqrt{M}$ after first running
Grover's algorithm, $N_{y}=$ number of elements which have non-zero amplitudes
after running the QFT, then we have%
\[
MN_{y}\geq N
\]

What is the \textbf{Local Constant or Balanced Signal Decision Problem}?

This is best explained by an example. Suppose we have a signal $S$ which is
1024 long consisting of zeros and ones. Suppose we are given two pairs of
locations $A_{128}=\{128,129\}$ and $A_{722}=\{722,723\}$ where the signal is
either constant or balanced at these locations and we wish to determine which
is the case. Here $A=A_{128}\cup A_{722}.$ For example we could have the
constant signal case $S(128)=0,S(129)=0,S(722)=1,S(723)=1$ or we have a
balanced signal case $S(128)=0,S(129)=1,S(722)=1,S(723)=0$.

\textbf{Result 10:} In chapter 5 section 5.2, we show that the Amplified-Haar
Wavelet Transform can solve this problem quadratically faster than a classical
algorithm to solve this problem.

\textbf{Result 11: }In general this problem cannot be solved by the Quantum
Haar Wavelet Transform alone because the values of the signal on the set
$\overline{A}$ (the complementary set of $A$) affect the results. We do need
the amplification step in order to solve this problem. However we identify a
specific case where the Quantum Haar Wavelet Transform can solve the problem
(when either we have $A$ is constant and $\overline{A}$ is balanced or $A$ is
balanced and $\overline{A}$ is constant). We show that in this case, the
Amplified-Haar Wavelet Transform is faster than the Quantum Haar-Wavelet
transform when $M>\frac{N^{1/3}(1-2M/N)^{4/3}}{2592^{1/3}(1-M/N)^{2/3}}.$

\section{A Brief History of Quantum Mechanics}

In this section we provide a brief history of quantum mechanics (see the list
of books in the references section especially books 1, 5, 14, 18, 19, 27, 29,
36 and 37).

As a material body is heated it emits radiation at different frequencies and
intensities as the temperature increases. The problem is to provide a
theoretical explanation for the observed effects. Rayleigh and Jeans applied
the principles of statistical mechanics to this problem but were not
completely successful. Their theoretical models predicted the "Ultraviolet
Catastrophe" which did not occur in practice. In 1900 Max Plank solved this
black body radiation problem by assuming that the energy of the emitted
radiation comes in energy packets or quanta and that the relationship between
energy $E$ and frequency $\upsilon$ is given by
\[
E=h\upsilon
\]

and where $h$ is Planck's constant where $h=6.626x10^{-34}Js.$ This assumption
led him to produce results for the black body radiation problem that matched
experimentally observed values which had not been done before. In 1918 he
received the Nobel prize for this work.

In 1905 Einstein explaned the photo electric effect by using Planck's quantum
approach. Light incident on a metal surface causes the emission of electrons.
The more intense the light, the more electrons of a given energy are produced.
Also light must exceed a certain minimum frequency before electrons are
emitted. Einstein produced the following formula for the photo electric effect%

\[
K=h(\upsilon-\upsilon_{0})
\]

where $\upsilon_{0}$ is the frequency of light below which the photo electric
effect does not occur, and $K$ is the energy of the emitted electron. Once
again, this approach agrees with experimentally observed results.

Around this same timeframe the model of the atom was provided by J.J. Thomson.
He had shown by experiment that atoms consist of positively and negatively
charged components. His model assumed that the positive charge was distributed
evenly throughout the atom, interspersed with negatively charged electrons.
However this classical model could not explain the line spectra of different
elements. Rutherford performed experiments concerning the scattering of alpha
particles by atoms. His experiments suggested that negatively charged
electrons orbited a central positively charged nucleus much like planets
orbiting the Sun, however problems remained. Orbiting electrons should emit
radiation and fall into the nucleus. The atom should only exist for a very
short time.

Niels Bohr decided that a model based on the quantum approach was needed.
Suppose the different energy levels of an atom are given by $E_{1},E_{2}...$
then the difference between these energy levels should be discrete values
given by%
\[
h\upsilon_{m,n}=E_{m}-E_{n}%
\]

where $\upsilon_{m,n\text{ }}$ is the frequency of light emitted when the atom
moves from the excited state $E_{m}$ to $E_{n}.$ The observed line spectra
could be explained using the formula%
\[
E_{n}=-\frac{Rh}{n^{2}}%
\]

By assuming the electrons moved in a circular orbit and the electrostatic
attraction force was balanced by the centrifugal force, Bohr was able to
obtain a theoretical value for $R$ from the formula%
\[
R=\frac{4\pi^{2}e^{4}m_{e}}{h^{3}}%
\]

which agreed with observation (where $e$ is the charge of the electron and
$m_{e}$ is its mass). Sommerfeld extended this work to the case of elliptical
orbits. \ Bohr's theory of the atom was successful and he created an institute
in Copenhagen for atomic studies.

In 1926 Schrodinger developed his famous wave equation which he used to
explain the spectral lines of the Hydrogen atom. This equation has the
following general form for the time dependent case%
\[
i\overline{h}\frac{\partial}{\partial t}\Psi=H\Psi
\]
where $H$ is the Hamiltonian operator and $\Psi$ is the state vector or wave
function of the system. The wave equation introduces the fundamental concept
of superposition for if $\Psi_{1}$ is a solution and $\Psi_{2}$ is a solution
then $\Psi_{1}+\Psi_{2}$ is also a solution by linearity.

Heisenberg developed his matrix mechanics formulation of quantum mechanics
which was shown to be equivalent to Schrodinger's wave equation version.
Heisenberg also discovered his famous Uncertainty Principle which is a
relationship between two complementary or conjugate variables such as position
and momentum.%
\[
\sigma_{x}\sigma_{p}\geq\frac{\overline{h}}{2}%
\]

where $\sigma$ is the standard deviation of the appropriate variable.

The Copenhagen Interpretation of quantum mechanics was put forward by Niels
Bohr. This contained the elements of unreality, non-locality and uncertainty.
Einstein challenged these principles in an ongoing and great debate with Niels
Bohr at the 1927 and 1930 Solvay Conferences culminating in the famous
Einstein, Podolsky and Rosen (EPR)\ paper of 1935. In this paper EPR\ claimed
that quantum mechanics should have the elements of reality, locality and
certainty which could be achieved by a Hidden Variable theory - a classical
theory. However in 1964 John Bell showed that certain correlations in quantum
theory would be much stronger than those of a hidden variable theory and he
produced Bell's Inequality which would identify which theory was true. This
meant that one could tell from performing an experiment whether quantum
mechanics was the correct theory or whether a hidden variable theory was the
correct theory. Many experiments have been performed that show quantum
mechanics is the true theory and have ruled out most hidden variable theories.
However each experiment performed so far has not ruled out all hidden variable
theories. Some loopholes have remained. In the future experiments will be
performed that will eventually rule out all the remaining loopholes but if we
apply the induction argument for theories we can say the probability that
quantum mechanics is true is currently very close to 1 and the probability
that there is a true hidden variable theory is very close to 0. Alternatively
if we use Karl Popper's approach we would say that the theory of quantum
mechanics has not been refuted. However there is still the chance it could be
refuted in favor of a conjectured hidden variable theory.

\section{A Brief History Of Quantum Computing}

In this section we provide a brief history of the development of quantum
computing to set the stage for this thesis.

The early days of quantum computing were kicked off with ideas from Paul
Benioff and Richard Feynman. Benioff investigated the idea of whether quantum
systems could efficiently simulate classical computers. In 1981 Richard
Feynman investigated the question whether a classical computer could simulate
a classical or quantum system exactly. In 1985 Feynman investigated the
notions of reversibility and irreversibility in computation.

Then in 1985 David Deutsch wrote a ground breaking paper entitled "Quantum
Theory, the Church-Turing Principle and the Universal Quantum Computer", in
the Proceedings of the Royal Society in which he replaced Turing's classical
ideas of computation with quantum ideas. Quantum computing was born. Deutsch
also presented the first quantum algorithm using two qubits, Deutsch's
algorithm which was slightly faster than a classical computer.

In 1992, David Deutsch and Richard Josza produced their Deutsch-Josza
algorithm that worked on n qubits. The idea is to be able to distinguish
whether a Boolean function is balanced or constant. Classically this would
take a work factor of $2^{n}$ however the quantum algorithm produced an
exponential speedup.

In 1994 Peter Shor published a paper entitled "Polynomial-Time Algorithms for
Prime Factorization and Discrete Logarithms on a Quantum Computer" that
specified a quantum algorithm to factor large integers using the quantum
Fourier transform that was exponentially faster than classical methods. He
also described a quantum algorithm to solve the discrete log problem.
Factoring and the discrete log problem are at the heart of cryptographic
algorithms that are used to protect internet traffic. If you can easily factor
numbers then you can break the RSA encryption algorithm. If you can easily
solve the discrete log problem then you can easily break the Diffie-Hellman
key exchange method. Shor's paper ignited research and global interest in
quantum computing, both in the unclassified and classifed arenas. The race is
now on to be the first nation to build a real quantum computer that could
implement either of these attacks.

In 1996 Lov Grover wrote a paper entitled "A fast quantum mechanical algorithm
for database search" in which he described a quantum algorithm that could
seach for an item in an $\ N$ long list with a work factor of $O(\sqrt{N}%
\dot{)}$ providing a quadratic speedup over the classical $O(N)$ method. It
was later shown by Zalka (ref 158.) that this algorithm was optimal.

Since Shor's algorithm and Grover's algorithm were published there has been
significant research in the area of quantum computing algorithms. There are
several good survey articles on the quant-ph website.

\section{Outline of Thesis - Amplified Quantum\ Transforms}

The two main algorithms of quantum computing are Grover's search algorithm and
Shor's algorithm for factoring $n=pq$ by using the quantum Fourier transform.
In this thesis we combine Grover's search algorithm with the quantum Fourier
transform to solve the Local Period Problem. We call this new algorithm The
Amplified Quantum Fourier Transform (Amplified-QFT). We show that this new
algorithm solves the Local Period Problem quadratically faster than the
quantum Fourier transform alone.

In Chapter 2 we analyze the Amplified-QFT algorithm when applied to a periodic
oracle. We calculate and compare the probabilities of success of the QFT
algorithm, the quantum hidden subgroup (QHS) algorithm and the\ Amplified-QFT
algorithm. The contents of this chapter are based off the published paper ref[14].

In Chapter 3 we analyze the Amplified-QFT when applied to a periodic oracle
with an error stream and calculate and compare the probabilities of success of
the QFT algorithm, the QHS algorithm and the Amplified-QFT algorithm.

In Chapter 4 we produce an uncertainty principle for the Amplified-QFT algorithm.

In Chapter 5 we show how the one dimensional Amplified Haar Wavelet Transforms
can be used to solve a certain decision problem.

\newpage

\renewcommand{\thechapter}{2}

\chapter{The Amplified Quantum Fourier Transform}

\section{Introduction}

In this chapter we create and analyze a new quantum algorithm called the
Amplified Quantum Fourier Transform (Amplified-QFT) for solving the following problem:

\textbf{The Local Period Problem:}\textit{ Let }$L=\{0,1,...,N-1\}$\textit{ be
a set of }$N$\textit{ labels and let }$A$\textit{ be a subset of }$M$\textit{
labels of period }$P$\textit{, i.e. a subset of the form}%
\[
A=\{j:j=s+rP,r=0,1,...,M-1\}
\]
\textit{ where }$P\leq\sqrt{N\text{ }}$\textit{and }$M<<N$\textit{, and where
}$M$\textit{ is assumed known. Given an oracle }%
\[
f:L\rightarrow\{0,1\}
\]
\textit{ which is }$1$\textit{ on }$A$\textit{ and }$0$\textit{ elsewhere,
find the local period }$P$\textit{ and the offset }$s$\textit{.}

The first part of this chapter provides some background information on
amplitude amplification, period finding and defines the Amplified-QFT
algorithm. The second part of the chapter summarizes the main results and
compares the Amplified-QFT algorithm against the Quantum Fourier Transform
(QFT) and Quantum Hidden Subgroup (QHS) algorithms when solving the local
period problem. It is shown that the Amplified-QFT algorithm is, on average,
quadratically faster than both the QFT and\ QHS algorithms. The third part of
the chapter provides the detailed proofs of the main results, describes the
method of recovering $P$ from an observation $y$ and describes the algorithm
for finding the offset $s$. In the final section of the chapter we provide a
general result where we replace the QFT with a general unitary operator U and
identify what property it must have to produce the same probabilities of
success as the QFT.

\section{Background-Amplitude Amplification}

In ref[4] Lov Grover specified a quantum search algorithm that searched for a
single marked element $x0$ in an $N$ long list $L$. An oracle $f:L\rightarrow
\{0,1\}$ is used to mark the element such that $f(x0)=1$ and $f$ is $0$
elsewhere. Grover's quantum algorithm finds the element with a work factor of
$O(\sqrt{N})$ whereas on a classical computer this would take $O(N)$, thereby
obtaining a quadratic speedup. Grover's algorithm can be summarized as follows:

a) Initialize the state to be the uniform superposition state $|\psi>=H|0>$
where $H$ is the Hadamard transform.

b) Reflect the current state about the plane orthogonal to the state $|x0>$ by
using the operator $(I-2|x0><x0|).$

c)\ Reflect the new state back around $|\psi>$ by using the operator $(2$
$|\psi>$ $<\psi|-I)$. This operator is a reflection about the average of the
amplitudes of the new state.

d) Repeat steps b) and c) $O(\sqrt{N})$ times until most of the probability is
on $|x0>.$

e) Measure the resulting state to obtain $x0.$

Also in ref[4], Grover suggested this algorithm could be extended to the case
of searching for an element in a subset $A$ of $M$ marked elements\ in an $N$
long list $L$. Once again an oracle $f:L\rightarrow\{0,1\}$ is used to mark
the elements of the subset $A$. Grover's algorithm solves this problem with a
work factor of $O(\sqrt{N/M})$. The elements of the set $A$ are sometimes
referred to as "good" and the elements not in $A$ are called "bad". Grover's
algorithm for this problem can be summarized as follows:

a) Initialize the state to be the uniform superposition state $|\psi>=H|0>$
where $H$ is the Hadamard transform.

b) Reflect the current state about the plane orthogonal to the state $|xgood>$
by using the operator $(I-2|xgood><xgood|)$, where $|xgood>$ is the normalized
sum of the good states defined by the set $A$. This changes the sign of the
amplitudes of the good states defined by $A.$

c)\ Reflect the new state back around $|\psi>$ by using the operator $(2$
$|\psi>$ $<\psi|-I)$.

d) Repeat steps b) and c) $O(\sqrt{N/M})$ times until most of the probability
is on the set $A.$

e) Measure the resulting state to obtain an element in the set $A.$

Both versions of Grover's algorithm are also known as Amplitude Amplification
algorithms which are generalized even further in ref [9]. The first part of
the Amplified-QFT algorithm consists of the second of these algorithms, except
for the final measurement step e).

\section{Background-Period Finding}

In ref[3], Peter Shor describes a quantum algorithm to solve the factorization
problem with exponential speed up over classical approaches. He translates the
factorization problem into one of finding the period of the function
$a^{x}ModN$ where $N$ is the number to be factored and $\gcd(a,N)=1$. The
period is found by making use of the QFT. Shor's factorization algorithm is
summarized below:

a) Find $Q:N^{2}\leq Q<2N^{2}$

b) Find $a:\gcd(a,N)=1$

c) Find the period of $a^{x}ModN$ using the QFT and using the $Qth$ root of unity

\qquad- Form the superposition $\frac{1}{\sqrt{Q}}\sum|x>|a^{x}ModN>$

\qquad- Apply the QFT to the first register $|x>\rightarrow\sum\omega^{xy}|y>$

\qquad- Measure $y$

\qquad- Form the continued fraction expansion of $y/Q$ to find $d/P$

\qquad- If $|y/Q-d/P|<1/2N^{2}$ and $\gcd(d,P)=1$ then $\ P$ is recovered

d) If the period is not even start over

e) If $a^{P/2}+1=0ModN$ start over

f) Find $\gcd(a^{p/2}-1,N)$ to find the factor of $N.$

Step c) is the quantum part of Shor's factorization algorithm. We make use of
the QFT and continued fraction expansion method to recover the period $P$ in
the second part of the Amplified-QFT algorithm.

\section{The Amplified\ Quantum Fourier Transform Algorithm}

The Amplified-QFT algorithm solves the Local Period Problem:

\textbf{The Local Period Problem: }Let $L=\{0,1,...,N-1\}$ be a set of $N$
labels and let $A$ be a subset of $M$ labels of period $P$, i.e. a subset of
the form
\[
A=\{j:j=s+rP,r=0,1,...,M-1\}
\]
where $P\leq\sqrt{N\text{ }}$and $M<<N$, and where $M$ is assumed known. Given
an oracle%
\[
f:L\rightarrow\{0,1\}
\]
which is $1$ on $A$ and $0$ elsewhere, find the local period $P$ and the
offset $s$.

The Amplified-QFT algorithm consists of the following steps where steps a)
through d)\ are the Amplitude Amplification steps and steps e) through i) are
the period finding steps that use the QFT:

a) Initialize the state to be the uniform superposition state $|\psi>=H|0>$
where $H$ is the Hadamard transform.

b) Reflect the current state about the plane orthogonal to the state $|xgood>$
by using the operator $(I-2|xgood><xgood|)$, where $|xgood>$ is the normalized
sum of the good states defined by the set $A.$This changes the sign of the
amplitudes of the good states defined by $A.$

c)\ Reflect the new state back around $|\psi>$ by using the operator $(2$
$|\psi>$ $<\psi|-I)$.

d) Repeat steps b) and c) $O(\sqrt{N/M})$ times until most of the probability
is on the set $A.$

e) Apply the QFT\ to the resulting state

f)\ Make a measurement $y$

g) Form the continued fraction expansion of $y/N$ to find $d/P$

h) If $|y/N-d/P|<1/2P^{2}$ and $\gcd(d,P)=1$ then $\ P$ is recovered

i) If $\gcd(d,P)\neq1$ repeat the algorithm starting at step a)

The Amplified-QFT algorithm produces the following states (See later sections
for the detailed analysis of the Amplified-QFT algorithm):

After applying steps b) and c) $k$ times where $k=\left\lfloor \frac{\pi
}{4\sin^{-1}(\sqrt{M/N})}\right\rfloor $ we arrive at the following state:%
\[
|\psi_{k}>=a_{k}\sum_{z\in A}|z>+b_{k}\sum_{z\notin A}|z>
\]

\noindent where%

\[
a_{k}=\frac{1}{\sqrt{M}}\sin(2k+1)\theta,b_{k}=\frac{1}{\sqrt{N-M}}%
\cos(2k+1)\theta
\]

\noindent are the appropriate amplitudes of the states and where the angle
$\theta$ is given by
\[
\sin\theta=\sqrt{M/N},\cos\theta=\sqrt{1-M/N}%
\]

\noindent.

The QFT at step e) performs the following action%

\[
|z>\rightarrow\frac{1}{\sqrt{N}}\sum_{y=0}^{N-1}e^{-2\pi izy/N}|y>
\]

After the application of the QFT to the state $|\psi_{k}>$ , letting
$\omega=e^{-2\pi i/N}$ we arrive at the following sate:%
\[
|\phi_{k}>=\sum_{y=0}^{N-1}\left[  \frac{a_{k}}{\sqrt{N}}\sum_{z\in A}%
\omega^{zy}+\frac{b_{k}}{\sqrt{N}}\sum_{z\notin A}\omega^{zy}\right]  |y>
\]

.

At step f) we measure this state with respect to the standard basis to yield
an integer $y\in\{0,1,...,N-1\}$ from which we can determine the period $P$
using the continued fraction method.

In a later section where we summarize the main results, we provide a table
showing the probabilities of measuring $y$ for the Amplified-QFT algorithm and
compare them against the probabilites obtained by performing the QFT and QHS algorithms.

\section{The QFT Algorithm}

The QFT\ algorithm applied to the Local\ Period Problem does not include the
amplitude amplification steps and consists of the following steps:

a) Initialize the state to be the uniform superposition state $|\psi>=H|0>$
where $H$ is the Hadamard transform.

b) Apply the oracle $f$ to $|\psi>$

c) Apply the QFT to this state

d)\ Make a measurement $y$

e) Form the continued fraction expansion of $y/N$ to find $d/P$

f) If $|y/N-d/P|<1/2P^{2}$ and $\gcd(d,P)=1$ then $\ P$ is recovered

g) If $\gcd(d,P)\neq1$ repeat the algorithm starting at step a)

At step b) after applying the oracle the state is given by (See later sections
for the detailed analysis of the QFT algorithm):%

\[
|\psi_{1}>=\frac{1}{\sqrt{N}}\left[  (-2)\sum_{z\in A}|z>+\sum_{z=0}%
^{N-1}|z>\right]
\]

At step c) the QFT applies the following action:%
\[
|z>\rightarrow\frac{1}{\sqrt{N}}\sum_{y=0}^{N-1}\omega^{zy}|y>
\]

\noindent to get%
\[
|\psi_{2}>=\sum_{y=0}^{N-1}\left[  \frac{(-2)}{N}\sum_{z\in A}\omega
^{zy}+\frac{1}{N}\sum_{z=0}^{N-1}\omega^{zy}\right]  |y>
\]

At step d) we measure this state with respect to the standard basis to yield
an integer $y\in\{0,1,...,N-1\}$ from which we can determine the period $P$
using the continued fraction method. We note that in the QFT algorithm case,
we would have to repeat the algorithm many times to recover the period $P$
because as we will see shortly, most of the probability is on the state $|0>.$

\section{The QHS Algorithm}

The QHS\ algorithm is a two register algorithm and does not include the
amplitude amplification steps. It consists of the following steps:

a) Initialize the state to be the uniform superposition state $|\psi>=H|0>|0>$
where $H$ is the Hadamard transform.

b) Apply the oracle $f$ and put the result into the second register of
$|\psi>$

c) Apply the QFT to the first register of this state

d)\ Make a measurement $y$

e) Form the continued fraction expansion of $y/N$ to find $d/P$

f) If $|y/N-d/P|<1/2P^{2}$ and $\gcd(d,P)=1$ then $\ P$ is recovered

g) If $\gcd(d,P)\neq1$ repeat the algorithm starting at step a)

At step b) we have the following state (See later sections for the detailed
analysis of the QHS algorithm):%

\[
|\psi_{1}>=\frac{1}{\sqrt{N}}%
{\displaystyle\sum\limits_{x=0}^{N-1}}
|x>|f(x)>
\]

After applying the QFT the state is given by:%

\[
|\psi_{2}>=%
{\displaystyle\sum\limits_{y=0}^{N-1}}
\frac{1}{N}|y>\sum_{x=0}^{N-1}\omega^{xy}|f(x)>
\]

At step d) we measure this state with respect to the standard basis to yield
an integer $y\in\{0,1,...,N-1\}$ from which we can determine the period $P$
using the continued fraction method.We note that in the QHS algorithm case, we
would have to repeat the algorithm many times to recover the period $P$
because as we will see shortly, most of the probability is on the state $|0>.$

\section{Summary of\ the Main Results}

We summarize the main results and compare the probability $\Pr(y)$ of
measuring $y$ in the final state arrived at for each of the three algorithms:
1) the Amplified-QFT algorithm 2) the QFT algorithm and 3) the QHS algorithm.
Here $\sin\theta=\sqrt{M/N}$ and $k=\left\lfloor \frac{\pi}{4\theta
}\right\rfloor $ and $0\leq$ $\frac{\sin^{2}(\pi MPy/N)}{\sin^{2}(\pi
Py/N)}\leq M^{2}$.

Case 1 (Amplified-QFT):

The probability $\Pr(y)$ is given exactly by

\begin{center}%
\[
\left\{
\begin{tabular}
[c]{lll}%
$\cos^{2}2k\theta$ & if & $y=0$\\
&  & \\
$tan^{2}\theta\sin^{2}2k\theta$ & if & $Py=0\operatorname{mod}N,y\neq0$\\
&  & \\
$\frac{1}{M^{2}}tan^{2}\theta\sin^{2}2k\theta\frac{\sin^{2}(\pi MPy/N)}%
{\sin^{2}(\pi Py/N)}$ & if & $Py\neq0\operatorname{mod}N\text{and }%
MPy\neq0\operatorname{mod}N$\\
&  & \\
$0$ & if & $Py\neq0\operatorname{mod}N\text{ and }MPy=0\operatorname{mod}N$%
\end{tabular}
\ \ \ \ \ \right\}
\]

\end{center}

Case 2 (QFT):

\bigskip The probability $\Pr(y)$ is given exactly by%

\[
\left\{
\begin{tabular}
[c]{lll}%
$\left(  1-\frac{2M}{N}\right)  ^{2}$ & if & $y=0$\\
&  & \\
$4\frac{M^{2}}{N^{2}}$ & if & $Py=0\operatorname{mod}N,y\neq0$\\
&  & \\
$\frac{4}{N^{2}}\frac{\sin^{2}(\pi MPy/N)}{\sin^{2}(\pi Py/N)}$ & if &
$Py\neq0\operatorname{mod}N\text{and }MPy\neq0\operatorname{mod}N$\\
&  & \\
$0$ & if & $Py\neq0\operatorname{mod}N\text{ and }MPy=0\operatorname{mod}N$%
\end{tabular}
\ \ \ \ \ \right\}
\]

Case 3 (QHS):

\bigskip The probability $\Pr(y)$ is given exactly by%

\[
\left\{
\begin{tabular}
[c]{lll}%
$1-\frac{2M(N-M)}{N^{2}}$ & if & $y=0$\\
&  & \\
$\frac{2M^{2}}{N^{2}}$ & if & $Py=0\operatorname{mod}N,y\neq0$\\
&  & \\
$\frac{2}{N^{2}}\frac{\sin^{2}(\pi MPy/N)}{\sin^{2}(\pi Py/N)}$ & if &
$Py\neq0\operatorname{mod}N\text{and }MPy\neq0\operatorname{mod}N$\\
&  & \\
$0$ & if & $Py\neq0\operatorname{mod}N\text{ and }MPy=0\operatorname{mod}N$%
\end{tabular}
\ \ \ \ \ \ \right\}
\]

We note that for the QFT and QHS algorithms $\Pr(y=0)$ is very close to $1$
because $M<<N.$ In the cases where $y\neq0$ we compare the ratios of $\Pr(y)$
\ in the Amplified-QFT and QFT case and then in the Amplified-QFT and
QHS\ case. Let $y$ be fixed such that either

1. $Py=0\operatorname{mod}N,y\neq0$ or

2. $Py\neq0\operatorname{mod}N$and $MPy\neq0\operatorname{mod}N$

and define $\Pr Ratio(y)=\Pr(y)_{Amplified-QFT}/\Pr(y)_{QFT}$ \ then we have
the following (see the later detailed sections)%

\begin{align*}
\frac{N}{4M}(\frac{N}{N-M})  &  \geq\Pr Ratio(y)\geq\frac{N}{4M}(\frac{N}%
{N-M})(1-\frac{2M}{N})^{2}\\
&  \Longrightarrow\Pr Ratio(y)\approx\frac{N}{4M}%
\end{align*}

and define $\Pr Ratio(y)=\Pr(y)_{Amplified-QFT}/\Pr(y)_{QHS}$ \ then we have
the following%

\begin{align*}
\frac{N}{2M}(\frac{N}{N-M})  &  \geq\Pr Ratio(y)\geq\frac{N}{2M}(\frac{N}%
{N-M})(1-\frac{2M}{N})^{2}\\
&  \Longrightarrow\Pr Ratio(y)\approx\frac{N}{2M}%
\end{align*}

Let $S_{ALG}=\{y:|\frac{y}{N}-\frac{d}{P}|\leq\frac{1}{2P^{2}},(d,P)=1\}$ be
the set of "successful" $y$'s. That is $S_{ALG\text{ }}$consists of those
$y$'s which can be measured after applying one of the three algorithms denoted
by $ALG$ and from which the period $P$ can be recovered by the method of
continued fractions. Note that the set $S_{ALG}$ is the same for each
algorithm. However the probability of this set varies with each algorithm. We
can see from the following that given $y1$ and $y2$, whose probability ratios
satisfy the same inequality, we can add their probabilities to get a new ratio
that satisfies the same inequality. In this way we can add probabilities over
a set on the numerator and denominator and maintain the inequality:
\begin{align*}
A  &  >\frac{P(y1)}{Q(y1)}>B\text{ }and\text{ }A>\frac{P(y2)}{Q(y2)}>B\\
&  \Longrightarrow A>\frac{P(y1)+P(y2)}{Q(y1)+Q(y2)}>B
\end{align*}

\noindent We see from the cases given above that%
\[
\frac{N}{4M}(\frac{N}{N-M})\geq\frac{\Pr(S_{Amplified-QFT})}{\Pr(S_{QFT})}%
\geq\frac{N}{4M}(\frac{N}{N-M})(1-\frac{2M}{N})^{2}%
\]

\noindent where the difference between the upper bound and lower bound is
exactly 1 and that%

\[
\frac{N}{2M}(\frac{N}{N-M})\geq\frac{\Pr(S_{Amplified-QFT})}{\Pr(S_{QHS})}%
\geq\frac{N}{2M}(\frac{N}{N-M})(1-\frac{2M}{N})^{2}%
\]

\noindent where the difference between the upper bound and lower bound is
exactly 2.

This shows that the Amplified-QFT is approximately$\frac{N}{4M}$ times more
successful than the QFT and $\frac{N}{2M}$ times more successful than the QHS
when $M<<N$. In addition it also shows that the QFT\ is $2$ times more
successful than the QHS in this problem. However, the success of the
Amplified-QFT algorithms comes at an increase in work factor of $O(\sqrt
{\frac{N}{M}})$. We note that in the case that P is a prime number that
$(d,P)=1$ is met trivially. However when P is composite the algorithms may
need to be rerun several times until $(d,P)=1$ is satisfied.

Towards the end of the chapter we show how to test whether a putative value of
$P$, given $s$ is known, can be tested to see if it is the correct value. We
also investigate the case where $s$ is unknown but is from a small known set
of values such that the values of $s$ can be exhausted over on a classical
computer. We also show how $s$ can be recovered by using a quantum algorithm
using amplitude amplification followed by a measurement.

\section{The Amplified-QFT\ is Quadratically Faster than the QFT\ or the QHS}

We show that the Amplified-QFT algorithm is, on average, quadratically faster
than the QFT\ or QHS\ algorithms. In order to show this, we use the geometric
probability distribution which provides the probability of the first success
in a sequence of trials where the probability of success is $p$ and the
probability of failure is \thinspace$1-p.$ \ For both the QFT\ and\ QHS
algorithms a trial is one complete execution of the algorithm. Because the
probability of measuring $y=0$ is close to $1$ we expect to have to repeat the
algorithm many times due to failure of measuring a successful $y$, before we
have the first success.

If $X$ is the random variable which counts the number of trials until the
first success then%

\[
P(X=k)=(1-p)^{k-1}p\text{ for }k=1,2...
\]

The expected value $E[X]$ and variance $Var[X]$ are given by:%

\[
E[X]=\frac{1}{p}\text{ and }Var[X]=\frac{1-p}{p^{2}}%
\]

The workfactor of the Amplified-QFT algorithm is given by the number of
iterations of each amplification step followed by a single QFT step:%

\[
O(\sqrt{\frac{N}{M}})
\]

For the QFT algorithm we have the probability of failure $\ 1-p$ is given by%

\[
\Pr(failure)=1-p\geq\Pr(y=0)=(1-\frac{2M}{N})^{2}%
\]

then%

\[
\Pr(success)=p\leq1-(1-\frac{2M}{N})^{2}=\frac{4M}{N}(1-\frac{M}{N})
\]

Then for the QFT algorithm, the expected number of trials until the first
success is%

\[
E[X]=\frac{1}{p}\geq\frac{N}{4M(1-\frac{M}{N})}\geq\frac{N}{4M}%
\]

The workfactor of the QFT\ algorithm is the expected number of times the QFT
has to be run, is given approximately by:%

\[
O(\frac{N}{M})
\]

Therefore the ratio of the expected work factor of the QFT algorithm and the
work factor of the Amplified-QFT is given by%

\[
O(\sqrt{\frac{N}{M}})
\]

showing that the Amplified-QFT\ algorithm is, on average, quadratically faster
than the QFT algorithm.

The variance in the number of times the QFT algorithm is run is given by%

\[
Var[X]=\frac{1-p}{p^{2}}\geq(\frac{N}{N-M})^{2}(\frac{N-2M}{4M})^{2}%
\]

For the QHS algorithm we have the probability of failure $\ 1-p$ is given by%

\[
\Pr(failure)=1-p\geq\Pr(y=0)=1-\frac{2M(N-M)}{N^{2}}%
\]

then%

\[
\Pr(success)=p\leq1-(1-\frac{2M(N-M)}{N^{2}})=\frac{2M}{N}(1-\frac{M}{N})
\]

Then for the QHS algorithm, the expected number of trials until the first
success is%

\[
E[X]=\frac{1}{p}\geq\frac{N}{2M(1-\frac{M}{N})}\geq\frac{N}{2M}%
\]

The workfactor of the QHS\ algorithm is the expected number of times the QHS
has to be run, is given approximately by:%

\[
O(\frac{N}{M})
\]

Therefore the ratio of the expected work factor of the QHS algorithm and the
work factor of the Amplified-QFT is given by%

\[
O(\sqrt{\frac{N}{M}})
\]

showing that the Amplified-QFT\ algorithm is, on average, quadratically faster
than the QHS algorithm.

The variance in the number of times the QHS algorithm is run is given by%

\[
Var[X]=\frac{1-p}{p^{2}}\geq(\frac{N}{N-M})^{2}(\frac{(N-M)^{2}+M^{2}}{4M^{2}%
})
\]

\section{The Amplified-QFT Algorithm - Detailed Analysis}

In this section we examine the Amplified-QFT algorithm in detail and produce
the results for the probability of success that were summarized in an earlier section.

The Amplified-QFT algorithm is defined by the following procedure (see earlier section):

\emph{Steps a) to d):} Apply the Amplitude Amplification algorithm to the
starting state $|0>$. The resulting state is given by $|\psi_{k}>$ (ref[4],
ref[7],ref[1]) where $k=\left\lfloor \frac{\pi}{4\sin^{-1}(\sqrt{M/N}%
)}\right\rfloor $:%
\[
|\psi_{k}>=a_{k}\sum_{z\in A}|z>+b_{k}\sum_{z\notin A}|z>
\]

\noindent where%

\[
a_{k}=\frac{1}{\sqrt{M}}\sin(2k+1)\theta,b_{k}=\frac{1}{\sqrt{N-M}}%
\cos(2k+1)\theta
\]

\noindent are the appropriate amplitudes of the states and where
\[
\sin\theta=\sqrt{M/N},\cos\theta=\sqrt{1-M/N}%
\]

\noindent Now we have , ref[7],

$k=\left\lfloor \frac{\pi}{4\theta}\right\rfloor $ $\Longrightarrow$
$\frac{\pi}{4\theta}-1\leq k\leq\frac{\pi}{4\theta}$ $\Longrightarrow$
$\frac{\pi}{2}-\theta\leq(2k+1)\theta\leq\frac{\pi}{2}+\theta$

$\Longrightarrow\sin\theta=\cos(\frac{\pi}{2}-\theta)\geq\cos(2k+1)\theta
\geq\cos(\frac{\pi}{2}+\theta)=-\sin\theta$

\noindent Notice that the total probability of the N-M labels that are not in
A is%

\begin{align*}
(N-M)(\frac{1}{\sqrt{N-M}}\cos(2k+1)\theta)^{2}  &  =\cos^{2}(2k+1)\theta\\
&  \Longrightarrow\cos^{2}(2k+1)\theta\leq\sin^{2}\theta=\sin^{2}(\sin
^{-1}(\sqrt{\frac{M}{N}}))\\
&  \Longrightarrow\cos^{2}(2k+1)\theta\leq\frac{M}{N}%
\end{align*}
whereas the total probability of the M labels in A is
\begin{align*}
M(\frac{1}{\sqrt{M}}\sin(2k+1)\theta)^{2}  &  =\sin^{2}(2k+1)\theta=1-\cos
^{2}(2k+1)\theta\\
&  \Longrightarrow\sin^{2}(2k+1)\theta\geq1-\frac{M}{N}%
\end{align*}
.

\emph{Step e):} Apply the QFT which performs the following action%

\[
|z>\rightarrow\frac{1}{\sqrt{N}}\sum_{y=0}^{N-1}e^{-2\pi izy/N}|y>
\]

After the application of the QFT to the state $|\psi_{k}>$ , letting
$\omega=e^{-2\pi i/N}$ , we have%
\[
|\phi_{k}>=\frac{a_{k}}{\sqrt{N}}\sum_{z\in A}\sum_{y=0}^{N-1}\omega
^{zy}|y>+\frac{b_{k}}{\sqrt{N}}\sum_{z\notin A}\sum_{y=0}^{N-1}\omega^{zy}|y>
\]

After interchanging the order of summation, we have%
\[
|\phi_{k}>=\sum_{y=0}^{N-1}\left[  \frac{a_{k}}{\sqrt{N}}\sum_{z\in A}%
\omega^{zy}+\frac{b_{k}}{\sqrt{N}}\sum_{z\notin A}\omega^{zy}\right]  |y>
\]

.

\emph{Steps f) to i):} Measure with respect to the standard basis to yield a
integer $y\in\{0,1,...,N-1\}$ from which we can determine the period P using
the continued fraction method.

The amplitude $Amp(y)$ of $|y>$ is given by
\begin{align*}
Amp(y)  &  =\frac{a_{k}}{\sqrt{N}}\sum_{z\in A}\omega^{zy}+\frac{b_{k}}%
{\sqrt{N}}\sum_{z\notin A}\omega^{zy}\\
&  =\frac{(a_{k}-b_{k})}{\sqrt{N}}\sum_{z\in A}\omega^{zy}+\frac{b_{k}}%
{\sqrt{N}}\sum_{z=0}^{N-1}\omega^{zy}\\
&  =\frac{(a_{k}-b_{k})}{\sqrt{N}}\sum_{r=0}^{M-1}\omega^{(s+rP)y}+\frac
{b_{k}}{\sqrt{N}}\sum_{z=0}^{N-1}\omega^{zy}\text{ (A is periodic)}\\
&  =\frac{(a_{k}-b_{k})}{\sqrt{N}}\omega^{sy}\sum_{r=0}^{M-1}\omega
^{rPy}+\frac{b_{k}}{\sqrt{N}}\sum_{z=0}^{N-1}\omega^{zy}%
\end{align*}

We calculate the $\Pr(y)$ for the following cases:

\qquad a) $y=0$

\qquad b) $Py=0\operatorname{mod}N$ and $y\neq0$

\qquad c) $Py\neq0\operatorname{mod}N$

\subsection{Amplified-QFT Analysis: y=0}

\noindent We calculate the probability $\Pr(y=0)$%

\begin{align*}
Amp(y)  &  =\frac{a_{k}}{\sqrt{N}}\sum_{z\in A}\omega^{zy}+\frac{b_{k}}%
{\sqrt{N}}\sum_{z\notin A}\omega^{zy}\\
&  =\frac{1}{\sqrt{N}}(Ma_{k}+(N-M)b_{k})\\
&  =\frac{1}{\sqrt{N}}\left[  \frac{M}{\sqrt{M}}\sin(2k+1)\theta+\frac
{N-M}{\sqrt{N-M}}\cos(2k+1)\theta\right] \\
&  =\sqrt{\frac{M}{N}}\sin(2k+1)\theta+\sqrt{1-\frac{M}{N}}\cos(2k+1)\theta\\
&  =\sin\theta\sin(2k+1)\theta+\cos\theta\cos(2k+1)\theta\\
&  =\cos(2k\theta)
\end{align*}

\noindent We have
\[
\Pr(y=0)=\cos^{2}(2k\theta)
\]

\subsection{Amplified-QFT Analysis: $Py=0\operatorname{mod}N,y\neq0$}

We calculate the probability $\Pr(y)$ in the case where
$Py=0\operatorname{mod}N,y\neq0$

Using the fact that%
\[
\sum_{z=0}^{N-1}\omega^{zy}=\frac{1-\omega^{Ny}}{1-\omega^{y}}=0,w^{y}\neq1
\]

\noindent we have
\begin{align*}
Amp(y)  &  =\frac{(a_{k}-b_{k})}{\sqrt{N}}\omega^{sy}\sum_{r=0}^{M-1}%
\omega^{rPy}+\frac{b_{k}}{\sqrt{N}}\sum_{z=0}^{N-1}\omega^{zy}\\
&  =\frac{(a_{k}-b_{k})}{\sqrt{N}}\omega^{sy}\sum_{r=0}^{M-1}\omega^{rPy}\\
&  =\frac{(a_{k}-b_{k})}{\sqrt{N}}\omega^{sy}M\\
&  =\frac{Mw^{sy}}{\sqrt{NM}}\sin(2k+1)\theta-\frac{Mw^{sy}}{\sqrt{N(N-M)}%
}\cos(2k+1)\theta\\
&  =\omega^{sy}\sqrt{\frac{M}{N}}(\sin(2k+1)\theta-\sqrt{\frac{M/N}{1-M/N}%
}\cos(2k+1)\theta)\\
&  =\omega^{sy}\sqrt{\frac{M}{N}}(\sin(2k+1)\theta-\frac{\sin\theta}%
{\cos\theta}\cos(2k+1)\theta)\\
&  =\omega^{sy}\tan\theta\sin2k\theta
\end{align*}

\noindent We have the probability $\Pr(y)$ in the case where
$Py=0\operatorname{mod}N,y\neq0$ is given by
\[
\Pr(y)=tan^{2}\theta\sin^{2}2k\theta
\]

\noindent Using $k=\left\lfloor \frac{\pi}{4\theta}\right\rfloor $
$\Longrightarrow$ $\frac{\pi}{4\theta}-1\leq k\leq\frac{\pi}{4\theta}$
$\Longrightarrow$ $\frac{\pi}{2}-2\theta\leq2k\theta\leq\frac{\pi}%
{2}\Longrightarrow\sin(\frac{\pi}{2}-2\theta)\leq\sin2k\theta\leq1$ we have
the following inequality for the probability $\Pr(y)$ in the case where
$Py=0\operatorname{mod}N,y\neq0$
\begin{align*}
\frac{\sin^{2}\theta}{\cos^{2}\theta}  &  \geq\Pr(y)=\tan^{2}\theta\sin
^{2}2k\theta\geq\tan^{2}\theta\sin^{2}(\frac{\pi}{2}-2\theta)\\
&  \Longrightarrow\frac{M}{N}\frac{1}{1-\frac{M}{N}}\geq\Pr(y)\geq\tan
^{2}\theta\sin^{2}(\frac{\pi}{2}-2\theta)\\
&  \Longrightarrow\frac{M}{N}(\frac{N}{N-M})\geq\Pr(y)\geq\frac{\sin^{2}%
\theta}{\cos^{2}\theta}\cos^{2}2\theta\\
&  \Longrightarrow\frac{M}{N}(\frac{N}{N-M})\geq\Pr(y)\geq\frac{\sin^{2}%
\theta}{\cos^{2}\theta}(2\cos^{2}\theta-1)^{2}\\
&  \Longrightarrow\frac{M}{N}(\frac{N}{N-M})\geq\Pr(y)\geq\frac{M}{N}(\frac
{N}{N-M})(1-\frac{2M}{N})^{2}%
\end{align*}

\subsection{Amplified-QFT Analysis: $Py\neq0\operatorname{mod}N$}

\noindent We calculate $\Pr(y)$ in the case where $Py\neq0\operatorname{mod}%
N.$

Making use of the previous results we have
\begin{align*}
Amp(y)  &  =\frac{(a_{k}-b_{k})}{\sqrt{N}}\omega^{sy}\sum_{r=0}^{M-1}%
\omega^{rPy}+\frac{b_{k}}{\sqrt{N}}\sum_{z=0}^{N-1}\omega^{zy}\\
&  =\frac{(a_{k}-b_{k})}{\sqrt{N}}\omega^{sy}\sum_{r=0}^{M-1}\omega^{rPy}\\
&  =\frac{(a_{k}-b_{k})}{\sqrt{N}}\omega^{sy}\left[  \frac{1-\omega^{MPy}%
}{1-\omega^{Py}}\right] \\
&  =\frac{1}{M}\frac{(a_{k}-b_{k})}{\sqrt{N}}\omega^{sy}M\left[
\frac{1-\omega^{MPy}}{1-\omega^{Py}}\right] \\
&  =\frac{1}{M}\omega^{sy}\tan\theta\sin2k\theta\left[  \frac{1-\omega^{MPy}%
}{1-\omega^{Py}}\right]
\end{align*}

\noindent\noindent\noindent Making use of the following identity%
\[
|1-e^{i\theta}|^{2}=4\sin^{2}(\theta/2)
\]

\noindent we have%
\[
\left\vert \frac{1-\omega^{MPy}}{1-\omega^{Py}}\right\vert ^{2}=\frac{\sin
^{2}(\pi MPy/N)}{\sin^{2}(\pi Py/N)}%
\]

\noindent and so the probability $\Pr(y)$ in the case where $Py\neq
0\operatorname{mod}N$ is given by
\[
\Pr(y)=\frac{1}{M^{2}}tan^{2}\theta\sin^{2}2k\theta\frac{\sin^{2}(\pi
MPy/N)}{\sin^{2}(\pi Py/N)}%
\]

\noindent Using the previous result $\frac{M}{N}(\frac{N}{N-M})\geq\tan
^{2}\theta\sin^{2}2k\theta\geq\frac{M}{N}(\frac{N}{N-M})(\frac{N-2M}{N})^{2}$
and letting $R=\frac{\sin^{2}(\pi MPy/N)}{\sin^{2}(\pi Py/N)}$ we have

\bigskip%
\begin{align*}
\frac{1}{M^{2}}\frac{M}{N}(\frac{N}{N-M})R  &  \geq\Pr(y)\geq\frac{1}{M^{2}%
}\frac{M}{N}(\frac{N}{N-M})(1-\frac{2M}{N})^{2}R\text{ and so}\\
\frac{1}{NM}(\frac{N}{N-M})R  &  \geq\Pr(y)\geq\frac{1}{NM}(\frac{N}%
{N-M})(1-\frac{2M}{N})^{2}R\text{ }%
\end{align*}

\noindent We notice that if in addition $MPy=0\operatorname{mod}N$ then
$\Pr(y)=0.$

\subsection{Amplified-QFT Summary}

The probability $\Pr(y)$ for the Amplified-QFT is summarized in the following
table and is given exactly by

\begin{center}%
\[
\left\{
\begin{tabular}
[c]{lll}%
$\cos^{2}2k\theta$ & if & $y=0$\\
&  & \\
$tan^{2}\theta\sin^{2}2k\theta$ & if & $Py=0\operatorname{mod}N,y\neq0$\\
&  & \\
$\frac{1}{M^{2}}tan^{2}\theta\sin^{2}2k\theta\frac{\sin^{2}(\pi MPy/N)}%
{\sin^{2}(\pi Py/N)}$ & if & $Py\neq0\operatorname{mod}N\text{and }%
MPy\neq0\operatorname{mod}N$\\
&  & \\
$0$ & if & $Py\neq0\operatorname{mod}N\text{ and }MPy=0\operatorname{mod}N$%
\end{tabular}
\ \ \ \ \ \right\}
\]

\end{center}

\section{The QFT Algorithm - Detailed Analysis.}

In this section we examine the QFT algorithm in detail and produce the results
for the probability of success that were summarized earlier in the paper. We
just apply the QFT to the binary oracle f, which is 1 on A and 0 elsewhere.

We begin with the following state%

\[
|\xi>=\frac{1}{\sqrt{N}}\sum_{z=0}^{N-1}|z>\otimes\frac{1}{\sqrt{2}}(|0>-|1>)
\]
\bigskip

\noindent and apply the unitary transform for f, $U_{f}$ ,\ to this state
which performs the following action:%

\[
U_{f}|z>|c>=|z>|c\oplus f(z)>
\]

\noindent to get the state $|\psi>$%

\begin{align*}
|\psi &  >=U_{f}\frac{1}{\sqrt{N}}\sum_{z=0}^{N-1}|z>\frac{1}{\sqrt{2}%
}(|0>-|1>)\\
&  =\frac{1}{\sqrt{N}}\left[  (-1)\sum_{z\in A}|z>+\sum_{z\notin A}|z>\right]
\frac{1}{\sqrt{2}}(|0>-|1>)\\
&  =\frac{1}{\sqrt{N}}\left[  (-2)\sum_{z\in A}|z>+\sum_{z=0}^{N-1}|z>\right]
\frac{1}{\sqrt{2}}(|0>-|1>)
\end{align*}

\noindent Next we apply the QFT to try to find the period P, dropping
$\frac{1}{\sqrt{2}}(|0>-|1>)$.

\noindent The QFT applies the following action:%
\[
|z>\rightarrow\frac{1}{\sqrt{N}}\sum_{y=0}^{N-1}\omega^{zy}|y>
\]

\noindent to get%
\[
|\phi>=\sum_{y=0}^{N-1}\left[  \frac{(-2)}{N}\sum_{z\in A}\omega^{zy}+\frac
{1}{N}\sum_{z=0}^{N-1}\omega^{zy}\right]  |y>
\]

We calculate the $\Pr(y)$ for the following cases:

\qquad a) $y=0$

\qquad b) $Py=0\operatorname{mod}N$ and $y\neq0$

\qquad c) $Py\neq0\operatorname{mod}N$

\subsection{QFT Analysis: $y=0$}

We calculate the probability $\Pr(y=0).$

We have%

\begin{align*}
Amp(y)  &  =\frac{(-2)}{N}\sum_{z\in A}\omega^{zy}+\frac{1}{N}\sum_{z=0}%
^{N-1}\omega^{zy}\\
&  =\frac{(-2)M}{N}+\frac{N}{N}\\
&  =1-\frac{2M}{N}%
\end{align*}

\noindent Therefore, in the QFT case, we have $\Pr(y=0)$ is very close to 1
and is given by
\[
\Pr(y=0)=1-\frac{4M}{N}+4\frac{M^{2}}{N^{2}}=\left(  1-\frac{2M}{N}\right)
^{2}%
\]

\noindent whereas in the Amplified-QFT case we have $\Pr(y=0)$ is given by
\[
\Pr(y=0)=\cos^{2}2k\theta
\]

\subsection{QFT Analysis: $Py=0\operatorname{mod}N,y\neq0$}

We calculate the probability $\Pr(y)$ where $Py=0\operatorname{mod}N,y\neq0.$

Using the fact that%
\[
\sum_{z=0}^{N-1}\omega^{zy}=\frac{1-\omega^{Ny}}{1-\omega^{y}}=0
\]

\noindent we have
\begin{align*}
Amp(y)  &  =\ \frac{-2}{N}\sum_{z\in A}\omega^{zy}+\frac{1}{N}\sum_{z=0}%
^{N-1}\omega^{zy}\\
&  =\frac{-2}{N}\omega^{sy}\sum_{r=0}^{M-1}\omega^{rPy}\\
&  =\frac{-2M}{N}\omega^{sy}%
\end{align*}

\noindent Therefore in the QFT$\ $case we have $\Pr(y)$ where
$Py=0\operatorname{mod}N,y\neq0$ is given by
\[
\Pr(y)=4\frac{M^{2}}{N^{2}}%
\]

\noindent which is small as $M<<N$, whereas in the Amplified-QFT case we have
$\Pr(y)$ is given by \
\[
\Pr(y)=tan^{2}\theta\sin^{2}2k\theta
\]

\noindent We can determine how the increase in amplitude varies with the
number of iterations $k$ of the Grover step in the Amplified-QFT by examining
the ratio of the amplitudes of the Amplified-QFT case and QFT\ case. This
ratio is given exactly by
\begin{align*}
AmpRatio(y)  &  =\frac{\frac{(a_{k}-b_{k})}{\sqrt{N}}\omega^{sy}M}{\frac
{-2M}{N}\omega^{sy}}\\
&  =\frac{(a_{k}-b_{k})}{-2}\sqrt{N}\\
&  =\frac{1}{-2}\left[  \sqrt{\frac{N}{M}}\sin(2k+1)\theta-\sqrt{\frac{N}%
{N-M}}\cos(2k+1)\theta\right] \\
&  =\frac{N}{-2M}\tan\theta\sin2k\theta
\end{align*}

\noindent We have the following for the probability ratio $\Pr Ratio(y),$ the
increase in probability due to amplification%

\[
\Pr Ratio(y)=\frac{N^{2}\tan^{2}\theta\sin^{2}2k\theta}{4M^{2}}%
\]

Using $k=\left\lfloor \frac{\pi}{4\theta}\right\rfloor $ and making use of
\[
\frac{M}{N}(\frac{N}{N-M})\geq\tan^{2}\theta\sin^{2}2k\theta\geq\frac{M}%
{N}(\frac{N}{N-M})(\frac{N-2M}{N})^{2}%
\]
we have the following inequality for the $\Pr Ratio(y)$:%

\begin{align*}
\frac{N}{4M}(\frac{N}{N-M})  &  \geq\Pr Ratio(y)\geq\frac{N}{4M}(\frac{N}%
{N-M})(1-\frac{2M}{N})^{2}\\
&  \Longrightarrow\Pr Ratio(y)\approx\frac{N}{4M}%
\end{align*}

\subsection{QFT Analysis: $Py\neq0\operatorname{mod}N$}

We calculate the probability $\Pr(y)$ in the case where $Py\neq
0\operatorname{mod}N.$

We have
\begin{align*}
Amp(y)  &  =\ \frac{-2}{N}\sum_{z\in A}\omega^{zy}+\frac{1}{N}\sum_{z=0}%
^{N-1}\omega^{zy}\\
&  =\frac{-2}{N}w^{sy}\sum_{r=0}^{M-1}\omega^{rPy}\\
&  =\frac{-2}{N}w^{sy}\left[  \frac{1-\omega^{MPy}}{1-\omega^{Py}}\right] \\
&  =\frac{-2}{N}w^{sy}\left[  \frac{1-\omega^{MPy}}{1-\omega^{Py}}\right]
\end{align*}

\noindent\noindent Once again, making use of the following identity%
\[
|1-e^{i\theta}|^{2}=4\sin^{2}(\theta/2)
\]

\noindent in the QFT\ case, we have $\Pr(y)$ where $Py\neq0\operatorname{mod}%
N.$is given by%

\[
\Pr(y)=\frac{4}{N^{2}}\left[  \frac{\sin^{2}(\pi MPy/N)}{\sin^{2}(\pi
Py/N)}\right]
\]

\noindent whereas in the Amplified-QFT case we have $\Pr(y)$ is given by
\[
\Pr(y)=\frac{1}{M^{2}}tan^{2}\theta\sin^{2}2k\theta\frac{\sin^{2}(\pi
MPy/N)}{\sin^{2}(\pi Py/N)}%
\]

We note that%

\[
0\leq\frac{\sin^{2}(\pi MPy/N)}{\sin^{2}(\pi Py/N)}\leq M^{2}%
\]

\noindent We notice that if in addition $MPy=0\operatorname{mod}N$ then
$\Pr(y)=0.$

\noindent The ratio of the amplitudes of the Amplified-QFT case and QFT\ case
is given exactly by
\begin{align*}
AmpRatio(y)  &  =\frac{\frac{(a_{k}-b_{k})}{\sqrt{N}}\omega^{sy}\left[
\frac{1-\omega^{MPy}}{1-\omega^{Py}}\right]  }{\frac{-2}{N}w^{sy}\left[
\frac{1-\omega^{MPy}}{1-\omega^{Py}}\right]  }\\
&  =\frac{(a_{k}-b_{k})}{-2}\sqrt{N}\\
&  =\frac{1}{-2}\left[  \sqrt{\frac{N}{M}}\sin(2k+1)\theta-\sqrt{\frac{N}%
{N-M}}\cos(2k+1)\theta\right] \\
&  =\frac{N}{-2M}\tan\theta\sin2k\theta
\end{align*}

\noindent We note that this ratio is the same as in that given in the previous
section and is independent of $y$. The variables in this ratio do not depend
in anyway on the QFT.

We have the following for the probability ratio $\Pr Ratio(y),$ the increase
in probability due to amplification%

\[
\Pr Ratio(y)=\frac{N^{2}\tan^{2}\theta\sin^{2}2k\theta}{4M^{2}}%
\]

As in the previous section, we have the following inequality for the $\Pr
Ratio(y)$, the increase in the probability due to amplification when
$k=\left\lfloor \frac{\pi}{4\theta}\right\rfloor $ and making use of $\frac
{M}{N}(\frac{N}{N-M})\geq\tan^{2}\theta\sin^{2}2k\theta\geq\frac{M}{N}%
(\frac{N}{N-M})(\frac{N-2M}{N})^{2}$%

\begin{align*}
\frac{N}{4M}(\frac{N}{N-M})  &  \geq\Pr Ratio(y)\geq\frac{N}{4M}(\frac{N}%
{N-M})(1-\frac{2M}{N})^{2}\\
&  \Longrightarrow\Pr Ratio(y)\approx\frac{N}{4M}%
\end{align*}

\subsection{QFT Summary}

The probability $\Pr(y)$ for the QFT is summarized in the following table and
is given exactly by%

\[
\left\{
\begin{tabular}
[c]{lll}%
$\left(  1-\frac{2M}{N}\right)  ^{2}$ & if & $y=0$\\
&  & \\
$4\frac{M^{2}}{N^{2}}$ & if & $Py=0\operatorname{mod}N,y\neq0$\\
&  & \\
$\frac{4}{N^{2}}\frac{\sin^{2}(\pi MPy/N)}{\sin^{2}(\pi Py/N)}$ & if &
$Py\neq0\operatorname{mod}N\text{and }MPy\neq0\operatorname{mod}N$\\
&  & \\
$0$ & if & $Py\neq0\operatorname{mod}N\text{ and }MPy=0\operatorname{mod}N$%
\end{tabular}
\ \ \ \ \right\}
\]

\section{The QHS Algorithm - Detailed Analysis}

In this section we examine the QHS algorithm in detail and produce the results
for the probability of success that were summarized earlier in the paper. The
QHS algorithm is a two register algorithm as follows (see ref[13] for
details). We begin with $|0>|0>$ where the first register is $n$ qubits and
the second register is $1$ qubit and apply the Hadamard transform to the first
register to get a uniform superposition state, followed by the unitary
transformation for the Oracle f to get:%

\[
|\psi>=\frac{1}{\sqrt{N}}%
{\displaystyle\sum\limits_{x=0}^{N-1}}
|x>|f(x)>
\]
Next we apply the QFT to the first register to get%

\begin{align*}
|\psi &  >=\frac{1}{\sqrt{N}}%
{\displaystyle\sum\limits_{x=0}^{N-1}}
\frac{1}{\sqrt{N}}\sum_{y=0}^{N-1}\omega^{xy}|y>|f(x)>\\
&  =%
{\displaystyle\sum\limits_{y=0}^{N-1}}
\frac{1}{N}\sum_{x=0}^{N-1}\omega^{xy}|y>|f(x)>\\
&  =%
{\displaystyle\sum\limits_{y=0}^{N-1}}
\frac{1}{N}|y>\sum_{x=0}^{N-1}\omega^{xy}|f(x)>\\
&  =\
{\displaystyle\sum\limits_{y=0}^{N-1}}
\frac{|||\Gamma(y)>||}{N}|y>\frac{|\Gamma(y)>}{|||\Gamma(y)>||}%
\end{align*}

\noindent where
\begin{align*}
|\Gamma(y)  &  >=\sum_{x=0}^{N-1}\omega^{xy}|f(x)>\\
&  =\sum_{x\in A}^{{}}\omega^{xy}|1>+\sum_{x\notin A}^{{}}\omega^{xy}|0>
\end{align*}

\noindent and where
\[
|||\Gamma(y)>||^{2}=\left\vert \sum_{x\in A}^{{}}\omega^{xy}\right\vert
^{2}+\left\vert \sum_{x\notin A}^{{}}\omega^{xy}\right\vert ^{2}%
\]

\noindent Next we make a measurement to get $y$ and find that the probability
of this measurement is%
\begin{align*}
\Pr(y)  &  =\frac{|||\Gamma(y)>||^{2}}{N^{2}}\\
&  =\frac{1}{N^{2}}\left\vert \sum_{x\in A}^{{}}\omega^{xy}\right\vert
^{2}+\frac{1}{N^{2}}\left\vert \sum_{x\notin A}^{{}}\omega^{xy}\right\vert
^{2}%
\end{align*}

\noindent The state that we end up in is of the form
\[
|\phi>=|y>\frac{|\Gamma(y)>}{|||\Gamma(y)>||}%
\]

We calculate the $\Pr(y)$ for the following cases:

\qquad a) $y=0$

\qquad b) $Py=0\operatorname{mod}N$ and $y\neq0$

\qquad c) $Py\neq0\operatorname{mod}N$

\subsection{QHS Analysis: $y=0$}

We calculate $\Pr(y=0)$

We have%
\begin{align*}
\Pr(y)  &  =\frac{1}{N^{2}}\left\vert \sum_{x\in A}^{{}}\omega^{xy}\right\vert
^{2}+\frac{1}{N^{2}}\left\vert \sum_{x\notin A}^{{}}\omega^{xy}\right\vert
^{2}\\
&  =\frac{M^{2}}{N^{2}}+\frac{(N-M)^{2}}{N^{2}}=\frac{M^{2}+N^{2}-2NM+M^{2}%
}{N^{2}}\\
&  =1-\frac{2M(N-M)}{N^{2}}%
\end{align*}

\noindent which is close to 1, whereas in the Amplified-QFT case we have
$\Pr(y=0)$ is given by
\[
\Pr(y=0)=\cos^{2}2k\theta
\]

\subsection{QHS Analysis: $Py=0\operatorname{mod}N,y\neq0$}

We calculate $\Pr(y)$ where $Py=0\operatorname{mod}N,y\neq0.$

We have%
\begin{align*}
\Pr(y)  &  =\frac{1}{N^{2}}\left\vert \sum_{x\in A}^{{}}\omega^{xy}\right\vert
^{2}+\frac{1}{N^{2}}\left\vert \sum_{x\notin A}^{{}}\omega^{xy}\right\vert
^{2}\\
&  =\frac{1}{N^{2}}\left\vert \omega^{sy}\sum_{r=0}^{M-1}\omega^{rPy}%
\right\vert ^{2}+\frac{1}{N^{2}}\left\vert \sum_{x\notin A}^{{}}\omega
^{xy}\right\vert ^{2}\\
&  =\frac{1}{N^{2}}\left\vert \omega^{sy}\sum_{r=0}^{M-1}\omega^{rPy}%
\right\vert ^{2}+\frac{1}{N^{2}}\left\vert -\omega^{sy}\sum_{r=0}^{M-1}%
\omega^{rPy}+\frac{1}{N}\sum_{x=0}^{N-1}\omega^{xy}\right\vert ^{2}\\
&  =\frac{2M^{2}}{N^{2}}%
\end{align*}

\noindent\noindent which is small because $M<<N$ and where we have used the
fact that
\[
\sum_{x=0}^{N-1}\omega^{xy}=0
\]

In the Amplified-QFT case we have $\Pr(y)$ is given by \
\[
\Pr(y)=tan^{2}\theta\sin^{2}2k\theta
\]

We have $\Pr Ratio(y)=\Pr(y)_{Amplified-QFT}/\Pr(y)_{QHS}$, the increase in
the probability due to amplification is given by%

\[
\frac{N^{2}\tan^{2}\theta\sin^{2}2k\theta}{2M^{2}}%
\]

We have the following inequality for the
\[
\Pr Ratio(y)=\Pr(y)_{Amplified-QFT}/\Pr(y)_{QHS}%
\]
where $k=\left\lfloor \frac{\pi}{4\theta}\right\rfloor $ and making use of
$\frac{M}{N}(\frac{N}{N-M})\geq\tan^{2}\theta\sin^{2}2k\theta\geq\frac{M}%
{N}(\frac{N}{N-M})(\frac{N-2M}{N})^{2}$%

\begin{align*}
\frac{N}{2M}(\frac{N}{N-M})  &  \geq\Pr Ratio(y)\geq\frac{N}{2M}(\frac{N}%
{N-M})(1-\frac{2M}{N})^{2}\\
&  \Longrightarrow\Pr Ratio(y)\approx\frac{N}{2M}%
\end{align*}

\subsection{QHS Analysis: $Py\neq0\operatorname{mod}N$}

We calculate $\Pr(y)$ where $Py\neq0\operatorname{mod}N.$

We have%
\begin{align*}
\Pr(y)  &  =\frac{1}{N^{2}}\left\vert \sum_{x\in A}^{{}}\omega^{xy}\right\vert
^{2}+\frac{1}{N^{2}}\left\vert \sum_{x\notin A}^{{}}\omega^{xy}\right\vert
^{2}\\
&  =\frac{1}{N^{2}}\left\vert \omega^{sy}\sum_{r=0}^{M-1}\omega^{rPy}%
\right\vert ^{2}+\frac{1}{N^{2}}\left\vert \sum_{x\notin A}^{{}}\omega
^{xy}\right\vert ^{2}\\
&  =\frac{1}{N^{2}}\left\vert \omega^{sy}\sum_{r=0}^{M-1}\omega^{rPy}%
\right\vert ^{2}+\frac{1}{N^{2}}\left\vert -\omega^{sy}\sum_{r=0}^{M-1}%
\omega^{rPy}+\frac{1}{N}\sum_{x=0}^{N-1}\omega^{xy}\right\vert ^{2}\\
&  =\frac{1}{N^{2}}\left\vert \omega^{sy}\left[  \frac{1-\omega^{MPy}%
}{1-\omega^{Py}}\right]  \right\vert ^{2}+\frac{1}{N^{2}}\left\vert
-\omega^{sy}\left[  \frac{1-\omega^{MPy}}{1-\omega^{Py}}\right]  \right\vert
^{2}\\
&  =\frac{2}{N^{2}}\frac{\sin^{2}(\pi MPy/N)}{\sin^{2}(\pi Py/N)}%
\end{align*}

\noindent\noindent where we have used the fact that
\[
\sum_{x=0}^{N-1}\omega^{xy}=0
\]

\noindent and that%

\[
|1-e^{i\theta}|^{2}=4\sin^{2}(\theta/2)
\]
In the Amplified-QFT case we have $\Pr(y)$ is given by
\[
\Pr(y)=\frac{1}{M^{2}}tan^{2}\theta\sin^{2}2k\theta\frac{\sin^{2}(\pi
MPy/N)}{\sin^{2}(\pi Py/N)}%
\]

We note that
\[
0\leq\frac{\sin^{2}(\pi MPy/N)}{\sin^{2}(\pi Py/N)}\leq M^{2}%
\]

We notice that if in addition $MPy=0\operatorname{mod}N$ then $\Pr(y)=0.$

We have the $\Pr Ratio(y)=\Pr(y)_{Amplified-QFT}/\Pr(y)_{QHS}$, the increase
in the probability due to amplification is given by%

\[
\Pr Ratio(y)=\frac{N^{2}\tan^{2}\theta\sin^{2}2k\theta}{2M^{2}}%
\]

We have the following inequality for the $\Pr Ratio(y)$ where $k=\left\lfloor
\frac{\pi}{4\theta}\right\rfloor $ and making use of $\frac{M}{N}(\frac
{N}{N-M})\geq\tan^{2}\theta\sin^{2}2k\theta\geq\frac{M}{N}(\frac{N}%
{N-M})(\frac{N-2M}{N})^{2}$%

\begin{align*}
\frac{N}{2M}(\frac{N}{N-M})  &  \geq\Pr Ratio(y)\geq\frac{N}{2M}(\frac{N}%
{N-M})(1-\frac{2M}{N})^{2}\\
&  \Longrightarrow\Pr Ratio(y)\approx\frac{N}{2M}%
\end{align*}

\subsection{QHS Summary}

We summarize the results for the QHS case. The $\Pr(y)$ is given exactly by:%

\[
\left\{
\begin{tabular}
[c]{lll}%
$1-\frac{2M(N-M)}{N^{2}}$ & if & $y=0$\\
&  & \\
$\frac{2M^{2}}{N^{2}}$ & if & $Py=0\operatorname{mod}N,y\neq0$\\
&  & \\
$\frac{2}{N^{2}}\frac{\sin^{2}(\pi MPy/N)}{\sin^{2}(\pi Py/N)}$ & if &
$Py\neq0\operatorname{mod}N\text{and }MPy\neq0\operatorname{mod}N$\\
&  & \\
$0$ & if & $Py\neq0\operatorname{mod}N\text{ and }MPy=0\operatorname{mod}N$%
\end{tabular}
\ \ \ \ \ \ \ \right\}
\]

\section{Recovering the Period P and the Offset s}

As in Shor's algorithm, we use the continued fraction expansion of $y/N$ to
find the period $P,$where $y$ is a measured value such that $y/N$ is close to
$d/P$ and $(d,P)=1$ . See ref[2] and ref[3]for details which we provide below.

Let$\{a\}_{N}$ be the residue of $a\operatorname{mod}N$ of smallest magnitude
such that $-N/2<\{a\}_{N}<N/2.$ Let $S_{N}=\{0,1,...,N-1\}$, $S_{P}=\{d\in
S_{N}:0\leq d<P\}$ and $Y=\{y\in S_{N}:|\{Py\}_{N}|\leq P/2\}$. Then the map
$Y\rightarrow S_{P}$ given by $y\rightarrow d=d(y)=round(Py/N)$ with inverse
$y=y(d)=round(Nd/P)$ is a bijection and $\{Py\}_{N}=Py-Nd(y)$. In addition the
following two sets are in 1-1 correspondence $\{y/N:y\in Y\}$ and $\{d/P:0\leq
d<P\}.$

We make use of the following theorem from the theory of continued fractions
ref[5] (Theorem 184 p.153):

\begin{theorem}
Let $x$ be a real number and let $a$ and $b$ be integers with $b>0$. If
$|x-\frac{a}{b}|\leq\frac{1}{2b^{2}}$ then the rational $a/b$ is a convergent
of the continued fraction expansion of $x$.
\end{theorem}

\begin{corollary}
If $P^{2}\leq N$ and $|\{Py\}_{N}|\leq\frac{P}{2}$ then $d(y)/P$ is a
convergent of the continued fraction expansion of $y/N$.
\end{corollary}

\begin{proof}
Since $\{Py\}_{N}=Py-Nd(y)$ we have $|Py-Nd(y)|\leq\frac{P}{2}$ or $|\frac
{y}{N}-\frac{d(y)}{P}|\leq\frac{1}{2N}\leq\frac{1}{2P^{2}}$ and we can apply
Theorem 1 so that $d/P$ is a convergent of the continued fraction expansion of
$y/N$.
\end{proof}

Since we know $y$ and $N$ we can find the continued fraction expansion of
$y/N$. However we also need that $(d,P)=1$ in order that $d/P$ is a convergent
and enabling us to read off $P$ directly. The probability that $(d,P)=1$ is
$\varphi(P)/P$ where $\varphi(P)$ is Euler's totient function. If $P$ is prime
we get $(d,P)=1$ trivially.

By making use of the following Theorem it can be shown that%
\[
\frac{\varphi(P)}{P}\geq\frac{e^{-\gamma}-\epsilon(P)}{\ln2}\frac{1}{\ln\ln N}%
\]
where $\epsilon(P)$ is a monotone decreasing sequence converging to zero.

\begin{theorem}
$\lim\inf\frac{\varphi(N)}{N/\ln\ln N}=e^{-\gamma}$
\end{theorem}

where $\gamma=0.57721566$ is Euler's constant and where $e^{-\gamma
}=0.5614594836$.

This may cause us to repeat the experiment $\Omega(\frac{1}{\ln\ln N})$ times
in order to get $(d,P)=1$.

We note that we needed to add a condition on the period $P$ that $P^{2}\leq N
$ \ or $P\leq\sqrt{N}$ in order for the proof of the corollary to work.

\subsection{Testing if $P_{1}=P$ when $s$ is known or is $0$}

We can easily test if $s=0$ by checking to see if $f(0)=1.$

Now given a putative value of the period $P_{1}$ and a known offset or shift
$s$, how can we test whether $P_{1}=P$ ?

Assuming we have access to the Oracle to test individual values, we can
confirm $f(s)=1$ since $s$ is known. We will show that if $f(s+P_{1})=1$ and
$f(s+(M-1)P_{1})=1$ then $P_{1}=P.$

Case 1: If $P_{1}>P$ then $s+(M-1)P_{1}>s+(M-1)P.$ But $s+(M-1)P$ is the
largest index $x$ such that $f(x)=1.$ Therefore if $P_{1}>P$ we must have
$f(s+(M-1)P_{1})=0.$

Case 2: If $0<P_{1}<P$ then $s<s+P_{1}<s+P$ but between $s$ and $P$ there are
no other values $x$ such that $f(x)=1.$Therefore if $0<P_{1}<P$ we must have
$f(s+P_{1})=0.$

Therefore if $f(s)=1,f(s+P_{1})=1$ and $f(s+(M-1)P_{1})=1$ we must have
$P_{1}=P.$

\subsection{Testing if $(s_{1},P_{1})=(s,P)$ when $s$ is from a small known
set and $s\neq0$}

If we assume $s$ is unknown and $s\neq0$ but is from a small known set of
possible values such that we can exhaust over this set on a classical computer
and we are given a putative value of the period $P_{1}$, how can we test
whether a pair of values $(s_{1},P_{1})$ is the correct pair $(s,P)$ ?

We need only test whether $f(s_{1})=1$, $f(s_{1}+P_{1})=1$ and $f(s_{1}%
+(M-1)P_{1})=1$ where M is assumed known.

Case 1: If $s_{1}<s$ then $f(s_{1})=0$ since $\ s$ is the smallest index $x$
with $\ f(x)=1.$

Case 2: If $s_{1}>s$ and $f(s_{1})=1$ then $s_{1}=s+rP$ with $r>0$ . If
$f(s_{1}+P_{1})=1$ then $s_{1}+P_{1}=s+tP=s_{1}+(t-r)P$ with $t>r>0.$ Hence
$P_{1}=(t-r)P>0.$ If $f(s_{1}+(M-1)P_{1})=1$ then $s_{1}+(M-1)P_{1}%
=s+rP+(M-1)(t-r)P>s+(M-1)P$ which is the largest index $x$ with $f(x)=1.$
Therefore $f(s_{1}+(M-1)P_{1})=0.$

Hence if $f(s_{1})=1$, $f(s_{1}+P_{1})=1$ and $f(s_{1}+(M-1)P_{1})=1$ we must
have $s_{1}=s$ and then by following the case when $s$ is known we must also
have $P_{1}=P.$

Therefore if one or more of the values $f(s_{1}),$ $f(s_{1}+P_{1}),$
$f(s_{1}+(M-1)P_{1})$ is zero, either $s_{1}$ or $P_{1}$ is wrong. For a given
$P_{1\text{ }}$we must exhaust over all possible values of $s$ before we can
be sure that $P_{1}\neq P.$ For in the case that $P_{1}\neq P,$ we will have
for every possible $s_{1}$ that at least one of the values $f(s_{1}),$
$f(s_{1}+P_{1}),$ $f(s_{1}+(M-1)P_{1})$ is zero. In such a case we must try
another putative $P_{1}.$

\subsection{Finding $s\neq0$ using a Quantum Computer}

We can assume $s\neq0$ as the case $s=0$ is trivial and was considered above.
Let $s=\alpha+\beta P$ where $\alpha=s\operatorname{mod}P$ so that
$0\leq\alpha\leq P-1$ and $0\leq\alpha+\beta P+(M-1)P\leq N-1.$

We assume we are given the correct value of $P_{.}$ If $P$ is wrong, it will
be detected in the algorithm.

Step 1:

We create an initial superposition on $N$ values%

\[
|\psi_{1}>=\frac{1}{\sqrt{N}}\sum_{x=0}^{N-1}|x>
\]

and apply the Oracle $f$ and put this into the amplitude. We then apply Grover
without measurement to amplify the amplitudes and we have the following state%

\[
|\psi_{1}>=a_{k}\sum_{\times\in A}|x>+b_{k}\sum_{x\notin A}|x>
\]

where%

\[
a_{k}=\frac{1}{\sqrt{M}}\sin(2k+1)\theta,b_{k}=\frac{1}{\sqrt{N-M}}%
\cos(2k+1)\theta
\]

are the appropriate amplitudes of the states and where
\[
\sin\theta=\sqrt{M/N},\cos\theta=\sqrt{1-M/N}%
\]

Next we measure the register and with probability exceeding $1-M/N$ we will
measure a value $x_{1}\in A$ where $x_{1}=s+r_{1}P$ with $0\leq r_{1}\leq
M-1.$ Note that the total probability of the set A is given by
\begin{align*}
\Pr(x  &  \in A)=M(\frac{1}{\sqrt{M}}\sin(2k+1)\theta)^{2}=\sin^{2}%
(2k+1)\theta=1-\cos^{2}(2k+1)\theta\\
&  \Longrightarrow\Pr(x\in A)=\sin^{2}(2k+1)\theta\geq1-\frac{M}{N}%
\end{align*}

Now using our measured value $x_{1}=s+r_{1}P$ with $0\leq r_{1}\leq M-1$ we
check that $f(x_{1})=1$ and $f(x_{1}-P)=1.$ If $f(x_{1}-P)=0$ then either the
value of $P$ we are using is wrong or we have $r_{1}=0$ and $x_{1}=s.$ If we
test $f(s)=1$, $f(s+P)=1$ and $f(s+(M-1)P)=1$ then we have the correct $P$ and
$s$ otherwise $P$ is wrong. So assuming $f(x_{1}-P)=1$ we must have either the
correct $P$ or a multiple of $P$. We can use the procedure in Step 2 or Step
2' to find $s.$ The method in Step 2 uses the Exact Quantum Counting algorithm
to find $s$ (See ref[11] for details). The method in Step 2' uses a method of
decreasing sequence of measurements to find $s.$

Step 2 (using the Exact Quantum Counting algorithm):

Let $T$ be such that $T\geq M$ is the smallest power of $2$ greater than $M$.
We form a superposition%

\[
|\varphi_{1}>=\frac{1}{\sqrt{T}}\sum_{x=0}^{T-1}|x>|0>
\]

\noindent and apply the function $g(x)=Max(0,x_{1}-(x+1)P)$ where
$x_{1}=s+r_{1}P$ is our measured value, with $0\leq r_{1}\leq M-1$and put the
values of $g(x)$ into the second register to get%

\[
|\varphi_{2}>=\frac{1}{\sqrt{T}}\sum_{x=0}^{T-1}|x>|g(x)>
\]

\noindent Notice that as $x$ increases from $0$, $g(x)$ is a decreasing
sequence $s+rP$ with $r=(r_{1}-x-1).$ When $g(x)$ dips below $0$ we set
$g(x)=0$ to ensure $g(x)\geq0.$ Now we apply $f$ to $g(x)$ and put the results
into the amplitude to get%

\[
|\varphi_{3}>=\frac{1}{\sqrt{T}}\sum_{x=0}^{T-1}(-1)^{f(g(x))}|x>|g(x)>
\]

Notice that $f(g(x))=1$ when $s\leq g(x)<s+r_{1}P$ and is $0$ elsewhere. We
apply the exact quantum counting algorithm which determines how many values
$f(g(x))=1.$Let this total be $R.$ If $P$ is correct we expect $R=r_{1}$ and
we can determine $s=x_{1}-RP=s+r_{1}P-RP.$ We can then test if we have the
correct pair of values $s,P$ by testing whether $f(s)=1$, $f(s+P)=1$ and
$f(s+(M-1)P)=1.$ If this test fails then $P$ must be an incorrect value and we
must repeat the period finding algorithm.

We use Theorem 8.3.4 of ref[11]: The Exact Quantum Counting algorithm requires
an expected number of applications of $U_{f}$ in $O(\sqrt{(R+1)(T-R+1)}$ and
outputs the correct value $R$ with probability at least $2/3.$

Step 2' (decreasing sequence of measurements method):

Let $T$ be such that $T\geq M$ is the smallest power of $2$ greater than $M$.
We form a superposition%

\[
|\varphi_{1}>=\frac{1}{\sqrt{T}}\sum_{x=0}^{T-1}|x>|0>
\]

and apply the function $g(x)=Max(0,x_{1}-(x+1)P)$ where $x_{1}=s+r_{1}P$ with
$0\leq r_{1}\leq M-1$and put these values into the second register to get%

\[
|\varphi_{2}>=\frac{1}{\sqrt{T}}\sum_{x=0}^{T-1}|x>|g(x)>
\]

Notice that as $x$ increases from $0$, $g(x)$ is a decreasing sequence $s+rP$
with $r=(r_{1}-x-1).$ When $g(x)$ dips below $0$ we set $g(x)=0$ to ensure
$g(x)\geq0.$ Now we apply $f$ to $g(x)$ and put the results into the third
register and then into the amplitude.%

\[
|\varphi_{3}>=\frac{1}{\sqrt{T}}\sum_{x=0}^{T-1}(-1)^{f(g(x))}|x>|g(x)>
\]

Notice that $f(g(x))=1$ when $s\leq g(x)<s+r_{1}P$ and is $0$ elsewhere.

We then run Grover without measurement to amplify the amplitudes and measure
the second register containing $g(x).$

With probability close to 1 we will measure a new value $x_{2}=s+r_{2}P$ with
$0\leq r_{2}<r_{1}.$ We test the values $f(x_{2})=1$ and $f(x_{2}-P)=1.$ If
$f(x_{2}-P)=0$ then either the value of $P$ we are using is wrong or we have
$r_{2}=0$ and $x_{2}=s.$ If we test $f(s)=1$, $f(s+P)=1$ and $f(s+(M-1)P)=1$
then we have the correct $P$ and $s$ otherwise $P$ is wrong. So assuming
$f(x_{2}-P)=1$ we must have either the correct $P$ or a multiple of $P$. We
repeat this algorithm and go to Step 2' replacing the value $x_{1} $ in the
function $g(x)$ with $x_{2}$ etc. As we repeat the algorithm we will measure a
decreasing sequence of values $x_{1},x_{2}...$ that converges to $s.$ This
procedure will eventually terminate with the correct pair of values $P$ and
$s$ or we will determine that we have been using an incorrect value of $P$ and
we must repeat the quantum algorithm for finding putative $P $ and repeat the process.

How many times do we expect to repeat Step 2'? When we make our first
measurement we expect $r_{1}=M/2.$ For our second measurement we expect
$r_{2}=r_{1}/2$ etc. Therefore we expect to repeat this algorithm $O($
$\ln_{2}(M))$ times.

\section{Replacing the QFT With a General Unitary Transform U}

In general, if we had any Oracle $f$ which is $1$ on a set of labels $A$ and
$0$ elsewhere and we replaced the QFT\ in the Amplified-QFT algorithm with any
unitary transform $U$ which performs the following%

\[
|z>\rightarrow\frac{1}{\sqrt{N}}\sum_{y=0}^{N-1}\alpha(y,z)|y>
\]

\noindent we can compute the $AmpRatio(y)\ =\frac{Amplified-Amplitude(U)}%
{Amplitude(U)}$as follows.

\noindent As before, we have the following state after applying $U_{f}$:%

\[
|\psi>=\frac{1}{\sqrt{N}}\left[  (-2)\sum_{z\in A}|z>+\sum_{z=0}%
^{N-1}|z>\right]
\]

\noindent Next we apply the general unitary transform $U$ to obtain the state%

\[
U|\psi>=\sum_{y=0}^{N-1}\left[  \frac{(-2)}{N}\sum_{z\in A}\alpha
(y,z)+\frac{1}{N}\sum_{z=0}^{N-1}\alpha(y,z)\right]  |y>
\]

\noindent In the Amplified-U case we apply Grover without measurement followed
by $U$ we obtain the state%

\[
|\phi_{k}>=\sum_{y=0}^{N-1}\left[  \frac{(a_{k}-b_{k})}{\sqrt{N}}\sum_{z\in
A}\alpha(y,z)+\frac{b_{k}}{\sqrt{N}}\sum_{z=0}^{N-1}\alpha(y,z)\right]  |y>
\]

\noindent If $\sum_{z=0}^{N-1}\alpha(y,z)=0$ and $\sum_{z\in A}\alpha
(y,z)\neq0$ we get the same $AmpRatio(y)$ formula that we obtained when
$U=QFT$%

\begin{align*}
AmpRatio(y)  &  =\frac{\ \frac{(a_{k}-b_{k})}{\sqrt{N}}\sum_{z\in A}%
\alpha(y,z)+\frac{b_{k}}{\sqrt{N}}\sum_{z=0}^{N-1}\alpha(y,z)}{\frac{(-2)}%
{N}\sum_{z\in A}\alpha(y,z)+\frac{1}{N}\sum_{z=0}^{N-1}\alpha(y,z)}\\
&  =\frac{\ \frac{(a_{k}-b_{k})}{\sqrt{N}}\sum_{z\in A}\alpha(y,z)}%
{\frac{(-2)}{N}\sum_{z\in A}\alpha(y,z)}\\
&  =\frac{\ \frac{(a_{k}-b_{k})}{\sqrt{N}}}{\frac{(-2)}{N}}\\
&  =\frac{(a_{k}-b_{k})}{-2}\sqrt{N}\\
&  =\frac{1}{-2}\left[  \sqrt{\frac{N}{M}}\sin(2k+1)\theta-\sqrt{\frac{N}%
{N-M}}\cos(2k+1)\theta\right] \\
&  =\frac{N}{-2M}\tan\theta\sin2k\theta
\end{align*}

\noindent This gives
\[
\Pr Ratio(y)=\frac{N^{2}}{4M^{2}}\tan^{2}\theta\sin^{2}2k\theta
\]

\bigskip As in the case when U=QFT, we have the following inequality for the
$\Pr Ratio(y)$ for a general U, the increase in the probability due to
amplification when $k=\left\lfloor \frac{\pi}{4\theta}\right\rfloor $ and
making use of $\frac{M}{N}(\frac{N}{N-M})\geq\tan^{2}\theta\sin^{2}%
2k\theta\geq\frac{M}{N}(\frac{N}{N-M})(\frac{N-2M}{N})^{2}$%

\begin{align*}
\frac{N}{4M}(\frac{N}{N-M})  &  \geq\Pr Ratio(y)\geq\frac{N}{4M}(\frac{N}%
{N-M})(1-\frac{2M}{N})^{2}\\
&  \Longrightarrow\Pr Ratio(y)\approx\frac{N}{4M}%
\end{align*}

\section{\bigskip General Amplification Procedure With General Oracle}

In this section we consider the case of a general amplification procedure with
a general oracle followed by a QFT. We produce a general upper bound on the
probability of measuring an observed value $y.$

Let $f$ be a general oracle which is $1$ on a set of labels $A$ and $0$
elsewhere. We assume there is a general amplification procedure which is
unknown, which produces the following general state:%

\[
|\psi>=\sum_{z=0}^{N-1}\sqrt{p_{z}}|z>
\]

where $p_{z}$ is a probability distribution produced by a general
amplification procedure. In addition we assume that%

\[
p(A)=\sum_{z\in A}p_{z}=\alpha\simeq1
\]

and%

\[
p(\overline{A})=\sum_{z\notin A}p_{z}=1-\alpha\simeq0
\]

Next we apply the QFT to the state $|\psi>$ by performing the following transformation%

\[
|z>\rightarrow\frac{1}{\sqrt{N}}\sum_{y=0}^{N-1}\omega^{zy}|y>
\]

which gives the state%

\[
|\varphi>=\sum_{y=0}^{N-1}\frac{1}{\sqrt{N}}\left[  \sum_{z\in A}\sqrt{p_{z}%
}\omega^{zy}+\sum_{z\notin A}\sqrt{p_{z}}\omega^{zy}\right]  |y>
\]

Next we wish to compute an upper bound on $P(y).$%

\begin{align*}
\Pr(y)  &  =\left\vert \frac{1}{\sqrt{N}}\sum_{z\in A}\sqrt{p_{z}}\omega
^{zy}+\frac{1}{\sqrt{N}}\sum_{z\notin A}\sqrt{p_{z}}\omega^{zy}\right\vert
^{2}\\
&  \leq\left\vert \frac{1}{\sqrt{N}}\sqrt{\left(  \sum_{z\in A}p_{z}\right)
\left(  \sum_{z\in A}\left\vert \omega^{zy}\right\vert ^{2}\right)  }+\frac
{1}{\sqrt{N}}\sqrt{\left(  \sum_{z\notin A}p_{z}\right)  \left(  \sum_{z\notin
A}\left\vert \omega^{zy}\right\vert ^{2}\right)  }\right\vert ^{2}\\
&  \text{ by the Cauchy-Schwarz inequality}\\
&  =\left(  \sqrt{\alpha\frac{M}{N}}+\sqrt{(1-\alpha)(1-\frac{M}{N})}\right)
^{2}%
\end{align*}
We see that in the specific case where the amplification procedure is perfect
and $\alpha=1$ the upper bound on $\Pr(y)$ is $\frac{M}{N}.$and otherwise, the
upper bound on $\Pr(y)$ is close to $\frac{M}{N}.$

Next we generalize this further to the case where we have a general unitary
transform $U$ with entries $\frac{1}{\sqrt{N}}\alpha(y,z)$ such that
$|\alpha(y,z)|^{2}=1$ for every $y$ and $z.$

Let $f$ be a general oracle which is $1$ on a set of labels $A$ and $0$
elsewhere. As before, we assume there is a general amplification procedure
which is unknown, which produces the following general state:%

\[
|\psi>=\sum_{x=0}^{N-1}\sqrt{p_{z}}|z>
\]

where $p_{z}$ is a probability distribution produced by a general
amplification procedure. In addition we assume that%

\[
p(A)=\sum_{z\in A}p_{z}=\alpha\simeq1
\]

and%

\[
p(\overline{A})=\sum_{z\notin A}p_{z}=1-\alpha\simeq0
\]
Next we apply $U$ to the state $|\psi>$ by performing the following transformation%

\[
|z>\rightarrow\frac{1}{\sqrt{N}}\sum_{y=0}^{N-1}\alpha(y,z)|y>
\]

which gives the state%

\[
|\varphi>=\sum_{y=0}^{N-1}\frac{1}{\sqrt{N}}\left[  \sum_{z\in A}\sqrt{p_{z}%
}\alpha(y,z)+\sum_{z\notin A}\sqrt{p_{z}}\alpha(y,z)\right]  |y>
\]

Next we wish to compute an upper bound on $P(y).$%

\begin{align*}
\Pr(y)  &  =\left\vert \frac{1}{\sqrt{N}}\sum_{z\in A}\sqrt{p_{z}}%
\alpha(y,z)+\frac{1}{\sqrt{N}}\sum_{z\notin A}\sqrt{p_{z}}\alpha
(y,z)\right\vert ^{2}\\
&  \leq\left\vert \frac{1}{\sqrt{N}}\sqrt{\left(  \sum_{z\in A}p_{z}\right)
\left(  \sum_{z\in A}\left\vert \alpha(y,z)\right\vert ^{2}\right)  }%
+\ \frac{1}{\sqrt{N}}\sqrt{\left(  \sum_{z\notin A}p_{z}\right)  \left(
\sum_{z\notin A}\left\vert \alpha(y,z)\right\vert ^{2}\right)  }\right\vert
^{2}\\
&  \text{ by the Cauchy-Schwarz inequality}\\
&  =\left(  \sqrt{\alpha\frac{M}{N}}+\sqrt{(1-\alpha)(1-\frac{M}{N})}\right)
^{2}%
\end{align*}

\newpage\renewcommand{\thechapter}{3}

\chapter{The Amplified Quantum Fourier Transform - With Error Stream}

\section{\textbf{Introduction}}

\bigskip In this paper, we generalize the results of the previous chapter
ref[14] and show how to use the A\textbf{mplified-QFT algorithm} to solve the
following problem:

\noindent\textbf{The Local Period Finding Problem, with Error Stream:} \ Let
$\mathcal{L}=\{0,1,...,N-1\}$ be a set of $N$ labels, and let $A$ be a
periodic subset of $M$ labels of period $P$, i.e., a subset of the form
\[
A=\left\{  j:j=s+rP,r=0,1,\ldots,M-1\right\}  \text{ ,}%
\]
where $P\leq\sqrt{N}$ and $M<<N$. Given a binary oracle
\[
h:\mathcal{L}\longrightarrow\left\{  0,1\right\}
\]
such that $h(x)=f(x)+g(x)$ where $+$ is the XOR operation and
\[
f,g:\mathcal{L}\longrightarrow\left\{  0,1\right\}
\]
and where $f(x)=1$ on $A$ and $0$ elsewhere and $g(x)$ is an Error Stream
which outputs a $1$ with Bernoulli probability $p$ and outputs a $0$ with
probability $q=1-p$. Let $G=\{x|g(x)=1\}$ with $|G|=L$ and let $C=A\cup G$ and
let $T=L+M$, with $\left\vert C\right\vert =T$. We assume $T$ is known because
if it is unknown, we can find it using the quantum counting algorithm.We
further assume that $A\cap G=\varnothing$ and note that $E[L]=Np$ and
$Var[L]=Npq$.

The Amplified-QFT algorithm which solves this problem consists of three steps.
\ \textbf{Step 1:} Apply Grover's algorithm without measurement to amplify the
amplitudes of the $T$ labels of the set $C$. \ \textbf{Step 2:} \ Apply the
QFT to the resulting state. \ \textbf{Step 3:} Measurement.

We compare the probabilities of success of three algorithms that can be used
to recover the period $P$: (1) Amplified-QFT (2) QFT and (3) QHS algorithms.
Let the set $S_{ALG}=\{y:|\frac{y}{N}-\frac{d}{P}|\leq\frac{1}{2P^{2}%
},(d,P)=1\}$ be the set of "successful" $y$'s. That is $S_{ALG\text{ }}%
$consists of those $y$'s which can be measured after applying one of the three
algorithms denoted by $ALG$ and from which the period P can be recovered by
the method of continued fractions. We show
\[
\frac{N}{4T}(\frac{N}{N-T})\geq\frac{\Pr(S_{Amplified-QFT})}{\Pr(S_{QFT})}%
\geq\frac{N}{4T}(\frac{N}{N-T})(1-\frac{2T}{N})^{2}%
\]
and%

\[
\frac{N}{2T}(\frac{N}{N-T})\geq\frac{\Pr(S_{Amplified-QFT})}{\Pr(S_{QHS})}%
\geq\frac{N}{2T}(\frac{N}{N-T})(1-\frac{2T}{N})^{2}%
\]

In the tables below, we summarize our results, comparing the probability of
measuring a $y$ in the final state arrived at after applying one of the three
algorithms- Amplified-QFT, QFT and QHS, where $\sin\theta=\sqrt{T/N}$ and
$k=\left\lfloor \frac{\pi}{4\theta}\right\rfloor $:

We compare each of the algorithms under the following four conditions on the
observation $y:$%

\begin{align*}
Case\text{ A }  &  \text{: }y=0\\
Case\text{ B }  &  \text{: }Py=0\operatorname{mod}N,y\neq0\\
Case\text{ C }  &  \text{: }Py\neq0\operatorname{mod}N\text{and }%
MPy\neq0\operatorname{mod}N\\
Case\text{ D }  &  \text{: }Py\neq0\operatorname{mod}N\text{ and
}MPy=0\operatorname{mod}N
\end{align*}

Case 1 (Amplified-QFT):

The probability $\Pr(y)$ is given exactly by

\begin{center}%
\[
\left\{
\begin{tabular}
[c]{l}%
Case A: $\cos^{2}2k\theta$\\
\\
Case B: $\tan^{2}\theta\sin^{2}2k\theta\left\vert \frac{M}{T}\omega^{sy}%
+\frac{1}{T}\sum_{z\in G}\omega^{zy}\right\vert ^{2}$\\
\\
Case C: $\tan^{2}\theta\sin^{2}2k\theta\left\vert \frac{1}{T}\omega
^{sy}\left[  \frac{1-\omega^{MPy}}{1-\omega^{Py}}\right]  +\frac{1}{T}%
\sum_{z\in G}\omega^{zy}\right\vert ^{2}$\\
\\
Case D: $\tan^{2}\theta\sin^{2}2k\theta\left\vert \frac{1}{T}\sum_{z\in
G}\omega^{zy}\right\vert ^{2}$%
\end{tabular}
\ \ \ \ \ \ \right\}
\]

\end{center}

Case 2 (QFT):

\bigskip The probability $\Pr(y)$ is given exactly by%

\[
\left\{
\begin{tabular}
[c]{l}%
Case A: $\left(  1-\frac{2T}{N}\right)  ^{2}$\\
\\
Case B: $\frac{4}{N^{2}}\left\vert \omega^{sy}M+\sum_{z\in G}\omega
^{zy}\right\vert ^{2}$\\
\\
Case C: $\frac{4}{N^{2}}\left\vert w^{sy}\left[  \frac{1-\omega^{MPy}%
}{1-\omega^{Py}}\right]  +\sum_{z\in G}\omega^{zy}\right\vert ^{2}$\\
\\
Case D: $\frac{4}{N^{2}}\left\vert \sum_{z\in G}\omega^{zy}\right\vert ^{2}$%
\end{tabular}
\ \ \ \ \ \ \ \right\}
\]

Let $y$ be fixed such that either Case B or Case C\ holds and define $\Pr
Ratio(y)=\Pr(y)_{Amplified-QFT}/\Pr(y)_{QFT}$ \ then we have the following%

\begin{align*}
\frac{N}{4T}(\frac{N}{N-T})  &  \geq\Pr Ratio(y)\geq\frac{N}{4T}(\frac{N}%
{N-T})(1-\frac{2T}{N})^{2}\\
&  \Longrightarrow\Pr Ratio(y)\approx\frac{N}{4T}%
\end{align*}

Case 3 (QHS):

\bigskip The probability $\Pr(y)$ is given exactly by%

\[
\left\{
\begin{tabular}
[c]{l}%
Case A: $1-\frac{2T(N-T)}{N^{2}}$\\
\\
Case B: $\frac{2}{N^{2}}\left\vert \omega^{sy}M+\sum_{z\in G}\omega
^{zy}\right\vert ^{2}$\\
\\
Case C: $\frac{2}{N^{2}}\left\vert w^{sy}\left[  \frac{1-\omega^{MPy}%
}{1-\omega^{Py}}\right]  +\sum_{z\in G}\omega^{zy}\right\vert ^{2}$\\
\\
Case D: $\frac{2}{N^{2}}\left\vert \sum_{z\in G}\omega^{zy}\right\vert ^{2}$%
\end{tabular}
\ \ \ \ \ \ \ \ \ \right\}
\]

Let $y$ be fixed such that either Case B or Case C\ holds and define $\Pr
Ratio(y)=\Pr(y)_{Amplified-QFT}/\Pr(y)_{QHS}$ \ then we have the following%

\begin{align*}
\frac{N}{2T}(\frac{N}{N-T})  &  \geq\Pr Ratio(y)\geq\frac{N}{2T}(\frac{N}%
{N-T})(1-\frac{2T}{N})^{2}\\
&  \Longrightarrow\Pr Ratio(y)\approx\frac{N}{2T}%
\end{align*}

Let $S_{ALG}=\{y:|\frac{y}{N}-\frac{d}{P}|\leq\frac{1}{2P^{2}},(d,P)=1\}$ be
the set of "successful" $y$'s. That is $S_{ALG\text{ }}$consists of those
$y$'s which can be measured after applying one of the three algorithms denoted
by $ALG$ and from which the period $P$ can be recovered by the method of
continued fractions. Note that the set $S_{ALG}$ is the same for each
algorithm. However the probability of this set varies with each algorithm. We
can see from the following that given $y1$ and $y2$, whose probability ratios
satisfy the same inequality, we can add their probabilities to get a new ratio
that satisfies the same inequality. In this way we can add probabilities over
a set on the numerator and denominator and maintain the inequality:
\begin{align*}
A  &  >\frac{P(y1)}{Q(y1)}>B\text{ }and\text{ }A>\frac{P(y2)}{Q(y2)}>B\\
&  \Longrightarrow A>\frac{P(y1)+P(y2)}{Q(y1)+Q(y2)}>B
\end{align*}

\noindent We see from the cases given above that\bigskip%
\[
\frac{N}{4T}(\frac{N}{N-T})\geq\frac{\Pr(S_{Amplified-QFT})}{\Pr(S_{QFT})}%
\geq\frac{N}{4T}(\frac{N}{N-T})(1-\frac{2T}{N})^{2}%
\]

\noindent where the difference between the upper bound and lower bound is
exactly 1 and that%

\[
\frac{N}{2T}(\frac{N}{N-T})\geq\frac{\Pr(S_{Amplified-QFT})}{\Pr(S_{QHS})}%
\geq\frac{N}{2T}(\frac{N}{N-T})(1-\frac{2T}{N})^{2}%
\]

\noindent where the difference between the upper bound and lower bound is
exactly 2.

\bigskip

This shows that the Amplified-QFT is approximately$\frac{N}{4T}$ times more
successful than the QFT and $\frac{N}{2T}$ times more successful than the QHS
when $T<<N$. In addition it also shows that the QFT\ is $2$ times more
successful than the QHS in this problem. However, the success of the
Amplified-QFT algorithms comes at an increase in work factor of $O(\sqrt
{\frac{N}{T}})$. We note that in the case that P is a prime number that
$(d,P)=1$ is met trivially. However when P is composite the algorithms may
need to be rerun several times until $(d,P)=1$ is satisfied.

\section{Comparison of Results Between L=0 and L%
$>$%
0}

In this section we compare the probabilities of making a measurement between
the two cases a)\ where there is no error stream and $L=0$ and b)\ where the
error stream is present and $L>0.$

We compare each of the algorithms under the following four conditions on the
observation $y:$%

\begin{align*}
Case\text{ A }  &  \text{: }y=0\\
Case\text{ B }  &  \text{: }Py=0\operatorname{mod}N,y\neq0\\
Case\text{ C }  &  \text{: }Py\neq0\operatorname{mod}N\text{and }%
MPy\neq0\operatorname{mod}N\\
Case\text{ D }  &  \text{: }Py\neq0\operatorname{mod}N\text{ and
}MPy=0\operatorname{mod}N
\end{align*}
\bigskip

1) Amplified-QFT\ case

The following are the probabilities of making a measurement $y$ when $L=0$:

Here $k_{1}=\left\lfloor \frac{\pi}{4\sin^{-1}(\sqrt{M/N})}\right\rfloor $ and
$\sin\theta_{1}=\sqrt{M/N}$

\begin{center}%
\[
\left\{
\begin{tabular}
[c]{l}%
Case A: $\cos^{2}2k_{1}\theta_{1}$\\
\\
Case B: $\tan^{2}\theta_{1}\sin^{2}2k_{1}\theta_{1}$\\
\\
Case C: $\frac{1}{M}\tan^{2}\theta_{1}\sin^{2}2k_{1}\theta_{1}\frac{\sin
^{2}(\pi MPy/N)}{\sin^{2}(\pi Py/N)}$\\
\\
Case D: $0$%
\end{tabular}
\ \ \ \ \ \ \ \right\}
\]

\end{center}

The following are the probabilities of making a measurement $y$ when $L>0$

Here $k_{2}=\left\lfloor \frac{\pi}{4\sin^{-1}(\sqrt{T/N})}\right\rfloor $ and
$\sin\theta_{2}=\sqrt{T/N}$%

\[
\left\{
\begin{tabular}
[c]{l}%
Case A: $\cos^{2}2k_{2}\theta_{2}$\\
\\
Case B: $\tan^{2}\theta_{2}\sin^{2}2k_{2}\theta_{2}\left\vert \frac{M}%
{T}\omega^{sy}+\frac{1}{T}\sum_{z\in G}\omega^{zy}\right\vert ^{2}$\\
\\
Case C: $\tan^{2}\theta_{2}\sin^{2}2k_{2}\theta_{2}\left\vert \frac{1}%
{T}\omega^{sy}\left[  \frac{1-\omega^{MPy}}{1-\omega^{Py}}\right]  +\frac
{1}{T}\sum_{z\in G}\omega^{zy}\right\vert ^{2}$\\
\\
Case D: $\tan^{2}\theta_{2}\sin^{2}2k_{2}\theta_{2}\left\vert \frac{1}{T}%
\sum_{z\in G}\omega^{zy}\right\vert ^{2}$%
\end{tabular}
\ \ \ \ \ \ \ \ \ \right\}
\]

We notice that in cases B, C and D with $L>0$ the probability of measuring $y$
now depends on $y$ whereas when $L=0$ it does not. In addition in cases B, C
and D with $L>0$ \ the probability depends upon a sum over a random set $G$:
$\frac{1}{T}\sum_{z\in G}\omega^{zy}.$

Notice that $T=L+M>M$ and $\sin\theta_{2}=\sqrt{T/N}>\sin\theta_{1}=\sqrt
{M/N}$ and $\theta_{2}>\theta_{1}.$

In the following theorems and corollaries we calculate the expected value and
variance of

Case B:%
\[
\left\vert \frac{M}{T}\omega^{sy}+\frac{1}{T}\sum_{z\in G}\omega
^{zy}\right\vert ^{2}%
\]

Case C:%
\[
\left\vert \frac{1}{T}\omega^{sy}\left[  \frac{1-\omega^{MPy}}{1-\omega^{Py}%
}\right]  +\frac{1}{T}\sum_{z\in G}\omega^{zy}\right\vert ^{2}%
\]

Case D:%
\[
\left\vert \frac{1}{T}\sum_{z\in G}\omega^{zy}\right\vert ^{2}%
\]

The results are summarized in the tables below. Let $\varphi=2\pi MPy/N$ and
let $\theta=2\pi Py/N.$

Then the expected values are:%

\[
\left\{
\begin{tabular}
[c]{l}%
Case B: $(M^{2}+L)/(L+M)^{2}$\\
\\
Case C: $\left(  \frac{\sin^{2}\frac{\varphi}{2}}{\sin^{2}\frac{\theta}{2}%
}+L\right)  /(L+M)^{2}$\\
\\
Case D: $L/(L+M)^{2}$%
\end{tabular}
\ \ \ \ \ \ \ \ \ \ \right\}
\]

and the variances are:%

\[
\left\{
\begin{tabular}
[c]{l}%
Case B: $(L^{2}-L+2M^{2}L)/(L+M)^{4}$\\
\\
Case C: $\left(  L^{2}-L+2\frac{\sin^{2}\frac{\varphi}{2}}{\sin^{2}%
\frac{\theta}{2}}L\right)  /(L+M)^{4}$\\
\\
Case D: $(L^{2}-L)/(L+M)^{4}$%
\end{tabular}
\ \ \ \ \ \ \ \ \ \ \right\}
\]

First we show%
\[
E\left\vert \sum_{z\in G}\omega^{zy}\right\vert ^{2}=L
\]
where the sum is taken over $L$ uniform random variables which take values in
$\{0,1,...,N-1\}$ and the expectation is computed over all such possible sums.

\begin{theorem}
Let $G=\{z_{1},z_{2},...,z_{L}\}$ be a set of $L$ uniform random variables
which take values in $\{0,1,...,N-1\}$. Consider the expected value of
$\left\vert \sum_{z_{j}\in G}\omega^{z_{j}y}\right\vert ^{2}$where the
expectation is taken over all $N^{L}$ sums then%
\[
E_{z}\left\vert \sum_{z_{j}\in G}\omega^{z_{j}y}\right\vert ^{2}=\frac
{1}{N^{L}}\sum_{G=\{z_{1},z_{2}...z_{L}\}}\left\vert \sum_{z_{j}\in G}%
\omega^{z_{j}y}\right\vert ^{2}=L
\]

\begin{proof}
Let $\omega^{zy}=\cos(\theta_{z})+i\sin(\theta_{z})$ where $\theta_{z}=2\pi
zy/N$ then%
\begin{align*}
&  E_{z}\left\vert \sum_{z\in G}\omega^{zy}\right\vert ^{2}\\
&  =\frac{1}{N^{L}}\sum_{G=\{z_{1},z_{2}...z_{L}\}}\left\vert \sum_{z_{j}\in
G}\omega^{z_{j}y}\right\vert ^{2}\\
&  =\frac{1}{N^{L}}\sum_{G=\{z_{1},z_{2}...z_{L}\}}\left\vert \sum_{z_{j}\in
G}\cos(\theta_{z_{j}})+i\sin(\theta_{z_{j}})\right\vert ^{2}\\
&  =\frac{1}{N^{L}}\sum_{G=\{z_{1},z_{2}...z_{L}\}}((\sum_{z_{j}\in G}%
\cos(\theta_{z_{j}}))^{2}+(\sum_{z_{j}\in G}\sin(\theta_{z_{j}}))^{2})\\
&  =\frac{1}{N^{L}}\sum_{G=\{z_{1},z_{2}...z_{L}\}}(\sum_{x_{j}\in G}%
\sum_{y_{j}\in G}\cos(\theta_{x_{j}})\cos(\theta_{y_{j}})+\sum_{x_{j}\in
G}\sum_{y_{j}\in G}\sin(\theta_{x_{j}})\sin(\theta_{y_{j}}))\\
&  =\frac{1}{N^{L}}\sum_{G=\{z_{1},z_{2}...z_{L}\}}\sum_{z_{j}\in G}\cos
^{2}(\theta_{z_{j}})+\sin^{2}(\theta_{z_{j}})+other\text{ }terms\\
&  =\frac{1}{N^{L}}\sum_{G=\{z_{1},z_{2}...z_{L}\}}\sum_{z_{j}\in G}1\\
&  =\frac{1}{N^{L}}\sum_{G=\{z_{1},z_{2}...z_{L}\}}L\\
&  =L
\end{align*}
where we have used the fact that
\[
\frac{1}{N^{L}}\sum_{G=\{z_{1},z_{2}...z_{L}\}}\sum_{x_{j}\in G}\sum_{y_{j}\in
G}\cos(\theta_{x_{j}})\cos(\theta_{y_{j}})=0\text{ }%
\]
where we sum over the variables independently for both cosines. In addition%
\[
\frac{1}{N^{L}}\sum_{G=\{z_{1},z_{2}...z_{L}\}}\sum_{x_{j}\in G}\sum_{y_{j}\in
G}\sin(\theta_{x_{j}})\sin(\theta_{y_{j}})=0\text{ }%
\]
where we sum over the variables independently for each sin. Another way to see
this is to consider the integral approximations of these sums:%
\begin{align*}
&  \frac{1}{2\pi}\int_{\theta1=0}^{\theta1=2\pi}\frac{1}{2\pi}\int_{\theta
2=0}^{\theta2=2\pi}\cos(\theta1)\cos(\theta2)d\theta2d\theta1\\
&  =\frac{1}{2\pi}\int_{\theta1=0}^{\theta1=2\pi}\cos(\theta1)\frac{1}{2\pi
}\int_{\theta2=0}^{\theta2=2\pi}\cos(\theta2)d\theta2d\theta1\\
&  =\frac{1}{2\pi}\int_{\theta1=0}^{\theta1=2\pi}\cos(\theta1)0d\theta1=0
\end{align*}
and%
\begin{align*}
\frac{1}{2\pi}  &  \int_{\theta1=0}^{\theta1=2\pi}\frac{1}{2\pi}\int
_{\theta2=0}^{\theta2=2\pi}\sin(\theta1)\sin(\theta2)d\theta2d\theta1\\
&  =\frac{1}{2\pi}\int_{\theta1=0}^{\theta1=2\pi}\sin(\theta1)\frac{1}{2\pi
}\int_{\theta2=0}^{\theta2=2\pi}\sin(\theta2)d\theta2d\theta1\\
&  =\frac{1}{2\pi}\int_{\theta1=0}^{\theta1=2\pi}\sin(\theta1)0d\theta1=0
\end{align*}

\end{proof}
\end{theorem}

\begin{theorem}
Let $G=\{z_{1},z_{2},...,z_{L}\}$ be a set of $L$ uniform random variables
which take values in $\{0,1,...,N-1\}$. Consider the expected value of
\[
\left\vert a\omega^{sy}+b\sum_{z_{j}\in G}\omega^{z_{j}y}\right\vert ^{2}%
\]
where the expectation is taken over all $N^{L}$ sums then%
\begin{align*}
&  E_{z}\left\vert a\omega^{sy}+b\sum_{z_{j}\in G}\omega^{z_{j}y}\right\vert
^{2}\\
&  =\frac{1}{N^{L}}\sum_{G=\{z_{1},z_{2}...z_{L}\}}\left\vert a\omega
^{sy}+b\sum_{z_{j}\in G}\omega^{z_{j}y}\right\vert ^{2}\\
&  =a^{2}+b^{2}L
\end{align*}

\begin{proof}
Let $\omega^{zy}=\cos(\theta_{z})+i\sin(\theta_{z})$ where $\theta_{z}=2\pi
zy/N$ and $\omega^{sy}=\cos(\varphi)+i\sin(\varphi)$ where $\varphi=2\pi sy/N$
then%
\begin{align*}
&  E_{z}\left\vert a\omega^{sy}+b\sum_{z_{j}\in G}\omega^{z_{j}y}\right\vert
^{2}\\
&  =\frac{1}{N^{L}}\sum_{G=\{z_{1},z_{2}...z_{L}\}}\left\vert a\omega
^{sy}+b\sum_{z_{j}\in G}\omega^{z_{j}y}\right\vert ^{2}\\
&  =\frac{1}{N^{L}}\sum_{G=\{z_{1},z_{2}...z_{L}\}}\left\vert a\cos
(\varphi)+ai\sin(\varphi)+b\sum_{z_{j}\in G}\cos(\theta_{z_{j}})+i\sin
(\theta_{z_{j}})\right\vert ^{2}\\
&  =\frac{1}{N^{L}}\sum_{G=\{z_{1},z_{2}...z_{L}\}}(a\cos(\varphi
)+b\sum_{z_{j}\in G}\cos(\theta_{z_{j}}))^{2}+(a\sin(\varphi)+b\sum_{z_{j}\in
G}\sin(\theta_{z_{j}}))^{2}\\
&  =\frac{1}{N^{L}}\sum_{G=\{z_{1},z_{2}...z_{L}\}}(a\cos(\varphi)+\sum
_{x_{j}\in G}b\cos(\theta_{x_{j}}))(a\cos(\varphi)+\sum_{y_{j}\in G}%
b\cos(\theta_{y_{j}}))\\
&  +(a\sin(\varphi)+b\sum_{x_{j}\in G}\sin(\theta_{x_{j}}))(a\sin
(\varphi)+b\sum_{y_{j}\in G}\sin(\theta_{y_{j}}))\\
&  =\frac{1}{N^{L}}\sum_{G=\{z_{1},z_{2}...z_{L}\}}(a^{2}\cos^{2}%
(\varphi)+b^{2}\sum_{x_{j}\in G}\sum_{y_{j}\in G}\cos(\theta_{x_{j}}%
)\cos(\theta_{y_{j}})+\\
&  2ab\cos(\varphi)\sum_{y_{j}\in G}\cos(\theta_{y_{j}}))\\
&  +(a^{2}\sin^{2}(\varphi)+b^{2}\sum_{x_{j}\in G}\sum_{y_{j}\in G}\sin
(\theta_{x_{j}})\sin(\theta_{y_{j}})+\\
&  2ab\sin(\varphi)\sum_{y_{j}\in G}\sin(\theta_{y_{j}}))\\
&  =\frac{1}{N^{L}}\sum_{G=\{z_{1},z_{2}...z_{L}\}}(a^{2}\cos^{2}%
(\varphi)+a^{2}\sin^{2}(\varphi)+\\
&  b^{2}\sum_{x_{j}\in G}\sum_{y_{j}\in G}\cos(\theta_{x_{j}})\cos
(\theta_{y_{j}})+b^{2}\sum_{x_{j}\in G}\sum_{y_{j}\in G}\sin(\theta_{x_{j}%
})\sin(\theta_{y_{j}}))\\
&  =\frac{1}{N^{L}}\sum_{G=\{z_{1},z_{2}...z_{L}\}}(a^{2}+b^{2}\sum_{x_{j}\in
G}\sum_{y_{j}\in G}(\cos(\theta_{x_{j}})\cos(\theta_{y_{j}})+\sin
(\theta_{x_{j}})\sin(\theta_{y_{j}})))\\
&  =a^{2}+b^{2}L\\
&  \text{ from the previous theorem.}%
\end{align*}

\end{proof}
\end{theorem}

\begin{corollary}
Let $G=\{z_{1},z_{2},...,z_{L}\}$ be a set of $L$ uniform random variables
which take values in $\{0,1,...,N-1\}$. Consider the expected value of
$\left\vert \frac{1}{T}\sum_{z_{j}\in G}\omega^{z_{j}y}\right\vert ^{2}$ and
$\left\vert \frac{M}{T}\omega^{sy}+\frac{1}{T}\sum_{z_{j}\in G}\omega^{z_{j}%
y}\right\vert ^{2}$where the expectation is taken over all $N^{L}$ sums then%
\begin{align*}
&  E_{z}\left\vert \frac{1}{T}\sum_{z_{j}\in G}\omega^{z_{j}y}\right\vert
^{2}\\
&  =L/T^{2}\\
&  =L/(L+M)^{2}\\
&  \rightarrow0\text{ as }L\rightarrow\infty
\end{align*}

\end{corollary}

and%
\begin{align*}
&  E_{z}\left\vert \frac{M}{T}\omega^{sy}+\frac{1}{T}\sum_{z_{j}\in G}%
\omega^{z_{j}y}\right\vert ^{2}\\
&  =(M^{2}+L)/T^{2}\\
&  =(M^{2}+L)/(L+M)^{2}\\
&  \rightarrow0\text{ as }L\rightarrow\infty
\end{align*}

\begin{theorem}
Let $G=\{z_{1},z_{2},...,z_{L}\}$ be a set of $L$ uniform random variables
which take values in $\{0,1,...,N-1\}$. Consider the expected value of
\[
\left\vert a+ic+b\sum_{z_{j}\in G}\omega^{z_{j}y}\right\vert ^{2}%
\]
where the expectation is taken over all $N^{L}$ sums then%
\begin{align*}
&  E_{z}\left\vert a+ic+b\sum_{z_{j}\in G}\omega^{z_{j}y}\right\vert ^{2}\\
&  =\frac{1}{N^{L}}\sum_{G=\{z_{1},z_{2}...z_{L}\}}\left\vert a+ic+b\sum
_{z_{j}\in G}\omega^{z_{j}y}\right\vert ^{2}\\
&  =a^{2}+c^{2}+b^{2}L
\end{align*}

\end{theorem}

\begin{proof}%
\begin{align*}
&  E_{z}\left\vert a+ic+b\sum_{z_{j}\in G}\omega^{z_{j}y}\right\vert ^{2}\\
&  =\frac{1}{N^{L}}\sum_{G=\{z_{1},z_{2}...z_{L}\}}\left\vert a+ic+b\sum
_{z_{j}\in G}\omega^{z_{j}y}\right\vert ^{2}\\
&  =\frac{1}{N^{L}}\sum_{G=\{z_{1},z_{2}...z_{L}\}}\left\vert a+ic+b\sum
_{z_{j}\in G}\cos(\theta_{z_{j}})+i\sin(\theta_{z_{j}})\right\vert ^{2}\\
&  =\frac{1}{N^{L}}\sum_{G=\{z_{1},z_{2}...z_{L}\}}(a+b\sum_{z_{j}\in G}%
\cos(\theta_{z_{j}}))^{2}+(c+b\sum_{z_{j}\in G}\sin(\theta_{z_{j}}))^{2}\\
&  =\frac{1}{N^{L}}\sum_{G=\{z_{1},z_{2}...z_{L}\}}(a+\sum_{x_{j}\in G}%
b\cos(\theta_{x_{j}}))(a+\sum_{y_{j}\in G}b\cos(\theta_{y_{j}}))\\
&  +(c+b\sum_{x_{j}\in G}\sin(\theta_{x_{j}}))(c+b\sum_{y_{j}\in G}\sin
(\theta_{y_{j}}))\\
&  =\frac{1}{N^{L}}\sum_{G=\{z_{1},z_{2}...z_{L}\}}(a^{2}+b^{2}\sum_{x_{j}\in
G}\sum_{y_{j}\in G}\cos(\theta_{x_{j}})\cos(\theta_{y_{j}})+\\
&  2ab\sum_{y_{j}\in G}\cos(\theta_{y_{j}}))\\
&  +(c^{2}+b^{2}\sum_{x_{j}\in G}\sum_{y_{j}\in G}\sin(\theta_{x_{j}}%
)\sin(\theta_{y_{j}})+\\
&  2cb\sum_{y_{j}\in G}\sin(\theta_{y_{j}}))\\
&  =\frac{1}{N^{L}}\sum_{G=\{z_{1},z_{2}...z_{L}\}}(a^{2}+c^{2}+\\
&  b^{2}\sum_{x_{j}\in G}\sum_{y_{j}\in G}\cos(\theta_{x_{j}})\cos
(\theta_{y_{j}})+b^{2}\sum_{x_{j}\in G}\sum_{y_{j}\in G}\sin(\theta_{x_{j}%
})\sin(\theta_{y_{j}}))\\
&  =\frac{1}{N^{L}}\sum_{G=\{z_{1},z_{2}...z_{L}\}}(a^{2}+b^{2}\sum_{x_{j}\in
G}\sum_{y_{j}\in G}(\cos(\theta_{x_{j}})\cos(\theta_{y_{j}})+\sin
(\theta_{x_{j}})\sin(\theta_{y_{j}})))\\
&  =a^{2}+c^{2}+b^{2}L\\
&  \text{ from the previous theorem.}%
\end{align*}

\end{proof}

\begin{corollary}
Let $\varphi=2\pi MPy/N$ and let $\theta=2\pi Py/N$ then%
\begin{align*}
&  E\left\vert \frac{1}{T}\omega^{sy}\left[  \frac{1-\omega^{MPy}}%
{1-\omega^{Py}}\right]  +\frac{1}{T}\sum_{z\in G}\omega^{zy}\right\vert ^{2}\\
&  =\frac{1}{T^{2}}\frac{\sin^{2}\frac{\varphi}{2}}{\sin^{2}\frac{\theta}{2}%
}+\frac{1}{T^{2}}L\\
&  \leq\frac{M^{2}+L}{T^{2}}\\
&  =\frac{M^{2}+L}{(M+L)^{2}}\\
&  \rightarrow0\text{ as }L\rightarrow\infty
\end{align*}

\end{corollary}

\begin{proof}
Let $\varphi=2\pi MPy/N$ and let $\theta=2\pi Py/N$ then%
\begin{align*}
\frac{1-\omega^{MPy}}{1-\omega^{Py}}  &  =\frac{1-\cos\varphi-i\sin\varphi
}{1-\cos\theta-i\sin\theta}\\
&  =\frac{[(1-\cos\varphi)-i\sin\varphi][(1-\cos\theta)+i\sin\theta]}%
{[(1-\cos\theta)-i\sin\theta][(1-\cos\theta)+i\sin\theta]}\\
&  =\frac{[2\sin^{2}\frac{\varphi}{2}-i\sin\varphi][2\sin^{2}\frac{\theta}%
{2}+i\sin\theta]}{2-2\cos\theta}\\
&  =\frac{4\sin^{2}\frac{\varphi}{2}\sin^{2}\frac{\theta}{2}-i2\sin\varphi
\sin^{2}\frac{\theta}{2}+i2\sin\theta\sin^{2}\frac{\varphi}{2}+\sin\varphi
\sin\theta}{4\sin^{2}\frac{\theta}{2}}\\
&  =\frac{%
\begin{array}
[c]{c}%
4\sin^{2}\frac{\varphi}{2}\sin^{2}\frac{\theta}{2}-i4\sin\frac{\varphi}{2}%
\cos\frac{\varphi}{2}\sin^{2}\frac{\theta}{2}+\\
i4\sin\frac{\theta}{2}\cos\frac{\theta}{2}\sin^{2}\frac{\varphi}{2}+4\sin
\frac{\theta}{2}\cos\frac{\theta}{2}\sin\frac{\varphi}{2}\cos\frac{\varphi}{2}%
\end{array}
}{4\sin^{2}\frac{\theta}{2}}\\
&  =\frac{\sin\frac{\varphi}{2}\sin\frac{\theta}{2}\cos(\frac{\varphi}%
{2}-\frac{\theta}{2})+i\sin\frac{\varphi}{2}\sin\frac{\theta}{2}\sin
(\frac{\varphi}{2}-\frac{\theta}{2})}{\sin^{2}\frac{\theta}{2}}\\
&  =\frac{\sin\frac{\varphi}{2}[\cos(\frac{\varphi}{2}-\frac{\theta}{2}%
)+i\sin(\frac{\varphi}{2}-\frac{\theta}{2})]}{\sin\frac{\theta}{2}}%
\end{align*}
Let $\lambda=2\pi sy/N$ then%
\begin{align*}
&  \frac{1}{T}(\cos\lambda+i\sin\lambda)\frac{1-\omega^{MPy}}{1-\omega^{Py}}\\
&  =\frac{1}{T}(\cos\lambda+i\sin\lambda)\frac{\sin\frac{\varphi}{2}%
[\cos(\frac{\varphi}{2}-\frac{\theta}{2})+i\sin(\frac{\varphi}{2}-\frac
{\theta}{2})]}{\sin\frac{\theta}{2}}\\
&  =\frac{1}{T}\frac{\sin\frac{\varphi}{2}[\cos\lambda\cos(\frac{\varphi}%
{2}-\frac{\theta}{2})-\sin\lambda\sin(\frac{\varphi}{2}-\frac{\theta}{2}%
)+}{\sin\frac{\theta}{2}}\\
&  \frac{i[\cos\lambda\sin(\frac{\varphi}{2}-\frac{\theta}{2})+\sin\lambda
\cos(\frac{\varphi}{2}-\frac{\theta}{2})]]}{\sin\frac{\theta}{2}}\\
&  =a+ic
\end{align*}
then
\begin{align*}
&  a^{2}+c^{2}\\
&  =\frac{1}{T^{2}}\frac{\sin^{2}\frac{\varphi}{2}}{\sin^{2}\frac{\theta}{2}%
}.[\cos^{2}\lambda+\sin^{2}\lambda]\\
&  =\frac{1}{T^{2}}\frac{\sin^{2}\frac{\varphi}{2}}{\sin^{2}\frac{\theta}{2}}%
\end{align*}

\end{proof}

\begin{theorem}
Let $G=\{z_{1},z_{2},...,z_{L}\}$ be a set of $L$ uniform random variables
which take values in $\{0,1,...,N-1\}$. Consider the variance of $\left\vert
\sum_{z_{j}\in G}\omega^{z_{j}y}\right\vert ^{2}$then%
\begin{align*}
Var\left\vert \sum_{z_{j}\in G}\omega^{z_{j}y}\right\vert ^{2}  &  =\frac
{1}{N^{L}}\sum_{G=\{z_{1},z_{2}...z_{L}\}}\left\vert \sum_{z_{j}\in G}%
\omega^{z_{j}y}\right\vert ^{4}-L^{2}\\
&  =2L^{2}-L-L^{2}\\
&  =L^{2}-L
\end{align*}

\end{theorem}

\begin{proof}
Note that
\[
Var(X)=E(X^{2})-(E(X))^{2}%
\]
Let $\omega^{zy}=\cos(\theta_{z})+i\sin(\theta_{z})$ where $\theta_{z}=2\pi
zy/N$ and consider%
\begin{align*}
&  \frac{1}{N^{L}}\sum_{G=\{z_{1},z_{2}...z_{L}\}}\left\vert \sum_{z_{j}\in
G}\omega^{z_{j}y}\right\vert ^{4}\\
&  =\frac{1}{N^{L}}\sum_{G=\{z_{1},z_{2}...z_{L}\}}\left\vert \sum_{z_{j}\in
G}\cos(\theta_{z_{j}})+i\sin(\theta_{z_{j}})\right\vert ^{4}\\
&  =\frac{1}{N^{L}}\sum_{G=\{z_{1},z_{2}...z_{L}\}}((\sum_{z_{j}\in G}%
\cos(\theta_{z_{j}}))^{2}+(\sum_{z_{j}\in G}\sin(\theta_{z_{j}}))^{2})^{2}\\
&  =\frac{1}{N^{L}}\sum_{G=\{z_{1},z_{2}...z_{L}\}}(\sum_{x_{j}\in G}%
\sum_{y_{j}\in G}\cos(\theta_{x_{j}})\cos(\theta_{y_{j}})+\sum_{x_{j}\in
G}\sum_{y_{j}\in G}\sin(\theta_{x_{j}})\sin(\theta_{y_{j}}))^{2}\\
&  =\frac{1}{N^{L}}\sum_{G=\{z_{1},z_{2}...z_{L}\}}(\sum_{u_{j}\in G}%
\sum_{v_{j}\in G}\sum_{x_{j}\in G}\sum_{y_{j}\in G}\cos(\theta_{u_{j}}%
)\cos(\theta_{v_{j}})\cos(\theta_{x_{j}})\cos(\theta_{y_{j}})+\\
&  2\sum_{u_{j}\in G}\sum_{v_{j}\in G}\sum_{x_{j}\in G}\sum_{y_{j}\in G}%
\cos(\theta_{u_{j}})\cos(\theta_{v_{j}})\sin(\theta_{x_{j}})\sin(\theta
_{y_{j}})+\\
&  \sum_{u_{j}\in G}\sum_{v_{j}\in G}\sum_{x_{j}\in G}\sum_{y_{j}\in G}%
\sin(\theta_{u_{j}})\sin(\theta_{v_{j}})\sin(\theta_{x_{j}})\sin(\theta
_{y_{j}}))\\
&  =\frac{1}{N^{L}}\sum_{G=\{z_{1},z_{2}...z_{L}\}}(\sum_{z_{j}\in G}\cos
^{4}(\theta_{z_{j}})+\\
&  \frac{1}{2}\left(
\begin{array}
[c]{c}%
4\\
2
\end{array}
\right)  \sum_{x_{j}\in G}\sum_{\substack{y_{j}\in G\\x_{j}\neq y_{j}}%
}\cos^{2}(\theta_{x_{j}})\cos^{2}(\theta_{y_{j}})+\\
&  2\sum_{x_{j}\in G}\sum_{\substack{y_{j}\in G\\x_{j}\neq y_{j}}}\cos
^{2}(\theta_{x_{j}})\sin^{2}(\theta_{y_{j}})+\\
&  2\sum_{z_{j}\in G}\cos^{2}(\theta_{z_{j}})\sin^{2}(\theta_{z_{j}})+\\
&  \sum_{z_{j}\in G}\sin^{4}(\theta_{z_{j}})+\\
&  \frac{1}{2}\left(
\begin{array}
[c]{c}%
4\\
2
\end{array}
\right)  \sum_{x_{j}\in G}\sum_{\substack{y_{j}\in G\\x_{j}\neq y_{j}}%
}\sin^{2}(\theta_{x_{j}})\sin^{2}(\theta_{y_{j}})\\
&  +otherterms)\\
&  =3L/8+\frac{6}{2}L(L-1)/4+2L(L-1)/4+2L/8+3L/8+\frac{6}{2}L(L-1)/4\\
&  =L+2L(L-1)=2L^{2}-L
\end{align*}
where we have used the fact that the sum over the $otherterms$ is $0,$ and we
have approximated the averages by integrals%
\[
\frac{1}{2\pi}\int_{\theta=0}^{\theta=2\pi}\cos^{2}(\theta)d\theta=\frac
{1}{2\pi}\int_{\theta=0}^{\theta=2\pi}\sin^{2}(\theta)d\theta=1/2
\]
$^{{}}$%
\[
\frac{1}{2\pi}\int_{\theta=0}^{\theta=2\pi}\cos^{4}(\theta)d\theta=\frac
{1}{2\pi}\int_{\theta=0}^{\theta=2\pi}\sin^{4}(\theta)d\theta=3/8
\]%
\[
\frac{1}{2\pi}\int_{\theta=0}^{\theta=2\pi}\cos^{2}(\theta)\sin^{2}%
(\theta)d\theta=1/8
\]

\end{proof}

\begin{theorem}
Let $G=\{z_{1},z_{2},...,z_{L}\}$ be a set of $L$ uniform random variables
which take values in $\{0,1,...,N-1\}$. Consider the variance of $\left\vert
a\omega^{sy}+b\sum_{z_{j}\in G}\omega^{z_{j}y}\right\vert ^{2}$then%
\begin{align*}
&  Var\left\vert a\omega^{sy}+b\sum_{z_{j}\in G}\omega^{z_{j}y}\right\vert
^{2}\\
&  =\frac{1}{N^{L}}\sum_{G=\{z_{1},z_{2}...z_{L}\}}\left\vert a\omega
^{sy}+b\sum_{z_{j}\in G}\omega^{z_{j}y}\right\vert ^{4}-(a^{2}+b^{2}L)^{2}\\
&  =a^{4}+b^{4}(2L^{2}-L)+4a^{2}b^{2}L-(a^{2}+b^{2}L)^{2}\\
&  =a^{4}+b^{4}(2L^{2}-L)+4a^{2}b^{2}L-a^{4}-2a^{2}b^{2}L-b^{4}L^{2}\\
&  =b^{4}(L^{2}-L)+2a^{2}b^{2}L
\end{align*}

\end{theorem}

\begin{proof}
Note that
\[
Var(X)=E(X^{2})-(E(X))^{2}%
\]
Let $\omega^{zy}=\cos(\theta_{z})+i\sin(\theta_{z})$ where $\theta_{z}=2\pi
zy/N$ and $\omega^{sy}=\cos(\varphi)+i\sin(\varphi)$ where $\varphi=2\pi sy/N$
then%
\begin{align*}
&  E\left\vert a\omega^{sy}+b\sum_{z_{j}\in G}\omega^{z_{j}y}\right\vert
^{4}\\
&  =\frac{1}{N^{L}}\sum_{G=\{z_{1},z_{2}...z_{L}\}}\left\vert a\omega
^{sy}+b\sum_{z_{j}\in G}\omega^{z_{j}y}\right\vert ^{4}\\
&  =\frac{1}{N^{L}}\sum_{G=\{z_{1},z_{2}...z_{L}\}}\left\vert a(\cos
(\varphi)+i\sin(\varphi))+b\sum_{z_{j}\in G}\cos(\theta_{z_{j}})+i\sin
(\theta_{z_{j}})\right\vert ^{4}\\
&  =\frac{1}{N^{L}}\sum_{G=\{z_{1},z_{2}...z_{L}\}}[(a\cos(\varphi
)+b\sum_{z_{j}\in G}\cos(\theta_{z_{j}}))^{2}+(a\sin(\varphi)+b\sum_{z_{j}\in
G}\sin(\theta_{z_{j}}))^{2}]^{2}\\
&  =\frac{1}{N^{L}}\sum_{G=\{z_{1},z_{2}...z_{L}\}}[(a\cos(\varphi
)+b\sum_{x_{j}\in G}\cos(\theta_{x_{j}}))(a\cos(\varphi)+b\sum_{y_{j}\in
G}\cos(\theta_{y_{j}}))+\\
&  (a\sin(\varphi)+b\sum_{x_{j}\in G}\sin(\theta_{x_{j}}))(a\sin
(\varphi)+b\sum_{y_{j}\in G}\sin(\theta_{y_{j}}))]^{2}\\
&  =\frac{1}{N^{L}}\sum_{G=\{z_{1},z_{2}...z_{L}\}}[a^{2}\cos^{2}%
(\varphi)+b^{2}\sum_{x_{j}\in G}\sum_{y_{j}\in G}\cos(\theta_{x_{j}}%
)\cos(\theta_{y_{j}})+\\
&  2ab\cos(\varphi)\sum_{x_{j}\in G}\cos(\theta_{x_{j}})+a^{2}\sin^{2}%
(\varphi)+b^{2}\sum_{x_{j}\in G}\sum_{y_{j}\in G}\sin(\theta_{x_{j}}%
)\sin(\theta_{y_{j}})+\\
&  2ab\sin(\varphi)\sum_{x_{j}\in G}\sin(\theta_{x_{j}})]^{2}\\
&  =\frac{1}{N^{L}}\sum_{G=\{z_{1},z_{2}...z_{L}\}}[a^{2}+b^{2}\sum_{x_{j}\in
G}\sum_{y_{j}\in G}\cos(\theta_{x_{j}})\cos(\theta_{y_{j}})+\\
&  b^{2}\sum_{x_{j}\in G}\sum_{y_{j}\in G}\sin(\theta_{x_{j}})\sin
(\theta_{y_{j}})+2ab\cos(\varphi)\sum_{x_{j}\in G}\cos(\theta_{x_{j}})+\\
&  2ab\sin(\varphi)\sum_{x_{j}\in G}\sin(\theta_{x_{j}})]^{2}\\
&  =\frac{1}{N^{L}}\sum_{G=\{z_{1},z_{2}...z_{L}\}}[a^{4}+b^{4}(\sum_{x_{j}\in
G}\sum_{y_{j}\in G}\cos(\theta_{x_{j}})\cos(\theta_{y_{j}})+\\
&  \sum_{x_{j}\in G}\sum_{y_{j}\in G}\sin(\theta_{x_{j}})\sin(\theta_{y_{j}%
}))^{2}+4a^{2}b^{2}\cos^{2}(\varphi)\sum_{x_{j}\in G}\sum_{y_{j}\in G}%
\cos(\theta_{x_{j}})\cos(\theta_{y_{j}})+\\
&  4a^{2}b^{2}\sin^{2}(\varphi)\sum_{x_{j}\in G}\sum_{y_{j}\in G}\sin
(\theta_{x_{j}})\sin(\theta_{y_{j}})+2a^{2}b^{2}(\sum_{x_{j}\in G}\sum
_{y_{j}\in G}\cos(\theta_{x_{j}})\cos(\theta_{y_{j}})+\\
&  \sum_{x_{j}\in G}\sum_{y_{j}\in G}\sin(\theta_{x_{j}})\sin(\theta_{y_{j}%
}))+otherterms]\\
&  =a^{4}+b^{4}(2L^{2}-L)+4a^{2}b^{2}\cos^{2}(\varphi)L/2+4a^{2}b^{2}\sin
^{2}(\varphi)L/2+2a^{2}b^{2}L\\
&  =a^{4}+b^{4}(2L^{2}-L)+4a^{2}b^{2}L
\end{align*}
Then%
\begin{align*}
&  Var\left\vert a\omega^{sy}+b\sum_{z_{j}\in G}\omega^{z_{j}y}\right\vert
^{2}=E\left\vert a\omega^{sy}+b\sum_{z_{j}\in G}\omega^{z_{j}y}\right\vert
^{4}-(a^{2}+b^{2}L)^{2}\\
&  =a^{4}+b^{4}(2L^{2}-L)+4a^{2}b^{2}L-(a^{2}+b^{2}L)^{2}\\
&  =a^{4}+b^{4}(2L^{2}-L)+4a^{2}b^{2}L-a^{4}-2a^{2}b^{2}L-b^{4}L^{2}\\
&  =b^{4}(L^{2}-L)+2a^{2}b^{2}L
\end{align*}

\end{proof}

\begin{corollary}
Let $G=\{z_{1},z_{2},...,z_{L}\}$ be a set of $L$ uniform random variables
which take values in $\{0,1,...,N-1\}$. Consider the variance of
\[
\left\vert \frac{1}{T}\sum_{z_{j}\in G}\omega^{z_{j}y}\right\vert ^{2}%
\]
and
\[
\left\vert \frac{M}{T}\omega^{sy}+\frac{1}{T}\sum_{z_{j}\in G}\omega^{z_{j}%
y}\right\vert ^{2}%
\]
then
\end{corollary}

\begin{align*}
&  Var_{z}\left\vert \frac{1}{T}\sum_{z_{j}\in G}\omega^{z_{j}y}\right\vert
^{2}\\
&  =\frac{(L^{2}-L)}{T^{4}}\\
&  =\frac{(L^{2}-L)}{(L+M)^{4}}\\
&  \rightarrow0\text{ as }L\rightarrow\infty
\end{align*}

and%

\begin{align*}
&  Var_{z}\left\vert \frac{M}{T}\omega^{sy}+\frac{1}{T}\sum_{z_{j}\in G}%
\omega^{z_{j}y}\right\vert ^{2}\\
&  =\frac{(L^{2}-L)}{T^{4}}+\frac{2M^{2}L}{T^{4}}\\
&  =\frac{(L^{2}-L)+2M^{2}L}{(L+M)^{4}}\\
&  \rightarrow0\text{ as }L\rightarrow\infty
\end{align*}

\begin{theorem}
Let $G=\{z_{1},z_{2},...,z_{L}\}$ be a set of $L$ uniform random variables
which take values in $\{0,1,...,N-1\}$. Consider the variance of $\left\vert
a+ic+b\sum_{z_{j}\in G}\omega^{z_{j}y}\right\vert ^{2}$then%
\begin{align*}
&  Var\left\vert a+ic+b\sum_{z_{j}\in G}\omega^{z_{j}y}\right\vert ^{2}\\
&  =\frac{1}{N^{L}}\sum_{G=\{z_{1},z_{2}...z_{L}\}}\left\vert a+ic+b\sum
_{z_{j}\in G}\omega^{z_{j}y}\right\vert ^{4}-(a^{2}+c^{2}+b^{2}L)^{2}\\
&  =a^{4}+c^{4}+b^{4}(2L^{2}-L)+4(a^{2}+c^{2})b^{2}L-(a^{4}+c^{4}+b^{4}%
L^{2}+2(a^{2}+c^{2})b^{2}L)\\
&  =b^{4}(L^{2}-L)+2(a^{2}+c^{2})b^{2}L
\end{align*}

\end{theorem}

\begin{proof}
Let $\omega^{zy}=\cos(\theta_{z})+i\sin(\theta_{z})$ where $\theta_{z}=2\pi
zy/N$ and $\omega^{sy}=\cos(\varphi)+i\sin(\varphi)$ where $\varphi=2\pi sy/N$
then%
\begin{align*}
&  E\left\vert a+ic+b\sum_{z_{j}\in G}\omega^{z_{j}y}\right\vert ^{4}\\
&  =\frac{1}{N^{L}}\sum_{G=\{z_{1},z_{2}...z_{L}\}}\left\vert a+ic+b\sum
_{z_{j}\in G}\omega^{z_{j}y}\right\vert ^{4}\\
&  =\frac{1}{N^{L}}\sum_{G=\{z_{1},z_{2}...z_{L}\}}\left\vert a+ic+b\sum
_{z_{j}\in G}\cos(\theta_{z_{j}})+i\sin(\theta_{z_{j}})\right\vert ^{4}\\
&  =\frac{1}{N^{L}}\sum_{G=\{z_{1},z_{2}...z_{L}\}}[(a+b\sum_{z_{j}\in G}%
\cos(\theta_{z_{j}}))^{2}+(c+b\sum_{z_{j}\in G}\sin(\theta_{z_{j}}))^{2}%
]^{2}\\
&  =\frac{1}{N^{L}}\sum_{G=\{z_{1},z_{2}...z_{L}\}}[(a+b\sum_{x_{j}\in G}%
\cos(\theta_{x_{j}}))(a+b\sum_{y_{j}\in G}\cos(\theta_{y_{j}}))+\\
&  (c+b\sum_{x_{j}\in G}\sin(\theta_{x_{j}}))(c+b\sum_{y_{j}\in G}\sin
(\theta_{y_{j}}))]^{2}\\
&  =\frac{1}{N^{L}}\sum_{G=\{z_{1},z_{2}...z_{L}\}}[a^{2}+b^{2}\sum_{x_{j}\in
G}\sum_{y_{j}\in G}\cos(\theta_{x_{j}})\cos(\theta_{y_{j}})+\\
&  2ab\sum_{x_{j}\in G}\cos(\theta_{x_{j}})+c^{2}+b^{2}\sum_{x_{j}\in G}%
\sum_{y_{j}\in G}\sin(\theta_{x_{j}})\sin(\theta_{y_{j}})+\\
&  2cb\sum_{x_{j}\in G}\sin(\theta_{x_{j}})]^{2}\\
&  =\frac{1}{N^{L}}\sum_{G=\{z_{1},z_{2}...z_{L}\}}[a^{2}+c^{2}+b^{2}%
\sum_{x_{j}\in G}\sum_{y_{j}\in G}\cos(\theta_{x_{j}})\cos(\theta_{y_{j}})+\\
&  b^{2}\sum_{x_{j}\in G}\sum_{y_{j}\in G}\sin(\theta_{x_{j}})\sin
(\theta_{y_{j}})+2ab\sum_{x_{j}\in G}\cos(\theta_{x_{j}})+\\
&  2cb\sum_{x_{j}\in G}\sin(\theta_{x_{j}})]^{2}\\
&  =\frac{1}{N^{L}}\sum_{G=\{z_{1},z_{2}...z_{L}\}}[a^{4}+c^{4}+b^{4}%
(\sum_{x_{j}\in G}\sum_{y_{j}\in G}\cos(\theta_{x_{j}})\cos(\theta_{y_{j}})+\\
&  \sum_{x_{j}\in G}\sum_{y_{j}\in G}\sin(\theta_{x_{j}})\sin(\theta_{y_{j}%
}))^{2}+4a^{2}b^{2}\sum_{x_{j}\in G}\sum_{y_{j}\in G}\cos(\theta_{x_{j}}%
)\cos(\theta_{y_{j}})+\\
&  4c^{2}b^{2}\sum_{x_{j}\in G}\sum_{y_{j}\in G}\sin(\theta_{x_{j}}%
)\sin(\theta_{y_{j}})+2(a^{2}+c^{2})b^{2}(\sum_{x_{j}\in G}\sum_{y_{j}\in
G}\cos(\theta_{x_{j}})\cos(\theta_{y_{j}})+\\
&  \sum_{x_{j}\in G}\sum_{y_{j}\in G}\sin(\theta_{x_{j}})\sin(\theta_{y_{j}%
}))+otherterms]\\
&  =a^{4}+c^{4}+b^{4}(2L^{2}-L)+4a^{2}b^{2}L/2+4c^{2}b^{2}L/2+2(a^{2}%
+c^{2})b^{2}L\\
&  =a^{4}+c^{4}+b^{4}(2L^{2}-L)+4(a^{2}+c^{2})b^{2}L
\end{align*}

\end{proof}

\begin{corollary}
Let $\varphi=2\pi MPy/N$ and let $\theta=2\pi Py/N$
\begin{align*}
&  Var\left\vert \frac{1}{T}\omega^{sy}\left[  \frac{1-\omega^{MPy}}%
{1-\omega^{Py}}\right]  +\frac{1}{T}\sum_{z\in G}\omega^{zy}\right\vert ^{2}\\
&  =\frac{1}{T^{4}}(L^{2}-L)+\frac{2}{T^{4}}\frac{\sin^{2}\frac{\varphi}{2}%
}{\sin^{2}\frac{\theta}{2}}L\\
&  \leq\frac{(L^{2}-L+2M^{2}L)}{(L+M)^{4}}\\
&  \rightarrow0\text{ as }L\rightarrow\infty
\end{align*}

\end{corollary}

\begin{proof}
Let $\varphi=2\pi MPy/N$ , let $\theta=2\pi Py/N$ and let $\lambda=2\pi sy/N$
then $b=\frac{1}{T}$ and from and earlier result
\begin{align*}
&  a^{2}+c^{2}\\
&  =\frac{1}{T^{2}}\frac{\sin^{2}\frac{\varphi}{2}}{\sin^{2}\frac{\theta}{2}%
}.[\cos^{2}\lambda+\sin^{2}\lambda]\\
&  =\frac{1}{T^{2}}\frac{\sin^{2}\frac{\varphi}{2}}{\sin^{2}\frac{\theta}{2}%
}\\
&  \leq\frac{M^{2}}{T^{2}}%
\end{align*}

\end{proof}

Next we show that as $L$ increases the upperbound of the expected value of the
probability in Cases B and C decreases to a minimum value.

\begin{lemma}
Let $k=\left\lfloor \frac{\pi}{4\sin^{-1}(\sqrt{T/N})}\right\rfloor $ and
$\sin\theta=\sqrt{T/N}$ where $T=L+M$ then the upperbound of the expected
value of the probability in Cases B and C decreases to a minimum value at
$MinL$ given by%
\[
MinL=-M^{2}+\sqrt{M(M-1)(M(M-1)+N)}%
\]

\end{lemma}

\begin{proof}
The probability in case B is given by%
\[
\tan^{2}\theta\sin^{2}2k\theta\left\vert \frac{M}{T}\omega^{sy}+\frac{1}%
{T}\sum_{z_{j}\in G}\omega^{z_{j}y}\right\vert ^{2}%
\]
This has expected value%
\[
\frac{M^{2}+L}{(L+M)^{2}}\tan^{2}\theta\sin^{2}2k\theta
\]
The probability in case C is given by
\[
\tan^{2}\theta\sin^{2}2k\theta\left\vert \frac{1}{T}\omega^{sy}\left[
\frac{1-\omega^{MPy}}{1-\omega^{Py}}\right]  +\frac{1}{T}\sum_{z\in G}%
\omega^{zy}\right\vert ^{2}%
\]
This has expected value%
\begin{align*}
&  (\frac{1}{T^{2}}\frac{\sin^{2}\frac{\varphi}{2}}{\sin^{2}\frac{\theta}{2}%
}+\frac{1}{T^{2}}L)\tan^{2}\theta\sin^{2}2k\theta\\
&  \leq\frac{M^{2}+L}{(M+L)^{2}}\tan^{2}\theta\sin^{2}2k\theta
\end{align*}
Therefore the expected value of the probability in case B dominates the
expected value in case C. We have by a later lemma that
\[
\frac{T}{N-T}\geq\tan^{2}\theta\sin^{2}2k\theta\geq(\frac{T}{N-T})(1-2T/N)^{2}%
\]
Therefore we have%
\begin{align*}
(\frac{L+M}{N-(L+M)})\frac{M^{2}+L}{(L+M)^{2}}  &  \geq\frac{M^{2}%
+L}{(L+M)^{2}}\tan^{2}\theta\sin^{2}2k\theta\\
\frac{1}{N-(L+M)}\frac{M^{2}+L}{(L+M)}  &  \geq\frac{M^{2}+L}{(L+M)^{2}}%
\tan^{2}\theta\sin^{2}2k\theta
\end{align*}
The upperbound of the expected value of the probability in Case B. decreases
to a minium value of $L.$ If we differentiate it and set it to zero we can
find the value of $L$ where it is a minimum. Consider
\begin{align*}
&  \frac{d}{dL}\frac{1}{N-(L+M)}\frac{M^{2}+L}{(L+M)}\\
&  =\frac{1}{(N-(L+M))^{2}}\frac{M^{2}+L}{(L+M)}+\frac{1}{(N-(L+M))(L+M)}\\
&  -\frac{M^{2}+L}{(N-(L+M))(L+M)^{2}}\\
&  =\frac{L^{2}+2M^{2}L+M(2M^{2}+N-M-NM)}{(N-(L+M))^{2}(L+M)^{2}}\text{ after
some rearranging}%
\end{align*}
Therefore $MinL$ is a solution to the quadratic equation
\[
L^{2}+2M^{2}L+M(2M^{2}+N-M-NM)=0
\]
This has solutions%
\begin{align*}
MinL  &  =\frac{1}{2}(-2M^{2}\pm\sqrt{4M^{4}-4M(2M^{2}+N-M-NM)})\\
&  =-M^{2}\pm\sqrt{M(M-1)(M(M-1)+N)}%
\end{align*}
Therefore the upper bound of the expected value of the probability achieves
its minimum at $MinL$ given by%
\[
MinL=-M^{2}+\sqrt{M(M-1)(M(M-1)+N)}%
\]

\end{proof}

This shows that as $L$ increases on the inteval $[0,MinL]$ and more errors are
introduced, the expected value of the probability of measuring any given $y$
decreases, making it harder to recover the period $P.$ Note that MinL need not
be an integer.

We can rewrite the function in the proof as%
\begin{align*}
f(L)  &  =\frac{L+M^{2}}{(N-(L+M))(L+M)}\\
&  =\frac{A}{N-(L+M)}+\frac{B}{(L+M)}%
\end{align*}

then we have
\begin{align*}
A(L+M)+B(N-(L+M))  &  =L+M^{2}\text{ then}\\
A-B  &  =1\text{ and}\\
AM+B(N-M)  &  =M^{2}%
\end{align*}

Solving these equations yields%

\begin{align*}
A  &  =1+\frac{M(M-1)}{N}\\
B  &  =\frac{M(M-1)}{N}%
\end{align*}

The first term in $f(L)$%
\[
\frac{A}{N-(L+M)}%
\]
is an increasing function of $L$ whereas the second term%
\[
\frac{B}{(L+M)}%
\]
is a decreasing function of $\dot{L}$. Adding these together produces a
function that has a minimum value.

\section{\textbf{The Three Step Amplified-QFT algorithm}}

In this section we provide the calculations for the probabilities of success
for the Amplified-QFT algorithm.

\emph{Problem:} We are given a binary valued Oracle $h(x)$ on $N$ labels
$\{0,1,...,N-1\}$, where $N=2^{n}$ , which takes the value $1$ on $C=A\cup G$
where $A$ is a periodic set of $M$ labels and $G$ is the set where the Error
Stream $g(x)=1$. We wish to determine the period P with the smallest number of
queries of the Oracle.

The Amplified-QFT algorithm is defined by the following three step procedure.

\emph{Step 1:} Apply all of Grover's algorithm in its entirety except for the
last measurement step to the starting state $|0>$. The resulting state is
given by $|\psi_{k}>$ (ref[4], ref[7],ref[1]) where $k=\left\lfloor \frac{\pi
}{4\sin^{-1}(\sqrt{T/N})}\right\rfloor $:%
\[
|\psi_{k}>=a_{k}\sum_{z\in C}|z>+b_{k}\sum_{z\notin C}|z>
\]

\noindent where%

\[
a_{k}=\frac{1}{\sqrt{T}}\sin(2k+1)\theta,b_{k}=\frac{1}{\sqrt{N-T}}%
\cos(2k+1)\theta
\]

\noindent are the appropriate amplitudes of the states and where
\[
\sin\theta=\sqrt{T/N},\cos\theta=\sqrt{1-T/N}%
\]

\noindent\ 

\emph{Step 2:} The QFT performs the following action%

\[
|z>\rightarrow\frac{1}{\sqrt{N}}\sum_{y=0}^{N-1}e^{-2\pi izy/N}|y>
\]

After the application of the QFT to the state $|\psi_{k}>$ , letting
$\omega=e^{-2\pi i/N}$ , we have%
\[
|\phi_{k}>=\frac{a_{k}}{\sqrt{N}}\sum_{z\in C}\sum_{y=0}^{N-1}\omega
^{zy}|y>+\frac{b_{k}}{\sqrt{N}}\sum_{z\notin C}\sum_{y=0}^{N-1}\omega^{zy}|y>
\]

After interchanging the order of summation, we have%
\[
|\phi_{k}>=\sum_{y=0}^{N-1}\left[  \frac{a_{k}}{\sqrt{N}}\sum_{z\in C}%
\omega^{zy}+\frac{b_{k}}{\sqrt{N}}\sum_{z\notin C}\omega^{zy}\right]  |y>
\]

.

\emph{Step 3:} Measure with respect to the standard basis to yield a integer
$y\in\{0,1,...,N-1\}$ from which we can determine the period P using the
continued fraction method.

\section{\textbf{Analysis of the Amplified-QFT Algorithm}}

We calculate the $\Pr(y)$ for the following cases:

\qquad a) $y=0$

\qquad b) $Py=0\operatorname{mod}N$ and $y\neq0$

\qquad c) $Py\neq0\operatorname{mod}N$

\bigskip The amplitude $Amp(y)$ of $|y>$ is given by
\begin{align*}
Amp(y)  &  =\frac{a_{k}}{\sqrt{N}}\sum_{z\in C}\omega^{zy}+\frac{b_{k}}%
{\sqrt{N}}\sum_{z\notin C}\omega^{zy}\\
&  =\frac{a_{k}}{\sqrt{N}}\sum_{z\in A}\omega^{zy}+\frac{a_{k}}{\sqrt{N}}%
\sum_{z\in G}\omega^{zy}+\frac{b_{k}}{\sqrt{N}}\sum_{z\notin C}\omega^{zy}\\
&  =\frac{(a_{k}-b_{k})}{\sqrt{N}}\left[  \sum_{z\in A}\omega^{zy}+\sum_{z\in
G}\omega^{zy}\right]  +\frac{b_{k}}{\sqrt{N}}\sum_{z=0}^{N-1}\omega^{zy}\\
&  =\frac{(a_{k}-b_{k})}{\sqrt{N}}\left[  \sum_{r=0}^{M-1}\omega
^{(s+rP)y}+\sum_{z\in G}\omega^{zy}\right]  +\frac{b_{k}}{\sqrt{N}}\sum
_{z=0}^{N-1}\omega^{zy}\\
&  =\frac{(a_{k}-b_{k})}{\sqrt{N}}\left[  \omega^{sy}\sum_{r=0}^{M-1}%
\omega^{rPy}+\sum_{z\in G}\omega^{zy}\right]  +\frac{b_{k}}{\sqrt{N}}%
\sum_{z=0}^{N-1}\omega^{zy}%
\end{align*}

\bigskip In the following we use the following four lemmas:

\begin{lemma}%
\[
\sum_{z=0}^{N-1}\omega^{zy}=\frac{1-\omega^{Ny}}{1-\omega^{y}}=0,w^{y}\neq1
\]

\end{lemma}

\begin{lemma}%
\[
\frac{T}{\sqrt{N}}(a_{k}-b_{k})=\tan\theta\sin2k\theta
\]

\end{lemma}

\begin{proof}%
\begin{align*}
\frac{T}{\sqrt{N}}(a_{k}-b_{k})  &  =\frac{T}{\sqrt{NT}}\sin(2k+1)\theta
-\frac{T}{\sqrt{N(N-T)}}\cos(2k+1)\theta\\
&  =\sqrt{\frac{T}{N}}(\sin(2k+1)\theta-\sqrt{\frac{T}{(N-T)}}\cos
(2k+1)\theta)\\
&  =\sqrt{\frac{T}{N}}(\sin(2k+1)\theta-\sqrt{\frac{T/N}{(1-T/N)}}%
\cos(2k+1)\theta)\\
&  =\sqrt{\frac{T}{N}}(\sin(2k+1)\theta-\frac{\sin\theta}{\cos\theta}%
\cos(2k+1)\theta)\\
&  =\tan\theta(\cos\theta\sin(2k+1)\theta-\sin\theta\cos(2k+1)\theta)\\
&  =\tan\theta\sin2k\theta
\end{align*}

\end{proof}

\begin{lemma}%
\[
\frac{T}{N}(\frac{N}{N-T})\geq\tan^{2}\theta\sin^{2}2k\theta\geq\frac{T}%
{N}(\frac{N}{N-T})(1-\frac{2T}{N})^{2}%
\]

\end{lemma}

\begin{proof}
Using $k=\left\lfloor \frac{\pi}{4\theta}\right\rfloor $ $\Longrightarrow$
$\frac{\pi}{4\theta}-1\leq k\leq\frac{\pi}{4\theta}$ $\Longrightarrow$
$\frac{\pi}{2}-2\theta\leq2k\theta\leq\frac{\pi}{2}\Longrightarrow\sin
(\frac{\pi}{2}-2\theta)\leq\sin2k\theta\leq1$ we have%
\begin{align*}
\frac{\sin^{2}\theta}{\cos^{2}\theta}  &  \geq\tan^{2}\theta\sin^{2}%
2k\theta\geq\tan^{2}\theta\sin^{2}(\frac{\pi}{2}-2\theta)\\
&  \Longrightarrow\frac{T}{N}\frac{1}{1-\frac{T}{N}}\geq\tan^{2}\theta\sin
^{2}2k\theta\geq\tan^{2}\theta\sin^{2}(\frac{\pi}{2}-2\theta)\\
&  \Longrightarrow\frac{T}{N}(\frac{N}{N-T})\geq\tan^{2}\theta\sin^{2}%
2k\theta\geq\frac{\sin^{2}\theta}{\cos^{2}\theta}\cos^{2}2\theta\\
&  \Longrightarrow\frac{T}{N}(\frac{N}{N-T})\geq\tan^{2}\theta\sin^{2}%
2k\theta\geq\frac{\sin^{2}\theta}{\cos^{2}\theta}(2\cos^{2}\theta-1)^{2}\\
&  \Longrightarrow\frac{T}{N}(\frac{N}{N-T})\geq\tan^{2}\theta\sin^{2}%
2k\theta\geq\frac{T}{N}(\frac{N}{N-T})(1-\frac{2T}{N})^{2}%
\end{align*}

\end{proof}

\begin{lemma}
If $Py\neq0\operatorname{mod}N$ then%
\[
\left\vert \frac{1-\omega^{MPy}}{1-\omega^{Py}}\right\vert ^{2}=\frac{\sin
^{2}(\pi MPy/N)}{\sin^{2}(\pi Py/N)}\leq M^{2}%
\]

\begin{proof}
We use the following result%
\[
|1-e^{i\theta}|^{2}=4\sin^{2}(\theta/2)
\]

\end{proof}
\end{lemma}

\noindent

\subsection{Amplified-QFT Analysis: $y=0$}

\noindent We have%

\begin{align*}
Amp(y)  &  =\frac{a_{k}}{\sqrt{N}}\sum_{z\in C}\omega^{zy}+\frac{b_{k}}%
{\sqrt{N}}\sum_{z\notin C}\omega^{zy}\\
&  =\frac{1}{\sqrt{N}}(Ta_{k}+(N-T)b_{k})\\
&  =\frac{1}{\sqrt{N}}\left[  \frac{T}{\sqrt{T}}\sin(2k+1)\theta+\frac
{N-T}{\sqrt{N-T}}\cos(2k+1)\theta\right] \\
&  =\sqrt{\frac{T}{N}}\sin(2k+1)\theta+\sqrt{1-\frac{T}{N}}\cos(2k+1)\theta\\
&  =\sin\theta\sin(2k+1)\theta+\cos\theta\cos(2k+1)\theta\\
&  =\cos(2k\theta)
\end{align*}

\noindent We have
\[
\Pr(y=0)=\cos^{2}(2k\theta)
\]

\begin{lemma}%
\[
\frac{4T}{N}(1-\frac{T}{N})=\sin^{2}2\theta\geq\Pr(y=0)=\cos^{2}2k\theta\geq0
\]

\end{lemma}

\begin{proof}
Using $k=\left\lfloor \frac{\pi}{4\theta}\right\rfloor $ $\Longrightarrow$
$\frac{\pi}{4\theta}-1\leq k\leq\frac{\pi}{4\theta}$ $\Longrightarrow$
$\frac{\pi}{2}-2\theta\leq2k\theta\leq\frac{\pi}{2}\Longrightarrow\sin
(2\theta)=\cos(\frac{\pi}{2}-2\theta)\geq\cos2k\theta\geq\cos\frac{\pi}{2}=0$
we have%
\[
\sin2\theta=2\sin\theta\cos\theta=2\sqrt{\frac{T}{N}}\sqrt{1-\frac{T}{N}}%
\]

\end{proof}

\subsection{Amplified-QFT Analysis: $Py=0\operatorname{mod}N,y\neq0$}

\noindent

\bigskip By making use of previous lemmas we have
\begin{align*}
Amp(y)  &  =\frac{(a_{k}-b_{k})}{\sqrt{N}}\left[  \omega^{sy}\sum_{r=0}%
^{M-1}\omega^{rPy}+\sum_{z\in G}\omega^{zy}\right]  +\frac{b_{k}}{\sqrt{N}%
}\sum_{z=0}^{N-1}\omega^{zy}\\
&  =\frac{(a_{k}-b_{k})}{\sqrt{N}}\left[  \omega^{sy}\sum_{r=0}^{M-1}%
\omega^{rPy}+\sum_{z\in G}\omega^{zy}\right] \\
&  =\frac{(a_{k}-b_{k})}{\sqrt{N}}\left[  \omega^{sy}M+\sum_{z\in G}%
\omega^{zy}\right] \\
&  =\frac{T(a_{k}-b_{k})}{\sqrt{N}}\left[  \frac{M}{T}\omega^{sy}+\frac{1}%
{T}\sum_{z\in G}\omega^{zy}\right] \\
&  =\tan\theta\sin2k\theta\left[  \frac{M}{T}\omega^{sy}+\frac{1}{T}\sum_{z\in
G}\omega^{zy}\right]
\end{align*}

Therefore by the following lemma we have this result for the $\Pr(y):$%

\begin{align*}
\Pr(y)  &  =\tan^{2}\theta\sin^{2}2k\theta\left\vert \frac{M}{T}\omega
^{sy}+\frac{1}{T}\sum_{z\in G}\omega^{zy}\right\vert ^{2}\\
&  \leq\tan^{2}\theta\sin^{2}2k\theta
\end{align*}

\begin{lemma}
$\left\vert \frac{M}{T}\omega^{sy}+\frac{1}{T}\sum_{z\in G}\omega
^{zy}\right\vert ^{2}\leq1$
\end{lemma}

\begin{proof}%
\begin{align*}
&  \left\vert \frac{M}{T}\omega^{sy}+\frac{1}{T}\sum_{z\in G}\omega
^{zy}\right\vert ^{2}\\
&  \leq\left\vert \frac{M}{T}\omega^{sy}\right\vert ^{2}+\left\vert 2\frac
{M}{T}\omega^{sy}\frac{1}{T}\sum_{z\in G}\omega^{zy}\right\vert +\left\vert
\frac{1}{T}\sum_{z\in G}\omega^{zy}\right\vert ^{2}\\
&  \leq\frac{M^{2}}{T^{2}}+\frac{2ML}{T^{2}}+\frac{L^{2}}{T^{2}}\\
&  =1
\end{align*}

\end{proof}

\subsection{Amplified-QFT Analysis: $Py\neq0\operatorname{mod}N$}

\noindent Making use of the previous lemmas we have
\begin{align*}
Amp(y)  &  =\frac{(a_{k}-b_{k})}{\sqrt{N}}\left[  \omega^{sy}\sum_{r=0}%
^{M-1}\omega^{rPy}+\sum_{z\in G}\omega^{zy}\right]  +\frac{b_{k}}{\sqrt{N}%
}\sum_{z=0}^{N-1}\omega^{zy}\\
&  =\frac{(a_{k}-b_{k})}{\sqrt{N}}\left[  \omega^{sy}\sum_{r=0}^{M-1}%
\omega^{rPy}+\sum_{z\in G}\omega^{zy}\right] \\
&  =\frac{(a_{k}-b_{k})}{\sqrt{N}}\left[  \omega^{sy}\left[  \frac
{1-\omega^{MPy}}{1-\omega^{Py}}\right]  +\sum_{z\in G}\omega^{zy}\right] \\
&  =\frac{T(a_{k}-b_{k})}{\sqrt{N}}\left[  \frac{1}{T}\omega^{sy}\left[
\frac{1-\omega^{MPy}}{1-\omega^{Py}}\right]  +\frac{1}{T}\sum_{z\in G}%
\omega^{zy}\right] \\
&  =\tan\theta\sin2k\theta\left[  \frac{1}{T}\omega^{sy}\left[  \frac
{1-\omega^{MPy}}{1-\omega^{Py}}\right]  +\frac{1}{T}\sum_{z\in G}\omega
^{zy}\right]
\end{align*}

\bigskip

\noindent Therefore by the following lemma we have this result for the
$\Pr(y):$
\begin{align*}
\Pr(y)  &  =\tan^{2}\theta\sin^{2}2k\theta\left\vert \frac{1}{T}\omega
^{sy}\left[  \frac{1-\omega^{MPy}}{1-\omega^{Py}}\right]  +\frac{1}{T}%
\sum_{z\in G}\omega^{zy}\right\vert ^{2}\\
&  \leq\tan^{2}\theta\sin^{2}2k\theta
\end{align*}

\begin{lemma}
\noindent$\left\vert \frac{1}{T}\omega^{sy}\left[  \frac{1-\omega^{MPy}%
}{1-\omega^{Py}}\right]  +\frac{1}{T}\sum_{z\in G}\omega^{zy}\right\vert
^{2}\leq1$

\begin{proof}%
\begin{align*}
&  \left\vert \frac{1}{T}\omega^{sy}\left[  \frac{1-\omega^{MPy}}%
{1-\omega^{Py}}\right]  +\frac{1}{T}\sum_{z\in G}\omega^{zy}\right\vert ^{2}\\
&  \leq\left\vert \frac{1}{T}\omega^{sy}\left[  \frac{1-\omega^{MPy}}%
{1-\omega^{Py}}\right]  \right\vert ^{2}+2\left\vert \frac{1}{T}\omega
^{sy}\left[  \frac{1-\omega^{MPy}}{1-\omega^{Py}}\right]  \right\vert
\left\vert \frac{1}{T}\sum_{z\in G}\omega^{zy}\right\vert +\left\vert \frac
{1}{T}\sum_{z\in G}\omega^{zy}\right\vert ^{2}\\
&  \leq\frac{1}{T^{2}}\left\vert \frac{\sin(\pi MPy/N)}{\sin(\pi
Py/N)}\right\vert ^{2}+\frac{2L}{T^{2}}\left\vert \frac{\sin(\pi MPy/N)}%
{\sin(\pi Py/N)}\right\vert +\frac{L^{2}}{T^{2}}\\
&  \leq\frac{M^{2}}{T^{2}}+\frac{2LM}{T^{2}}+\frac{L^{2}}{T^{2}}\\
&  =1
\end{align*}

\end{proof}
\end{lemma}

\noindent We notice that if in addition $MPy=0\operatorname{mod}N$ then%

\begin{align*}
\Pr(y)  &  =\tan^{2}\theta\sin^{2}2k\theta\left\vert \frac{1}{T}\sum_{z\in
G}\omega^{zy}\right\vert ^{2}\\
&  \leq\frac{L^{2}}{T^{2}}\tan^{2}\theta\sin^{2}2k\theta
\end{align*}

\subsection{Amplified-QFT Summary}

The probability $\Pr(y)$ is given exactly by

\begin{center}%
\[
\left\{
\begin{tabular}
[c]{l}%
Case A: $\cos^{2}2k\theta$\\
\\
Case B: $\tan^{2}\theta\sin^{2}2k\theta\left\vert \frac{M}{T}\omega^{sy}%
+\frac{1}{T}\sum_{z\in G}\omega^{zy}\right\vert ^{2}$\\
\\
Case C: $\tan^{2}\theta\sin^{2}2k\theta\left\vert \frac{1}{T}\omega
^{sy}\left[  \frac{1-\omega^{MPy}}{1-\omega^{Py}}\right]  +\frac{1}{T}%
\sum_{z\in G}\omega^{zy}\right\vert ^{2}$\\
\\
Case D: $\tan^{2}\theta\sin^{2}2k\theta\left\vert \frac{1}{T}\sum_{z\in
G}\omega^{zy}\right\vert ^{2}$%
\end{tabular}
\ \ \ \ \ \ \right\}
\]

\end{center}

\section{\textbf{Applying the QFT to the Oracle.}}

In this section we just apply the QFT to the binary Oracle $h$, which is $1$
on $C$ and $0$ elsewhere.

We begin with the following state%

\[
|\xi>=\frac{1}{\sqrt{N}}\sum_{z=0}^{N-1}|z>\otimes\frac{1}{\sqrt{2}}(|0>-|1>)
\]
\bigskip

\noindent and apply the unitary transform for $h$, $U_{h}$ ,\ to this state
which performs the following action:

\bigskip%
\[
U_{h}|z>|c>=|z>|c\oplus h(z)>
\]

\noindent to get the state $|\psi>$%

\begin{align*}
|\psi &  >=U_{h}\frac{1}{\sqrt{N}}\sum_{z=0}^{N-1}|z>\frac{1}{\sqrt{2}%
}(|0>-|1>)\\
&  =\frac{1}{\sqrt{N}}\left[  (-1)\sum_{z\in C}|z>+\sum_{z\notin C}|z>\right]
\frac{1}{\sqrt{2}}(|0>-|1>)\\
&  =\frac{1}{\sqrt{N}}\left[  (-2)\sum_{z\in C}|z>+\sum_{z=0}^{N-1}|z>\right]
\frac{1}{\sqrt{2}}(|0>-|1>)\\
&  =\frac{1}{\sqrt{N}}\left[  (-2)\sum_{z\in A}|z>-2\sum_{z\in G}%
|z>+\sum_{z=0}^{N-1}|z>\right]  \frac{1}{\sqrt{2}}(|0>-|1>)
\end{align*}

\noindent Next we apply the QFT to try to find the period P, dropping
$\frac{1}{\sqrt{2}}(|0>-|1>)$.

\noindent The QFT applies the following action:%
\[
|z>\rightarrow\frac{1}{\sqrt{N}}\sum_{y=0}^{N-1}\omega^{zy}|y>
\]

\noindent to get%
\begin{align*}
|\phi &  >=\sum_{y=0}^{N-1}\left[  \frac{-2}{N}\sum_{z\in C}\omega^{zy}%
+\frac{1}{N}\sum_{z=0}^{N-1}\omega^{zy}\right]  |y>\\
&  =\sum_{y=0}^{N-1}\left[  \frac{-2}{N}\sum_{z\in A}\omega^{zy}-\frac{2}%
{N}\sum_{z\in G}\omega^{zy}+\frac{1}{N}\sum_{z=0}^{N-1}\omega^{zy}\right]
|y>\\
&  =\sum_{y=0}^{N-1}\left[  \frac{-2}{N}\omega^{sy}\sum_{r=0}^{M-1}%
\omega^{rPy}-\frac{2}{N}\sum_{z\in G}\omega^{zy}+\frac{1}{N}\sum_{z=0}%
^{N-1}\omega^{zy}\right]  |y>
\end{align*}

\subsection{QFT Analysis: $y=0$}

We have%

\begin{align*}
Amp(y)  &  =\frac{(-2)}{N}\sum_{z\in C}\omega^{zy}+\frac{1}{N}\sum_{z=0}%
^{N-1}\omega^{zy}\\
&  =\frac{-2T}{N}+\frac{N}{N}\\
&  =1-\frac{2T}{N}%
\end{align*}

\noindent Therefore, in the QFT case, we have $\Pr(y=0)$ is very close to $1$
and is given by
\[
\Pr(y=0)=\left(  1-\frac{2T}{N}\right)  ^{2}%
\]

\noindent whereas in the Amplified-QFT case we have $\Pr(y=0)$ is given by
\[
\Pr(y=0)=\cos^{2}2k\theta
\]

\subsection{QFT Analysis: $Py=0\operatorname{mod}N,y\neq0$}

Using previous lemmas we have
\begin{align*}
Amp(y)  &  =\ \frac{-2}{N}\omega^{sy}\sum_{r=0}^{M-1}\omega^{rPy}-\frac{2}%
{N}\sum_{z\in G}\omega^{zy}+\frac{1}{N}\sum_{z=0}^{N-1}\omega^{zy}\\
&  =\frac{-2}{N}\omega^{sy}\sum_{r=0}^{M-1}\omega^{rPy}-\frac{2}{N}\sum_{z\in
G}\omega^{zy}\\
&  =\frac{-2}{N}\left[  \omega^{sy}M+\sum_{z\in G}\omega^{zy}\right]
\end{align*}

\noindent Therefore in the QFT$\ $case we have $\Pr(y)$ is given by
\begin{align*}
\Pr(y)  &  =\frac{4}{N^{2}}\left\vert \omega^{sy}M+\sum_{z\in G}\omega
^{zy}\right\vert ^{2}\\
&  \leq\frac{4T^{2}}{N^{2}}%
\end{align*}

\noindent whereas in the Amplified-QFT case we have $\Pr(y)$ is given by \
\[
\Pr(y)=\tan^{2}\theta\sin^{2}2k\theta\left\vert \frac{M}{T}\omega^{sy}%
+\frac{1}{T}\sum_{z\in G}\omega^{zy}\right\vert ^{2}%
\]

\noindent We can determine how the increase in amplitude varies with the
number of iterations $k$ of the Grover step in the Amplified-QFT by examining
the ratio of the amplitudes of the Amplified-QFT case and QFT\ case. This
ratio is given exactly by
\begin{align*}
AmpRatio(y)  &  =\frac{\tan\theta\sin2k\theta\left[  \frac{M}{T}\omega
^{sy}+\frac{1}{T}\sum_{z\in G}\omega^{zy}\right]  }{\frac{-2}{N}\left[
\omega^{sy}M+\sum_{z\in G}\omega^{zy}\right]  }\\
&  =\frac{N}{-2T}\tan\theta\sin2k\theta
\end{align*}

\noindent We also have the following inequality for the $\Pr Ratio(y)$, the
increase in the probability due to amplification:%

\begin{align*}
\frac{N}{4T}(\frac{N}{N-T})  &  \geq\Pr Ratio(y)\geq\frac{N}{4T}(\frac{N}%
{N-T})(1-\frac{2T}{N})^{2}\\
&  \Longrightarrow\Pr Ratio(y)\approx\frac{N}{4T}%
\end{align*}

\subsection{QFT Analysis: $Py\neq0\operatorname{mod}N$}

We have
\begin{align*}
Amp(y)  &  =\ \frac{-2}{N}\omega^{sy}\sum_{r=0}^{M-1}\omega^{rPy}-\frac{2}%
{N}\sum_{z\in G}\omega^{zy}+\frac{1}{N}\sum_{z=0}^{N-1}\omega^{zy}\\
&  =\frac{-2}{N}\omega^{sy}\sum_{r=0}^{M-1}\omega^{rPy}-\frac{2}{N}\sum_{z\in
G}\omega^{zy}\\
&  =\frac{-2}{N}w^{sy}\left[  \frac{1-\omega^{MPy}}{1-\omega^{Py}}\right]
-\frac{2}{N}\sum_{z\in G}\omega^{zy}%
\end{align*}

\bigskip\noindent In the QFT\ case, we have $\Pr(y)$ is given by$\qquad
\qquad\qquad\qquad\qquad\qquad$%
\begin{align*}
\Pr(y)  &  =\frac{4}{N^{2}}\left\vert w^{sy}\left[  \frac{1-\omega^{MPy}%
}{1-\omega^{Py}}\right]  +\sum_{z\in G}\omega^{zy}\right\vert ^{2}\\
&  \leq\frac{4}{N^{2}}\left\vert w^{sy}\left\vert \frac{\sin(\pi MPy/N)}%
{\sin(\pi Py/N)}\right\vert +\sum_{z\in G}\omega^{zy}\right\vert ^{2}\\
&  \leq\frac{4}{N^{2}}\left[  M+L\right]  ^{2}\\
&  \leq\frac{4T^{2}}{N^{2}}%
\end{align*}

\noindent whereas in the Amplified-QFT case we have $\Pr(y)$ is given by
\[
\Pr(y)=\tan^{2}\theta\sin^{2}2k\theta\left\vert \frac{1}{T}\omega^{sy}\left[
\frac{1-\omega^{MPy}}{1-\omega^{Py}}\right]  +\frac{1}{T}\sum_{z\in G}%
\omega^{zy}\right\vert ^{2}%
\]

\noindent The ratio of the amplitudes of the Amplified-QFT case and QFT\ case
is given exactly by
\begin{align*}
AmpRatio(y)  &  =\frac{\tan\theta\sin2k\theta\left[  \frac{1}{T}\omega
^{sy}\left[  \frac{1-\omega^{MPy}}{1-\omega^{Py}}\right]  +\frac{1}{T}%
\sum_{z\in G}\omega^{zy}\right]  }{\frac{-2}{N}w^{sy}\left[  \frac
{1-\omega^{MPy}}{1-\omega^{Py}}\right]  -\frac{2}{N}\sum_{z\in G}\omega^{zy}%
}\\
&  =\frac{N}{-2T}\tan\theta\sin2k\theta
\end{align*}

\noindent We note that this ratio is the same as in that given in the previous
section and is independent of $y$. The variables in this ratio do not depend
in anyway on the QFT. We also have the following inequality for the $\Pr
Ratio(y)$, the increase in the probability due to amplification:%

\begin{align*}
\frac{N}{4T}(\frac{N}{N-T})  &  \geq\Pr Ratio(y)\geq\frac{N}{4T}(\frac{N}%
{N-T})(1-\frac{2T}{N})^{2}\\
&  \Longrightarrow\Pr Ratio(y)\approx\frac{N}{4T}%
\end{align*}

We notice that if in addition $MPy=0\operatorname{mod}N$ then%

\begin{align*}
\Pr(y)  &  =\frac{4}{N^{2}}\left\vert \sum_{z\in G}\omega^{zy}\right\vert
^{2}\\
&  \leq\frac{4L^{2}}{N^{2}}%
\end{align*}

\subsection{QFT Summary}

The probability $\Pr(y)$ is given exactly by%

\[
\left\{
\begin{tabular}
[c]{l}%
Case A: $\left(  1-\frac{2T}{N}\right)  ^{2}$\\
\\
Case B: $\frac{4}{N^{2}}\left\vert \omega^{sy}M+\sum_{z\in G}\omega
^{zy}\right\vert ^{2}$\\
\\
Case C: $\frac{4}{N^{2}}\left\vert w^{sy}\left[  \frac{1-\omega^{MPy}%
}{1-\omega^{Py}}\right]  +\sum_{z\in G}\omega^{zy}\right\vert ^{2}$\\
\\
Case D: $\frac{4}{N^{2}}\left\vert \sum_{z\in G}\omega^{zy}\right\vert ^{2}$%
\end{tabular}
\ \ \ \ \ \ \ \right\}
\]

\section{\textbf{Applying the QHS to the Oracle}}

In this section we provide the calculations for the probabilities of success
for the QHS algorithm.The QHS algorithm is a two register algorithm as follows
(see ref[13] for details). We begin with $|0>|0>$ where the first register is
$n$ qubits and the second register is $1$ qubit and apply the Hadamard
transform to the first register to get a uniform superposition state, followed
by the unitary transformation for the Oracle h to get:%

\[
|\psi>=\frac{1}{\sqrt{N}}%
{\displaystyle\sum\limits_{x=0}^{N-1}}
|x>|h(x)>
\]
Next we apply the QFT to the first register to get%

\begin{align*}
|\psi &  >=\frac{1}{\sqrt{N}}%
{\displaystyle\sum\limits_{x=0}^{N-1}}
\frac{1}{\sqrt{N}}\sum_{y=0}^{N-1}\omega^{xy}|y>|h(x)>\\
&  =%
{\displaystyle\sum\limits_{y=0}^{N-1}}
\frac{1}{N}\sum_{x=0}^{N-1}\omega^{xy}|y>|h(x)>\\
&  =%
{\displaystyle\sum\limits_{y=0}^{N-1}}
\frac{1}{N}|y>\sum_{x=0}^{N-1}\omega^{xy}|h(x)>\\
&  =\
{\displaystyle\sum\limits_{y=0}^{N-1}}
\frac{|||\Gamma(y)>||}{N}|y>\frac{|\Gamma(y)>}{|||\Gamma(y)>||}%
\end{align*}

\noindent where
\begin{align*}
|\Gamma(y)  &  >=\sum_{x=0}^{N-1}\omega^{xy}|h(x)>\\
&  =\sum_{x\in C}^{{}}\omega^{xy}|1>+\sum_{x\notin C}^{{}}\omega^{xy}|0>
\end{align*}

\noindent and where
\[
|||\Gamma(y)>||^{2}=\left\vert \sum_{x\in C}^{{}}\omega^{xy}\right\vert
^{2}+\left\vert \sum_{x\notin C}^{{}}\omega^{xy}\right\vert ^{2}%
\]

\bigskip\noindent Next we make a measurement to get $y$ and find that the
probability of this measurement is%
\begin{align*}
\Pr(y)  &  =\frac{|||\Gamma(y)>||^{2}}{N^{2}}\\
&  =\frac{1}{N^{2}}\left\vert \sum_{x\in C}^{{}}\omega^{xy}\right\vert
^{2}+\frac{1}{N^{2}}\left\vert \sum_{x\notin C}^{{}}\omega^{xy}\right\vert
^{2}%
\end{align*}

\noindent The state that we end up in is of the form
\[
|\phi>=|y>\frac{|\Gamma(y)>}{|||\Gamma(y)>||}%
\]

\noindent So now we are interested in the probability of measuring $\ y$ in
the usual cases in order to recover the period $P$.

\subsection{\bigskip QHS Analysis: $y=0$}

We have
\begin{align*}
\Pr(y)  &  =\frac{1}{N^{2}}\left\vert \sum_{x\in C}^{{}}\omega^{xy}\right\vert
^{2}+\frac{1}{N^{2}}\left\vert \sum_{x\notin C}^{{}}\omega^{xy}\right\vert
^{2}\\
&  =\frac{T^{2}}{N^{2}}+\frac{(N-T)^{2}}{N^{2}}=\frac{T^{2}+N^{2}-2NT+T^{2}%
}{N^{2}}\\
&  =1-\frac{2T(N-T)}{N^{2}}%
\end{align*}

\noindent whereas in the Amplified-QFT case we have $\Pr(y=0)$ is given by
\[
\Pr(y=0)=\cos^{2}2k\theta
\]

\subsection{QHS Analysis: $Py=0\operatorname{mod}N,y\neq0$}

We have%
\begin{align*}
\Pr(y)  &  =\frac{1}{N^{2}}\left\vert \sum_{x\in C}^{{}}\omega^{xy}\right\vert
^{2}+\frac{1}{N^{2}}\left\vert \sum_{x\notin C}^{{}}\omega^{xy}\right\vert
^{2}\\
&  =\frac{1}{N^{2}}\left\vert \omega^{sy}\sum_{r=0}^{M-1}\omega^{rPy}%
+\sum_{x\in G}^{{}}\omega^{xy}\right\vert ^{2}+\frac{1}{N^{2}}\left\vert
\sum_{x\notin C}^{{}}\omega^{xy}\right\vert ^{2}\\
&  =\frac{1}{N^{2}}\left\vert \omega^{sy}M+\sum_{x\in G}^{{}}\omega
^{xy}\right\vert ^{2}+\frac{1}{N^{2}}\left\vert -\omega^{sy}M-\sum_{x\in
G}^{{}}\omega^{xy}+\frac{1}{N}\sum_{x=0}^{N-1}\omega^{xy}\right\vert ^{2}\\
&  =\frac{2}{N^{2}}\left\vert \omega^{sy}M+\sum_{x\in G}^{{}}\omega
^{xy}\right\vert ^{2}\\
&  \leq\frac{2T^{2}}{N^{2}}%
\end{align*}

In the Amplified-QFT case we have $\Pr(y)$ is given by \
\[
\Pr(y)=\tan^{2}\theta\sin^{2}2k\theta\left\vert \frac{M}{T}\omega^{sy}%
+\frac{1}{T}\sum_{z\in G}\omega^{zy}\right\vert ^{2}%
\]

\bigskip

By comparing the results of the QHS and the Amplified-QFT algorithms we have
the following inequality for the $\Pr Ratio(y)=\Pr(y)_{Amplified-QFT}%
/\Pr(y)_{QHS}$, the increase in the probability due to amplification%

\begin{align*}
\Pr Ratio(y)  &  =\frac{\tan^{2}\theta\sin^{2}2k\theta\left\vert \frac{M}%
{T}\omega^{sy}+\frac{1}{T}\sum_{z\in G}\omega^{zy}\right\vert ^{2}}{\frac
{2}{N^{2}}\left\vert \omega^{sy}M+\sum_{x\in G}^{{}}\omega^{xy}\right\vert
^{2}}\\
&  =\frac{N^{2}}{2T^{2}}\tan^{2}\theta\sin^{2}2k\theta
\end{align*}

which gives%

\begin{align*}
\frac{N}{2T}(\frac{N}{N-T})  &  \geq\Pr Ratio(y)\geq\frac{N}{2T}(\frac{N}%
{N-T})(1-\frac{2T}{N})^{2}\\
&  \Longrightarrow\Pr Ratio(y)\approx\frac{N}{2T}%
\end{align*}

\subsection{QHS Analysis: $Py\neq0\operatorname{mod}N$}

We have%
\begin{align*}
\Pr(y)  &  =\frac{1}{N^{2}}\left\vert \omega^{sy}M+\sum_{x\in G}^{{}}%
\omega^{xy}\right\vert ^{2}+\frac{1}{N^{2}}\left\vert -\omega^{sy}M-\sum_{x\in
G}^{{}}\omega^{xy}+\frac{1}{N}\sum_{x=0}^{N-1}\omega^{xy}\right\vert ^{2}\\
&  =\frac{2}{N^{2}}\left\vert \omega^{sy}\sum_{r=0}^{M-1}\omega^{rPy}%
+\sum_{x\in G}^{{}}\omega^{xy}\right\vert ^{2}\\
&  =\frac{2}{N^{2}}\left\vert \omega^{sy}\ \left[  \frac{1-\omega^{MPy}%
}{1-\omega^{Py}}\right]  +\sum_{x\in G}^{{}}\omega^{xy}\right\vert ^{2}\\
&  \leq\frac{2}{N^{2}}\left\vert \omega^{sy}\ \left\vert \frac{\sin(\pi
MPy/N)}{\sin(\pi Py/N)}\right\vert +\sum_{x\in G}^{{}}\omega^{xy}\right\vert
^{2}\\
&  \leq\frac{2}{N^{2}}\left[  M+L\right]  ^{2}\\
&  \leq\frac{2T^{2}}{N^{2}}%
\end{align*}

\bigskip In the Amplified-QFT case we have $\Pr(y)$ is given by
\[
\Pr(y)=\tan^{2}\theta\sin^{2}2k\theta\left\vert \frac{1}{T}\omega^{sy}\left[
\frac{1-\omega^{MPy}}{1-\omega^{Py}}\right]  +\frac{1}{T}\sum_{z\in G}%
\omega^{zy}\right\vert ^{2}%
\]

By comparing the results of the QHS and the Amplified-QFT algorithms we have
the following inequality for the $\Pr Ratio(y)=\Pr(y)_{Amplified-QFT}%
/\Pr(y)_{QHS}$, the increase in the probability due to amplification%

\begin{align*}
\Pr Ratio(y)  &  =\frac{\tan^{2}\theta\sin^{2}2k\theta\left\vert \frac{1}%
{T}\omega^{sy}\left[  \frac{1-\omega^{MPy}}{1-\omega^{Py}}\right]  +\frac
{1}{T}\sum_{z\in G}\omega^{zy}\right\vert ^{2}}{\frac{2}{N^{2}}\left\vert
\omega^{sy}\ \left[  \frac{1-\omega^{MPy}}{1-\omega^{Py}}\right]  +\sum_{x\in
G}^{{}}\omega^{xy}\right\vert ^{2}}\\
&  =\frac{N^{2}}{2T^{2}}\tan^{2}\theta\sin^{2}2k\theta
\end{align*}

which gives%

\begin{align*}
\frac{N}{2T}(\frac{N}{N-T})  &  \geq\Pr Ratio(y)\geq\frac{N}{2T}(\frac{N}%
{N-T})(1-\frac{2T}{N})^{2}\\
&  \Longrightarrow\Pr Ratio(y)\approx\frac{N}{2T}%
\end{align*}

We notice that if in addition $MPy=0\operatorname{mod}N$ then%

\begin{align*}
\Pr(y)  &  =\frac{2}{N^{2}}\left\vert \sum_{z\in G}\omega^{zy}\right\vert
^{2}\\
&  \leq\frac{2L^{2}}{N^{2}}%
\end{align*}

\subsection{QHS Summary}

The $\Pr(y)$ in the QHS case is:%

\[
\left\{
\begin{tabular}
[c]{l}%
Case A: $1-\frac{2T(N-T)}{N^{2}}$\\
\\
Case B: $\frac{2}{N^{2}}\left\vert \omega^{sy}M+\sum_{z\in G}\omega
^{zy}\right\vert ^{2}$\\
\\
Case C: $\frac{2}{N^{2}}\left\vert w^{sy}\left[  \frac{1-\omega^{MPy}%
}{1-\omega^{Py}}\right]  +\sum_{z\in G}\omega^{zy}\right\vert ^{2}$\\
\\
Case D: $\frac{2}{N^{2}}\left\vert \sum_{z\in G}\omega^{zy}\right\vert ^{2}$%
\end{tabular}
\ \ \ \ \ \ \ \ \right\}
\]

\newpage\renewcommand{\thechapter}{4}

\chapter{An Uncertainty Principle for the Amplified-QFT}

In this chapter we show there is an uncertainty principle for the
Amplified-QFT\ algorithm. This result provides a relationship between the
support of the state vector after Grover's algorithm has been run and the
support of the state vector after the QFT has been run. This result uses the
results of Donoho and Stark found in ref[15], ref[16] and ref[17]. First we
state and prove the Donoho and Stark lemma 1 from their paper which we will
use to good effect for the Amplified-QFT\ case.

\begin{lemma}
If $\{x_{j}\}$ $j=0,1,...,N-1$ has $T$ nonzero elements, then $\{y_{k}\}$
$k=0,1,...,N-1$ cannot have $T$ consecutive zeros, where $\{y_{k}\}$ is the
discrete Fourier transform of $\{x_{j}\}.$
\end{lemma}

\begin{proof}
Define%
\[
y_{k}=\frac{1}{\sqrt{N}}\sum_{j=0}^{N-1}x_{j}w^{jk}%
\]
where $w=\exp(-2\pi i/N).$ Suppose there are $T$ consecutive positions
$\{y_{t+r}\}$ $r=0,1,...,T-1$ which are all zero. Then we have a system of
$\ T$ equations each of which are zero as follows:%
\[
y_{t+r}=\frac{1}{\sqrt{N}}\sum_{j=0}^{N-1}x_{j}w^{j(t+r)}=0,r=0,...,T-1
\]
However there are only $T$ values of $x_{j}$ which are nonzero. Let us call
these positions $S=\{s_{j}\}$ $j=0,1,...,T-1.$ Then we can rewrite our system
of equations as follows:%
\[
y_{t+r}=\frac{1}{\sqrt{N}}\sum_{j=0}^{T-1}x_{s_{j}}w^{s_{j}(t+r)}=0
\]
Then we have a system of $T$ equations in $T$ unknowns equal to zero.
\[
Zx=0
\]
However the vector $x$ contains elements $x_{s_{j}}$ which are all nonzero.
Therefore the matrix $Z$ must be singular. However we will show that $Z$ is
non-singular, thereby showing that we cannot have such a system of equations
and cannot have $T$ consecutive $y_{k}=0.$ Let us take a closer look at $Z.$%
\[
Z=\left[
\begin{array}
[c]{cccc}%
w^{s_{0}t} & w^{s_{1}t} & ... & w^{s_{T-1}t}\\
w^{s_{0}(t+1)} & w^{s_{1}(t+1)} & ... & w^{s_{T-1}(t+1)}\\
... & ... & ... & ...\\
w^{s_{0}(t+T-1)} & w^{s_{1}(t+T-1)} & ... & w^{s_{T-1}(t+T-1)}%
\end{array}
\right]
\]
We can consider an equivalent set of equations%
\[
ZPP^{-1}x=0
\]
where $P^{-1}x$ has all nonzero elements%
\[
P=\left[
\begin{array}
[c]{cccc}%
w^{-s_{0}t} & 0 & ... & 0\\
0 & w^{-s_{1}t} & ... & 0\\
... & ... & ... & ...\\
0 & 0 & ... & w^{-s_{T-1}t}%
\end{array}
\right]
\]
where $\ ZP$ is given by%
\[
Z^{\prime}=\left[
\begin{array}
[c]{cccc}%
1 & 1 & ... & 1\\
w^{s_{0}} & w^{s_{1}} & ... & w^{s_{T-1}}\\
... & ... & ... & ...\\
w^{s_{0}(T-1)} & w^{s_{1}(T-1)} & ... & w^{s_{T-1}(T-1)}%
\end{array}
\right]
\]
which can be rewritten as a Vandermonde matrix%
\[
V=\left[
\begin{array}
[c]{cccc}%
1 & 1 & ... & 1\\
\alpha_{0} & \alpha_{1} & ... & \alpha_{T-1}\\
... & ... & ... & ...\\
\alpha_{0}^{T-1} & \alpha_{1}^{T-1} & ... & \alpha_{T-1}^{T-1}%
\end{array}
\right]
\]
which is known to be nonsingular.
\end{proof}

Note that in the proof we can consider the positions to be taken
$\operatorname{mod}N$ so that the $\{x_{j}\}$ and $\{y_{k}\}$ can be viewed as
being on a circle. This takes into account that we cannot have $T$ consecutive
zeros wrapping around the endpoints.

If we cannot have $T$ consecutive zeros in $\{y_{k}\}$ then the number of
nonzero elements in $\{y_{k}\}$ must be at least $N/T$. For example, we could
have $T-1$ zeros followed by a single nonzero element in every $T$ long block
of $\{y_{k}\}$. Since there are $N/T$ such blocks we have the following result:

\begin{theorem}
(Donoho and Stark)\ Let $\{x_{j}\}$ $j=0,1...,N-1$ have $N_{x}$ nonzero
elements. Let $\{y_{k}\}$ $k=0,1,...,N-1$ be the Fourier transform of
$\{x_{j}\}$ with $N_{y}$ nonzero elements. Then
\[
N_{x}N_{y}\geq N
\]

\end{theorem}

Next we apply this to the quantum case. Consider the following state where
$\{p_{x}\}$ is a probability distribution%
\[
|\psi_{x}>=\sum_{x=0}^{N-1}\sqrt{p_{x}}|x>
\]

and apply the QFT which maps%

\[
|x>\rightarrow\frac{1}{\sqrt{N}}\sum_{y=0}^{N-1}w^{xy}|y>
\]

to get the state%

\[
|\psi_{y}>=\frac{1}{\sqrt{N}}\sum_{y=0}^{N-1}\sum_{x=0}^{N-1}\sqrt{p_{x}%
}w^{xy}|y>
\]

Suppose we have $T$ of the $\sqrt{p_{x}}$ amplitudes of $|\psi_{x}>$ nonzero
at positions $s_{j},$ $j=0,1,..,T-1$ and suppose we have $T$ consecutive
amplitudes of $|\psi_{y}>$ equal to zero, we have a system of $T$ equations as
in the lemma%

\[
y_{t+r}=\frac{1}{\sqrt{N}}\sum_{j=0}^{T-1}\sqrt{p_{s_{j}}}w^{s_{j}%
(t+r)}=0,r=0,1,...,T-1
\]

We see we can apply the lemma so that there are not $T$ consecutive amplitudes
of $|y>$ in the state $|\psi_{y}>$ that are all zero. Therefore if $N_{x}$ is
the number of nonzero amplitudes of $|\psi_{x}>$ and $N_{y}$ is the number of
nonzero amplitudes of $|\psi_{y}>$ then we have%
\[
N_{x}N_{y}\geq N
\]

Next we apply this result to the Amplified-QFT algorithm. We recall that we
have a set of labels $L=\{0,1,...,N-1\}$ and an oracle $f:L\rightarrow\{0,1\}$
which is $1$ on a periodic subset of labels $A$ of size $M$ and $0$ elsewhere.
We apply Grover's algorithm without measurement to arrive at the following
state (See Chapter 2)\ where $k=\left\lfloor \frac{\pi}{4\sin^{-1}(\sqrt
{M/N})}\right\rfloor $ is the number of steps of the Grover iteration:%
\[
|\psi_{x}>=a_{k}\sum_{x\in A}|x>+b_{k}\sum_{x\notin A}|x>
\]

\noindent where%

\[
a_{k}=\frac{1}{\sqrt{M}}\sin(2k+1)\theta,b_{k}=\frac{1}{\sqrt{N-M}}%
\cos(2k+1)\theta
\]

\noindent are the appropriate amplitudes of the states and where
\[
\sin\theta=\sqrt{M/N},\cos\theta=\sqrt{1-M/N}%
\]

We note that the number of non-zero amplitudes is $N$ because Grover's
algorithm puts nearly all of the probability on the set $A$ but leaves some
residual probability on $\overline{A}.$ In fact%

\begin{align*}
p(A)  &  \geq1-\frac{M}{N}\text{ and}\\
p(\overline{A})  &  \leq\frac{M}{N}%
\end{align*}

However we can still produce an uncertainty relation. Next we apply the QFT to
$|\psi_{x}>$ to obtain the state $|\psi_{y}>$%
\begin{align*}
|\psi_{y}  &  >=\sum_{y=0}^{N-1}\left[  \frac{a_{k}}{\sqrt{N}}\sum_{x\in
A}\omega^{xy}+\frac{b_{k}}{\sqrt{N}}\sum_{x\notin A}\omega^{xy}\right]  |y>\\
&  =\sum_{y=0}^{N-1}\left[  \frac{a_{k}-b_{k}}{\sqrt{N}}\sum_{x\in A}%
\omega^{xy}+\frac{b_{k}}{\sqrt{N}}\sum_{x=0}^{N-1}\omega^{xy}\right]  |y>
\end{align*}

Now, for $y\neq0$ we have
\[
\sum_{x=0}^{N-1}\omega^{xy}=0
\]

and the amplitude for $y\neq0$ is given by an $M$ term sum%

\[
\frac{a_{k}-b_{k}}{\sqrt{N}}\sum_{x\in A}\omega^{xy}%
\]

where $a_{k}-b_{k}\neq0.$

Consider the following system of $M$ equations%

\[
y_{t+r}=\frac{1}{\sqrt{N}}\sum_{j=0}^{M-1}\sqrt{p_{x_{j}}}w^{x_{j}%
(t+r)}=0,r=0,1,...,M-1
\]

where $A=\{x_{0},x_{1},...,x_{M-1}\}$ and $\sqrt{p_{x_{j}}}=a_{k}-b_{k}.$

We see we can apply the lemma so that there are not $M$ consecutive amplitudes
of $|y>$ in the state $|\psi_{y}>$ that are all zero.Therefore the number of
nonzero amplitudes of $|\psi_{y}>$ must be at least $N/M.$Therefore if $|A|=M$
and $N_{y}$ is the number of nonzero amplitudes of $|\psi_{y}>$ then we have%
\[
MN_{y}\geq N
\]

where $M$ is the number of the largest nonzero amplitudes of $|\psi_{x}>.$ If
Grover's algorithm worked perfectly $M$ would be exactly the number of nonzero
amplitudes of $|\psi_{x}>.$However since it works imperfectly, $M$ is the
number of elements whose probabilities are $>1/N$

\bigskip\newpage\renewcommand{\thechapter}{5}

\chapter{The Amplified-Haar Wavelet Transform}

\section{\textbf{Introduction}}

In the Deutsch-Jozsa problem we are given a function which is either constant
(all zeros or all ones) or balanced (is zeros half the time and ones half the
time). This problem is easily solved using the Hadamard transform followed by
a measurement. In this chapter we generalize this problem to consider the
Local Constant or Balanced Signal Decision Problem and we generalize the idea
of the amplified quantum Fourier transform in ref[14] to consider another
amplified quantum transform - the amplified 1-d Haar wavelet transform (ref 44).

Let $L=\{0,1,...,N-1\}$ be a set of $N=2^{n}$ labels and let $2M<<N$. Let%
\[
A=%
{\displaystyle\bigcup\limits_{i\in E}}
A_{i}%
\]
be a subset of $L$ of size $2M$, where $E$ is any set of $M$ even labels from
$L$, and $A_{i}=\{i,i+1\},i\in E$ are sets of consecutive labels and
$A_{i}\cap A_{j}=\phi$. Let%
\[
f:L\rightarrow\{0,1\}
\]
be an oracle which is $1$ on $A$ and $0$ elsewhere. Let
\[
S:L\rightarrow\{0,1\}
\]
be a signal. We wish to solve the following problem:

the \textbf{Local Constant or Balanced Signal Decision Problem}. \textit{We
wish to determine which of the following two possibilities are the case:}

\textit{a) On each }$A_{i}$\textit{ we can have\ }$S(i)=0$\textit{ and
}$S(i+1)=1$\textit{ or }$S(i)=1$\textit{ and }$S(i+1)=0$\textit{ - this
corresponds to the signal }$S$\textit{ being balanced on each }$A_{i}$

\textit{or}

\textit{b) On each }$A_{i}$\textit{ , }$S(i)=0$\textit{ and }$S(i+1)=0$%
\textit{ or }$S(i)=1$\textit{ and }$S(i+1)=1$\textit{ - this corresponds to
the signal }$S$\textit{ being constant on each }$A_{i}$\textit{.}

\textit{The value of the signal }$S$\textit{ on }$L\backslash A$\textit{ can
be any value in }$\{0,1\}.$

To solve the Local Constant or Balanced Signal Decision Problem, we first run
Grover's algorithm to amplify the amplitudes on the set $A$ by using the
Oracle $f$. We then put the signal $S$ into the amplitudes as $+/-1$ values
and then run the Haar Wavelet Transform on the resulting state. In the case a)
above, we will find that most of the probability lies in the following
interval of labels $[N/2,N-1]$. In case b) above, we will find that most of
the probability lies in the following interval $[0,N/2-1]$. Therefore if we
make a measurement with respect to the standard basis we can verify which
interval the measurement lies and discover whether a) is the case or b) is the
case. This algorithm works because of the special construction of the Haar
matrix, which computes sums and differences between successive values on the
even cut and puts the result in the upper half interval.

The Haar wavelet transform $W$ of dimension $2^{n}$ by $2^{n}$ has the
following form:

$W=W_{n}W_{n-1}...W_{1}$ where each $W_{k}$ is defined as

$W_{k}=\left[
\begin{array}
[c]{cc}%
H_{k} & 0\\
0 & I_{k}%
\end{array}
\right]  $where $H_{k}$ is of dimension $2^{n-k+1}$ by $2^{n-k+1}$and $I_{k}$
is the identity matrix of dimension $2^{n}-2^{n-k+1}$ by $2^{n}-2^{n-k+1}$ and
$O$ is the all zero matrix of the appropriate dimension,

where

$H_{k}=\frac{1}{\sqrt{2}}\left[
\begin{array}
[c]{cccccc}%
1 & 1 & 0 & 0 & ... & 0\\
0 & 0 & 1 & 1 & . & .\\
. & . & . & . & 0 & 0\\
0 & 0 & ... & 0 & 1 & 1\\
1 & -1 & 0 & 0 & ... & 0\\
0 & 0 & 1 & -1 & . & .\\
. & . & . & . & 0 & 0\\
0 & 0 & ... & 0 & 1 & -1
\end{array}
\right]  $

For example, consider $W$ of dimension $\ 2^{2}$ by $2^{2}$ then we have

$W=\left[
\begin{array}
[c]{cccc}%
\frac{1}{\sqrt{2}} & \frac{1}{\sqrt{2}} & 0 & 0\\
\frac{1}{\sqrt{2}} & \frac{-1}{\sqrt{2}} & 0 & 0\\
0 & 0 & 1 & 0\\
0 & 0 & 0 & 1
\end{array}
\right]  \frac{1}{\sqrt{2}}\left[
\begin{array}
[c]{cccc}%
1 & 1 & 0 & 0\\
0 & 0 & 1 & 1\\
1 & -1 & 0 & 0\\
0 & 0 & 1 & -1
\end{array}
\right]  $

Consider the example of $Wx$ where $x=\frac{1}{2}[1,1,1,1]^{T}$

We have

$Wx=\left[
\begin{array}
[c]{cccc}%
\frac{1}{\sqrt{2}} & \frac{1}{\sqrt{2}} & 0 & 0\\
\frac{1}{\sqrt{2}} & \frac{-1}{\sqrt{2}} & 0 & 0\\
0 & 0 & 1 & 0\\
0 & 0 & 0 & 1
\end{array}
\right]  \frac{1}{\sqrt{2}}\left[
\begin{array}
[c]{cccc}%
1 & 1 & 0 & 0\\
0 & 0 & 1 & 1\\
1 & -1 & 0 & 0\\
0 & 0 & 1 & -1
\end{array}
\right]  \frac{1}{2}\left[
\begin{array}
[c]{c}%
1\\
1\\
1\\
1
\end{array}
\right]  $

$=\frac{1}{2\sqrt{2}}\left[
\begin{array}
[c]{cccc}%
\frac{1}{\sqrt{2}} & \frac{1}{\sqrt{2}} & 0 & 0\\
\frac{1}{\sqrt{2}} & \frac{-1}{\sqrt{2}} & 0 & 0\\
0 & 0 & 1 & 0\\
0 & 0 & 0 & 1
\end{array}
\right]  \left[
\begin{array}
[c]{c}%
2\\
2\\
0\\
0
\end{array}
\right]  =\left[
\begin{array}
[c]{c}%
1\\
0\\
0\\
0
\end{array}
\right]  $

where the result is in the upper half as expected.

Next consider the example of $Wx$ where $x=\frac{1}{2}[1,-1,1,-1]^{T}$

We have

$Wx=\left[
\begin{array}
[c]{cccc}%
\frac{1}{\sqrt{2}} & \frac{1}{\sqrt{2}} & 0 & 0\\
\frac{1}{\sqrt{2}} & \frac{-1}{\sqrt{2}} & 0 & 0\\
0 & 0 & 1 & 0\\
0 & 0 & 0 & 1
\end{array}
\right]  \frac{1}{\sqrt{2}}\left[
\begin{array}
[c]{cccc}%
1 & 1 & 0 & 0\\
0 & 0 & 1 & 1\\
1 & -1 & 0 & 0\\
0 & 0 & 1 & -1
\end{array}
\right]  \frac{1}{2}\left[
\begin{array}
[c]{c}%
1\\
-1\\
1\\
-1
\end{array}
\right]  $

$=\frac{1}{2\sqrt{2}}\left[
\begin{array}
[c]{cccc}%
\frac{1}{\sqrt{2}} & \frac{1}{\sqrt{2}} & 0 & 0\\
\frac{1}{\sqrt{2}} & \frac{-1}{\sqrt{2}} & 0 & 0\\
0 & 0 & 1 & 0\\
0 & 0 & 0 & 1
\end{array}
\right]  \left[
\begin{array}
[c]{c}%
0\\
0\\
2\\
2
\end{array}
\right]  =\frac{1}{\sqrt{2}}\left[
\begin{array}
[c]{c}%
0\\
0\\
1\\
1
\end{array}
\right]  $

where the results are in the lower half as expected.

\section{\textbf{The Local Constant or Balanced Signal Decision
Problem-Analysis}}

The Amplified-Haar algorithm which solves the Local Constant or Balanced
Signal Decision Problem is defined by the following four step procedure.

\emph{Step 1:} Apply all of Grover's algorithm in its entirety except for the
last measurement step to the starting state $|0>$. The resulting state is
given by $|\psi_{k}>$ where $k=\left\lfloor \frac{\pi}{4\sin^{-1}(\sqrt
{2M/N})}\right\rfloor $:%
\[
|\psi_{k}>=a_{k}\sum_{z\in A}|z>+b_{k}\sum_{z\notin A}|z>
\]

\noindent where%

\[
a_{k}=\frac{1}{\sqrt{2M}}\sin(2k+1)\theta,b_{k}=\frac{1}{\sqrt{N-2M}}%
\cos(2k+1)\theta
\]

\noindent are the appropriate amplitudes of the states and where
\[
\sin\theta=\sqrt{2M/N},\cos\theta=\sqrt{1-2M/N}%
\]

.

\emph{Step 2: }We apply the signal $S$ to an auxiliary qubit $\frac{1}%
{\sqrt{2}}(|0>-|1>)$ added onto $|\psi_{k}>$ to put the signal into the
amplitudes of the state $|\psi_{k}>$ to get the state $|\lambda_{k}>$ where%

\[
|\lambda_{k}>=a_{k}\sum_{z\in A}(-1)^{S(z)}|z>+b_{k}\sum_{z\notin
A}(-1)^{S(z)}|z>
\]

\emph{Step 3:} We apply the Haar wavelet transform $W$ to the resulting state
$|\lambda_{k}>$.

\emph{Step 4: }We make a measurement $z$ and note which range the measured
value is in to determine the solution of the problem. If $z$ is in $[0,N/2-1]$
then the signal was constant on each $A_{i}$ otherwise the signal was balanced
on each $A_{i}.$

\noindent Now we have,

$k=\left\lfloor \frac{\pi}{4\theta}\right\rfloor $ $\Longrightarrow$
$\frac{\pi}{4\theta}-1\leq k\leq\frac{\pi}{4\theta}$ $\Longrightarrow$
$\frac{\pi}{2}-\theta\leq(2k+1)\theta\leq\frac{\pi}{2}+\theta$

$\Longrightarrow\sin\theta=\cos(\frac{\pi}{2}-\theta)\geq\cos(2k+1)\theta
\geq\cos(\frac{\pi}{2}+\theta)=-\sin\theta$

\noindent Notice that the total probability of the $N-2M$ labels that are not
in $A$ is%

\begin{align*}
(N-2M)(\frac{1}{\sqrt{N-2M}}\cos(2k+1)\theta)^{2}  &  =\cos^{2}(2k+1)\theta\\
&  \Longrightarrow\cos^{2}(2k+1)\theta\leq\sin^{2}\theta=\sin^{2}(\sin
^{-1}(\sqrt{\frac{2M}{N}}))\\
&  \Longrightarrow\cos^{2}(2k+1)\theta\leq\frac{2M}{N}%
\end{align*}
whereas the total probability of the $2M$ labels in $A$ is
\begin{align*}
2M(\frac{1}{\sqrt{2M}}\sin(2k+1)\theta)^{2}  &  =\sin^{2}(2k+1)\theta
=1-\cos^{2}(2k+1)\theta\\
&  \Longrightarrow\sin^{2}(2k+1)\theta\geq1-\frac{2M}{N}%
\end{align*}

We notice that in Step 3, after we have applied the first orthogonal transform
$W_{1}$ of $W$ we have essentially solved our problem. If the signal is
constant on $A$ then the total probability of the labels in $[0,N/2-1]$ is at
least $\frac{1}{2}M(a_{k}+a_{k})^{2}=2M(\frac{1}{\sqrt{2M}}\sin(2k+1)\theta
)^{2}=\sin^{2}(2k+1)\theta\geq1-\frac{2M}{N}$. So we have moved most of the
probability of the set $A$ to the lower half range of labels. Similarly if the
signal is balanced on $A$ then this probability would be moved to the upper
half range of labels $[N/2,N-1]$. The successive remaining orthogonal
transforms $W_{n}...W_{2}$ do not move the probabilities outside of these
ranges. So we see we need only apply $W_{1}$ to solve this problem and make a
measurement. The work factor of this algorithm is dominated by the Grover step
and which is $O(\sqrt{\frac{N}{2M}}).$

A classical solution to this problem would be to randomly choose labels $x$ in
the range $[0,N-1]$ and to verify that $f(x)=1$. If $x$ is even then we check
the values of $s(x)$ and $s(x+1)$ to see what kind of signal we have. If $x$
is odd, we check $s(x)$ and $s(x-1)$. This procedure has workfactor $O(N/2M)$
showing that the amplified-Haar wavelet transform is quadratically faster.

In order to consider a quantum algorithm that would solve the Local Constant
or Balanced Signal Problem without using amplification, we need to consider
the values of the signal on the set $L\backslash A$ \ We can consider the
following problem where we want to find out which situation is the case:

a) The set $A$ is constant~and the set $L\backslash A$ is balanced

or

b) The set $A$ is balanced and the set $L\backslash A$ is constant.

Suppose we perform the Haar transform on the signal $S$ corresponding to these
situations and make a measurement. Regardless of which case we are in, if the
measured value is in the interval $[0,N/2-1]$ we have measured a value due to
the probability of the constant part of the signal, whereas if we measure a
value in the range $[N/2,N-1]$ we have measured a value due to the probability
of the balanced part of the signal. We are performing sampling of a
probability distribution and we need to determine which case we are in. If we
repeat this process we can estimate the means of the Binomial probability
distributions we are sampling from.

What are the means and variances of the Binomial distributions we are sampling from?

Letting $p$ be probability of the constant signal and $q$ be the probability
of the balanced signal we have:

case a) $p_{a}=M/N$ and $q_{a}=1-M/N$ with variance $\sigma^{2}=Np_{a}%
q_{a}=M(1-M/N)$

and in

case b) $p_{b}=1-M/N$ and $q_{b}=M/N$ with variance $\sigma^{2}=Np_{b}%
q_{b}=M(1-M/N)$ (note the variances are the same)

Suppose we make a series of $n$ measurements with $c$ measurements from the
constant interval and $b$ measurements from the balanced interval. Then we can
get an estimator for $p$ which we will denote $\widehat{p}=c/n$ which is
normally distributed as $Normal(p_{a},\sigma^{2}/n)$ in case a) and which is
normally distributed as $Normal(p_{b},\sigma^{2}/n)$ in case b).

We want a sample size $n$ such that these two distributions intersect at
$p_{a}+3$ $\sigma/\sqrt{n}$ in case a)\ and $p_{b}-3$ $\sigma/\sqrt{n}$ in
case b). This gives us a sample size $n$ that is large enough that we
determine case a) if $\widehat{p}<p_{a}+3$ $\sigma$ and case b) if
$\widehat{p}>p_{b}-3$ $\sigma.$

We have
\begin{align*}
M/N+3\sqrt{M(1-M/N)/n}  &  =(1-M/N)-3\sqrt{M(1-M/N)/n}\\
&  \Rightarrow6\sqrt{M(1-M/N)/n}=1-2M/N\\
&  \Rightarrow n=\frac{36M(1-M/N)}{(1-2M/N)^{2}}\\
&  \sim36M\text{ when }M<<N
\end{align*}

In order for the Amplified-Haar transform to win we need the work factor of
the above method to be worse than the work factor of the Amplified-Haar
transform which is $O(\sqrt{N/2M}).$ This gives the following inquality:%

\begin{align*}
\frac{36M(1-M/N)}{(1-2M/N)^{2}}  &  >\sqrt{N/2M}\\
&  \Rightarrow\frac{1296M^{2}(1-M/N)^{2}}{(1-2M/N)^{4}}>N/2M\\
&  \Rightarrow M^{3}>\frac{N(1-2M/N)^{4}}{2592(1-M/N)^{2}}\\
&  \Rightarrow M>\frac{N^{1/3}(1-2M/N)^{4/3}}{2592^{1/3}(1-M/N)^{2/3}}%
\end{align*}

So we see that in this problem situation the Amplified-Haar transform wins if
$M>N^{1/3}$ approximately speaking.

We should note that in the more general setting, if the set $L\backslash A$ is
a mixture of balanced and constant components then just performing the Haar
transform alone will not help to solve the problem of determining the nature
of the set $A$ because the probabilities of the set of $A$ become impacted by
the makeup of $L\backslash A.$ The Amplified-Haar transform is not affected by
this and is able to easily solve this more general situation.

\newpage

\renewcommand{\thechapter}{6}

\chapter{REFERENCES}

\bigskip The following references are directly relevant to the chapters of
this thesis.

[1] Nakahara and Ohmi, \textquotedblleft Quantum Computing: From Linear
Algebra to Physical Realizations\textquotedblright, CRC Press (2008).

[2] S. Lomonaco, \textquotedblleft Shor's Quantum Factoring
Algorithm,\textquotedblright\ AMS PSAPM, vol. 58, (2002), 161-179.

[3] P. Shor, \textquotedblleft Polynomial time algorithms for prime
factorization and discrete logarithms on a quantum computer\textquotedblright,
SIAM J. on Computing, 26(5) (1997) pp1484-1509 (quant-ph/9508027).

[4] L. Grover, \textquotedblleft A fast quantum mechanical search algorithm
for database search\textquotedblright, Proceedings of the 28th Annual ACM
Symposium on Theory of Computing (STOC 1996), (1996) 212-219.

[5] Hardy and Wright \textquotedblleft An Introduction to the Theory of
Numbers\textquotedblright, Oxford Press Fifth Edition (1979).

[6] S. Lomonaco and L. Kauffman, \textquotedblleft Quantum Hidden Subgroup
Algorithms: A Mathematical Perspective,\textquotedblright\ AMS CONM, vol. 305,
(2002), 139-202.

[7] S. Lomonaco, \textquotedblleft Grover's Quantum Search
Algorithm,\textquotedblright\ AMS PSAPM, vol. 58, (2002), 181-192.

[8] S. Lomonaco and L. Kauffman, \textquotedblleft Is Grover's Algorithm a
Quantum Hidden Subgroup Algorithm?,\textquotedblright\ Journal of Quantum
Information Processing, Vol. 6, No. 6, (2007), 461-476.

[9] G. Brassard, P. Hoyer, M. Mosca and A. Tapp, "Quantum Amplitude
Amplification and Estimation", AMS CONM, vol 305, (2002), 53-74.

[10] M. Nielsen and I. Chuang, "Quantum Computation and Quantum Information",
Cambridge University Press (2000).

[11] P. Kaye, R. Laflamme and M. Mosca, "An Introduction to Quantum
Computing", Oxford University Press (2007).

[12] N. Yanofsky and M. Mannucci, "Quantum Computing For Computer Scientists",
Cambridge University Press (2008).

[13] S. Lomonaco, "A Lecture on Shor's Quantum Factoring Algorithm Version
1.1",quant-ph/0010034v1 9 Oct 2000.

[14] Cornwell, D., \textquotedblleft The amplified quantum Fourier transform:
solving the local period problem\textquotedblright, Quantum Inf Process (2013)
12: 1225-1253.

[15] Donoho and Stark, "Uncertainty Principles and Signal Recovery", SIAM J.
Appl Math, Vol 49, No 3, pp. 906-93 (1989)

[16] Massar and Spindel, "Uncertainty Relations for the Discrete Fourier
Transform", quant-ph, arXiv:0710.0723v2 (2008)

[17] Loo, K, "Quantum Algorithm Uncertainty Principles", math-ph,

arXiv:math-ph/0210007v2 (2004)

The following references are general references of relevance to the topics of
this thesis.

1.\qquad S. Aaronson, \textquotedblleft The Equivalence of Searching and
Sampling\textquotedblright,

arXiv:1009.5104 [quant-ph]

2.\qquad S. Aaronson and A. Ambainis, \textquotedblleft Quantum search of
spatial regions\textquotedblright, arXiv:quant-ph/0303041

3.\qquad S. Aaronson, \textquotedblleft Quantum lower bound on recursive
Fourier sampling\textquotedblright,

arXiv:quant-ph/0209060

4.\qquad G. Abal, R. Donangelo, M. Forets and R. Portugal, \textquotedblleft
Spatial quantum search in a triangular network\textquotedblright,
arXiv:1009.1422 [quant-ph]

5.\qquad A. Ambainis,\textquotedblright Quantum Algorithms\textquotedblright,arXiv:quant-ph/0504012.

6.\qquad A. Ambainis, A. Backurs, N. Nahimovs, R. Ozols and A. Rivosh,
\textquotedblleft Search by quantum walks on two-dimensional grid without
amplitude amplification\textquotedblright,

arXiv:1112.3337 [quant-ph]

7.\qquad A. Ambainis and A. Montanaro, \textquotedblleft Quantum algorithms
for search with wildcards and combinatorial group testing\textquotedblright,
arXiv:1210.1148 [quant-ph]

8.\qquad A. Ambainis, \textquotedblleft Quantum search with variable
times\textquotedblright,

arXiv:quant-ph/0609168

9.\qquad A. Ambainis, \textquotedblleft A better lower bound for quantum
algorithms searching an ordered list\textquotedblright, arXiv:quant-ph/9902053

10.\qquad A. Barenco, A. Ekert, K. Suominem and P. Torma,\textquotedblright
Approximate Fourier transform and decoherence\textquotedblright,arXiv:quant-ph/9601018

11.\qquad M. Ben-Or and A. Hassidim, \textquotedblleft Quantum search in an
ordered list via adaptive learning\textquotedblright

12.\qquad S. Berry and J. Wang, \textquotedblleft Quantum Walk-based search
and centrality\textquotedblright, arXiv:1010.0764 [quant-ph]

13.\qquad E. Biham, O. Biham, D. Biron, M. Grassl and D. Lidar,
\textquotedblleft Grover's quantum search algorithm for an arbitrary initial
amplitude distribution\textquotedblright,

arXiv:quant-ph/9807027

14.\qquad O. Biham, D. Shapira and Y. Shimoni, \textquotedblleft Analysis of
Grover's quantum search algorithm as a dynamical system\textquotedblright, arXiv:quant-ph/0307141

15.\qquad D. Biron, O. Biham, E. Biham, M. Grassl and D. Lidar,
\textquotedblleft Generalized Grover search algorithm for arbitrary initial
amplitude distribution\textquotedblright, arXiv:quant-ph/9801066

16.\qquad C. Bowden, G. Chen, Z. Diao and A. Klappenecker, \textquotedblleft
The universality of the quantum Fourier transform in forming the basis of
quantum computing algorithms\textquotedblright, arXiv:quant-ph/0007122

17.\qquad M. Boyer, G. Brassard, P. Hoeyer and A. Tapp, \textquotedblleft
Tight bounds on quantum searching\textquotedblright, arXiv:quant-ph/9605034

18.\qquad G. Brassard, P. Hoyer and A. Tapp,\textquotedblright Quantum
Counting\textquotedblright,\textquotedblright arXiv:quant-ph/9805082.

19.\qquad G. Brassard, P. Hoyer, M. Mosca and A. Tapp, "Quantum Amplitude
Amplification and Estimation", AMS CONM, vol 305, (2002), 53-74.

20.\qquad H. Burhman and R. de Wolf, \textquotedblleft Lower bounds for
quantum search and derandomization\textquotedblright, arXiv:quant-ph/9811046

21.\qquad C. Cafaro and S. Mancini, \textquotedblleft On Grover's Search
Algorithm from a Quantum Information Geometry Viewpoint\textquotedblright%
,arXiv.1110.6713 [quant-ph]

22.\qquad N. Cerf, L. Grover and C. Williams, \textquotedblleft Nested quantum
search and NP-complete problems\textquotedblright, arXiv:quant-ph/9806078

23.\qquad S. Chakraborty, S. Adhikari, \textquotedblleft Non-classical
Correlations in the Quantum Search Algorithm\textquotedblright,
arXiv:1302.6005v1 [quant-ph]

24.\qquad S. Chakraborty, S. Banerjee, S. Adhikari and A. Kumar,
\textquotedblleft Entanglement in the Grover's Search
Algorithm\textquotedblright, arXiv: 1305.4454 [quant-ph]

25.\qquad A. Chamoli and S. Masood, \textquotedblleft Two-Dimensional Quantum
Search Algorithm\textquotedblright, arXiv:1012.5629 [quant-ph]

26.\qquad A. Chamoli and M. Bhandari, \textquotedblleft Success rate and
entanglement evolution in search algorithm\textquotedblright, arXiv:quant-ph/0702221

27.\qquad J. Chappell, M. Lohe, L. Smekal, A. Iqbal and D. Abbot,
\textquotedblleft An improved formalism for the Grover search
algorithm\textquotedblright, arXiv:1201.1707 [quant-ph]

28.\qquad J. Chen and H. Fan, \textquotedblleft Quantum mechanical
perspectives and generalization of the fractional Fourier
transform\textquotedblright, arXiv:1307.6271

29.\qquad A. Childs and T. Lee, \textquotedblleft Optimal quantum adversary
lower bounds for ordered search\textquotedblright, arXiv:0708.3396

30.\qquad A. Childs, A. Landahl and P.Parrilo, \textquotedblleft Improved
quantum algorithms for the ordered search problem via semidefinite
programming\textquotedblright,

arXiv:quant-ph/0608161

31.\qquad A. Childs and J. Goldstone, \textquotedblleft Spatial search by
quantum walk\textquotedblright,

arXiv:quant-ph/0306054

32.\qquad A. Childs, E. Deotto, E. Farhi, J. Goldstone, S. Gutmann, A.
Landahl, \textquotedblleft Quantum search by measurement\textquotedblright, arXiv:quant-ph/0204013

33.\qquad B. Choi and V. Korepin, \textquotedblleft Quantum partial search of
a database with several target items\textquotedblright, arXiv:quant-ph/0608106

34.\qquad B. Choi, T. Walker and S. Braunstein, \textquotedblleft Sure success
partial search\textquotedblright, arXiv:quant-ph/0603136

35.\qquad R. Cleve and J. Watrous,\textquotedblright Fast parallel circuits
for the quantum Fourier transform\textquotedblright,arXiv:quant-ph/0006004

36.\qquad Cornwell, D., "The amplified quantum Fourier transform: solving the
local period problem", Quantum Inf Process (2013) 12: 1225-1253.

37.\qquad D. Coppersmith, \textquotedblleft An approximate Fourier transform
useful in quantum factoring\textquotedblright,arXiv:quant-ph/0201067

38.\qquad J.Cui and H. Fan, \textquotedblleft Correlations in Grover
Search\textquotedblright, arXiv:0904:1703 [quant-ph]

39.\qquad Z. Diao. \textquotedblleft Exactness of the Original Grover Search
Algorithm\textquotedblright,

arXiv:1010.3652 [quant-ph]

40.\qquad S. Dolev, I. Pitowsky and B. Tamir, \textquotedblleft Grover's
quantum search algorithm and Diophantine approximation\textquotedblright, arXiv:quant-ph/0507234

41.\qquad K. Dorai and D. Suter, \textquotedblleft Efficient implementations
of the quantum Fourier transform : an experimental
perspective\textquotedblright, arXiv:quant-ph/0211030

42.\qquad M. Falk, \textquotedblleft Quantum Search on the Spatial
Grid\textquotedblright, arXiv:1303.4127 [quant-ph]

43.\qquad E. Farhi and S. Gutmann,\textquotedblright Quantum mechanical square
root speedup in a structured search problem\textquotedblright, arXiv:quant-ph/9711035

44.\qquad A. Fijany and C. Williams, \textquotedblleft Quantum wavelet
transforms: fast algorithms and complete circuits\textquotedblright, arXiv:quant-ph9809004

45.\qquad D. Floess, E. Andersson and M. Hillery, \textquotedblleft Quantum
algorithms for testing Boolean functions\textquotedblright, arXiv:1006.1423v1 [quant-ph]

46.\qquad P. Gawron, J. Klemka and R. Winiarcyzk, \textquotedblleft Noise
effects in the quantum search algorithm from the computational complexity
point of view\textquotedblright,

arXiv:1108.1915 [quant-ph]

47.\qquad M. Gocwin, \textquotedblleft On the complexity of searching maximum
of a function on a quantum computer\textquotedblright, arXiv:quant-ph/0507060

48.\qquad R. Griffiths and C. Niu, \textquotedblleft Semiclassical Fourier
transform for quantum computation\textquotedblright, arXiv:quant-ph/9511007

49.\qquad L. Grover, "A fast quantum mechanical search algorithm for database
search", Proceedings of the 28th Annual ACM Symposium on Theory of Computing
(STOC 1996), (1996) 212-219.

50.\qquad L. Grover, \textquotedblleft A fast quantum mechanical search
algorithm for database search\textquotedblright, arXiv:quant-ph/9605043

51.\qquad L. Grover, \textquotedblleft Quantum computers can search
arbitrarily large databases by a single query\textquotedblright, arXiv:quant-ph/9706005v3.

52.\qquad L. Grover,\textquotedblright Quantum Search on Structured
Problems\textquotedblright,

arXiv:quant-ph/9802035

53.\qquad L. Grover, \textquotedblleft Searching with quantum
computers\textquotedblright,

arXiv:quant-ph/0011118

54.\qquad L. Grover,\textquotedblright Tradeoffs in the Quantum Search
Algorithm\textquotedblright,arXiv:quant-ph/0201152.

55.\qquad L. Grover,\textquotedblright Quantum Searching amidst
Uncertainty\textquotedblright,

arXiv:quant-ph/0507116.

56.\qquad L. Grover,\textquotedblright Superlinear amplitude
amplification\textquotedblright,

arXiv:0806.0154 [quant-ph]

57.\qquad L. Grover, \textquotedblleft A different kind of quantum
search\textquotedblright,

arXiv:quant-ph/0503205

58.\qquad L. Grover, \textquotedblleft How fast can a quantum computer
search\textquotedblright,arXiv:quant-ph/9809029

59.\qquad L. Grover, \textquotedblleft Quantum computers can search rapidly by
using almost any transformation\textquotedblright,arXiv:quant-ph/9712011

60.\qquad L. Grover, \textquotedblleft Quantum mechanics helps in searching
for a needle in a haystack\textquotedblright,arXiv:quant-ph/9706033

61.\qquad L. Grover and J Radhakrishnan,\textquotedblright Quantum search for
multiple items using parallel queries\textquotedblright,arXiv:quant-ph/0407217.

62.\qquad L. Grover and J. Radhakrishnan, \textquotedblleft Is partial quantum
search of a database any easier\textquotedblright, arXiv:quant-ph/0407122

63.\qquad L. Grover and T.Rudolph,\textquotedblright Creating superpositions
that correspond to efficiently integrable probability
distributions\textquotedblright,arXiv:quant-ph/0208112.

64.\qquad L. Grover, A. Patel and T. Tulsi, \textquotedblleft Quantum
algorithms with fixed points: the case of database search\textquotedblright, arXiv:quant-ph/0603132

65.\qquad L. Gyongyosi and S. Imre, \textquotedblleft An improvement in
quantum Fourier transform\textquotedblright, arXiv:1207.4464

66.\qquad L.Hales, \textquotedblleft The quantum Fourier transform and
extensions of the Abelian hidden subgroup problem\textquotedblright, arXiv:quant-ph/0212002

67.\qquad L. Hales and S. Hallgren,\textquotedblright Sampling Fourier
transforms on different domains\textquotedblright,arXiv:quant-ph/9812060

68.\qquad Hardy and Wright "An Introduction to the Theory of Numbers", Oxford
Press Fifth Edition (1979).

69.\qquad B. Hein and G. Tanner, \textquotedblleft Quantum search algorithms
on a regular lattice\textquotedblright, arXiv:1005:3676 [quant-ph]

70.\qquad B. Hein and G. Tanner, \textquotedblleft Quantum search algorithms
on the hypercube\textquotedblright, arXiv:0906.3094 [quant-ph]

71.\qquad M. Hillery, D. Reitzner and V. Bunek, \textquotedblleft Searching
via walking: How to find a marked subgraph of a graph using quantum
walks\textquotedblright, arXiv:0911:1102 [quant-ph]

72.\qquad T. Hogg,\textquotedblright Single-Step Quantum Search Using Problem
Structure\textquotedblright,

arXiv:quant-ph/9812049.

73.\qquad T. Hogg, \textquotedblleft A framework for structured quantum
search\textquotedblright,arXiv:quant-ph/9701013

74.\qquad T. Hogg, \textquotedblleft A framework for quantum search
heuristics\textquotedblright,arXiv:quant-ph/9611004

75.\qquad T. Hogg and M. Yanik,\textquotedblright Local search methods for
quantum computers\textquotedblright,

arXiv:quant-ph/9802043

76.\qquad P. Hoyer, J. Neerbek and Y. Shi,\textquotedblright Quantum
complexities of ordered searching, sorting and element
distinctness\textquotedblright, arXiv:quant-ph/0102078

77.\qquad J. Hsieh, C. Li, J. Lin and D. Chu, \textquotedblleft Formulation of
a family of sure success quantum search algorithms\textquotedblright, arXiv:quant-ph/0210201

78.\qquad M. Hunziker, D. Meyer, J. Park, J. Pommersheim and M. Rothstein,

\textquotedblright The Geometry of Quantum Learning\textquotedblright,arXiv:quant-ph/0309059.

79.\qquad N. Ilano, C. Villagonzalo and R. Banzon, \textquotedblleft Analysis
of the damped quantum search and its application to the one-dimensional Ising
system\textquotedblright,

arXiv:1208.5509v1 [quant-ph]

80.\qquad N. Ilano, C. Villagonzalo and R. Banzon, \textquotedblleft
Optimization of the damped quantum search\textquotedblright, arXiv:1208.5475 [quant-ph]

81.\qquad L. Ip, \textquotedblleft Solving shift problems and hidden coset
problem using the Fourier transform\textquotedblright,arXiv:quant-ph/0205034

82.\qquad S. Iriyama, M. Ohya, I.V. Volovich, \textquotedblleft On Quantum
Algorithm for Binary Search and its Computational Complexity\textquotedblright%
, arXiv:1306.5039v1 [quant-ph]

83.\qquad S. Ivanov, H. Tonchev and N. Vitanov, \textquotedblleft
Time-efficient implementation of quantum search with qudits\textquotedblright,
arXiv:1209.4489 [quant-ph]

84.\qquad R. Josza, \textquotedblleft Searching in Grover's
algorithm\textquotedblright, arXiv:quant-ph/9901021

85.\qquad R. Josza, \textquotedblleft Quantum algorithms and the Fourier
transform\textquotedblright,arXiv:quant-ph/9707033

86.\qquad P. Kaye, R. Laflamme and M. Mosca, "An Introduction to Quantum
Computing", Oxford University Press (2007).

87.\qquad V. Korepin and Y. Xu, \textquotedblleft Binary quantum
search\textquotedblright, arXiv:0705.0777

88.\qquad V. Korepin and J. Liao, \textquotedblleft Quest for fast partial
search algorithm\textquotedblright, arXiv:quant-ph/0510179

89.\qquad V. Korepin and L. Grover, \textquotedblleft Simple algorithm for
partial quantum search\textquotedblright, arXiv:quant-ph/0504157

90.\qquad V. Korepin, \textquotedblleft Optimization of partial
search\textquotedblright, arXiv:quant-ph/0503238

91.\qquad K. Kumar and G. Paraoanu, \textquotedblleft A quantum no reflection
theorem and the speeding up of Grover's search algorithm\textquotedblright,
arXiv:1105.4032 [quant-ph]

92.\qquad T. Laarhoven, M. Mosca and J. van de Pol, \textquotedblleft Solving
the Shortest Vector Problem in Lattices Faster Using Quantum
Search\textquotedblright, arXiv:1301.6176v1 [quant-ph]

93.\qquad C. Lavor, L. Manssur and R. Portugal, \textquotedblleft Grover's
algorithm: quantum database search\textquotedblright, arXiv:quant-ph/03010179

94.\qquad J. Lee, H. Lee and M. Hillery, \textquotedblleft Searches on star
graphs and equivalent oracle problems\textquotedblright, arXiv:1102:5480 [quant-ph]

95.\qquad T. Lee, F. Magniez and M. Santha, \textquotedblleft Improved Quantum
Query Algorithms for Triangle Finding and Associativity
Testing\textquotedblright, arXiv:1210.1014v1 [quant-ph]

96.\qquad S. Lloyd,\textquotedblright Quantum search without
entanglement\textquotedblright,

arXiv:quant-ph/9903057

97.\qquad S. Lomonaco, "Grover's Quantum Search Algorithm," AMS PSAPM, vol.
58, (2002), 181-192.

98.\qquad S. Lomonaco, "Shor's Quantum Factoring Algorithm," AMS PSAPM, vol.
58, (2002), 161-179.

99.\qquad S. Lomonaco, "A Lecture on Shor's Quantum Factoring Algorithm
Version 1.1",quant-ph/0010034v1 9 Oct 2000.

100.\qquad S. Lomonaco and L. Kauffman, "Quantum Hidden Subgroup Algorithms: A
Mathematical Perspective," AMS CONM, vol. 305, (2002), 139-202.

101.\qquad S. Lomonaco and L. Kauffman, "Is Grover's Algorithm a Quantum
Hidden Subgroup Algorithm?," Journal of Quantum Information Processing, Vol.
6, No. 6, (2007), 461-476.

102.\qquad C. Lomont, \textquotedblleft A quantum Fourier transform
algorithm\textquotedblright, arXiv:quant-ph/0404060

103.\qquad N. Lovett, M. Everitt, R. Heath and V. Kendon, \textquotedblleft
The quantum walk search algorithm: Factors affecting
efficiency\textquotedblright, arXiv:1110.4366v2 [quant-ph]

104.\qquad N. Lovett, M. Everitt, M. Trevers, D. Mosby, D. Stockton and V.
Kendon, \textquotedblleft Spatial search using the discrete time quantum
walk\textquotedblright, arXiv:1010:4705 [quant-ph]

105.\qquad F. Magniez, A. Nayak, J. Roland and M. Santha, \textquotedblleft
Search via quantum walk\textquotedblright, arXiv:quant-ph/0608026

106.\qquad A. Mani and C. Patvardhan, \textquotedblleft A Fast measurement
based fixed-point Quantum Search Algorithm\textquotedblright, arXiv:1102.2332 [quant-ph]

107.\qquad A. Mani and C. Patvardhan, \textquotedblleft A Fast fixed-point
Quantum Search Algorithm by using Disentanglement and
Measurement\textquotedblright, arXiv:1203.3178 [quant-ph]

108.\qquad F. Marquezino, R. Portugal and S. Boettcher, \textquotedblleft
Quantum Search Algorithms on Hierarchical Networks\textquotedblright,
arXiv:1205.0529 [quant-ph]

109.\qquad D. Meyer and T. Wong, \textquotedblleft Nonlinear Quantum Search
Using the Gross-Pitaevskii Equation\textquotedblright, arXiv:1303.0371v3 [quant-ph]

110.\qquad F. Marquezino, R. Portugal and S. Boettcher, \textquotedblleft
Spatial Search Algorithms on Hanoi Networks\textquotedblright, arXiv:1209.2871 [quant-ph]

111.\qquad A. Mizel, \textquotedblleft Critically damped quantum
search\textquotedblright, arXiv:0810.0470 [quant-ph]

112.\qquad A. Montanaro,\textquotedblright Quantum search with
advice\textquotedblright,arXiv:0908.3066 [quant-ph]

113.\qquad A. Montanaro, \textquotedblleft Quantum search of partially ordered
sets\textquotedblright,arXiv:quant-ph/0702196

114.\qquad C. Moore, D. Rockmore and A. Russell, \textquotedblleft Generic
quantum Fourier transforms\textquotedblright, arXiv:quant-ph/0304064

115.\qquad C. Moore, D. Rockmore, A. Russell and L. Schulman,
\textquotedblleft The power of strong Fourier sampling: quantum algorithms for
affine groups and hidden shifts\textquotedblright, arXiv:quant-ph/0503095

116.\qquad M. Mosca,\textquotedblright Quantum Algorithms\textquotedblright%
,arXiv:0808.0369 [quant-ph]

117.\qquad M. Mosca and C. Zalka,\textquotedblright Exact quantum Fourier
transforms and discrete logarithm algorithms\textquotedblright,arXiv:quant-ph/0301093.

118.\qquad Y. Most, Y. Shimoni and O. Biham, \textquotedblleft Entanglement of
periodic states, the quantum Fourier transform and Shor's factoring
algorithm\textquotedblright, arXiv:1001:3145

119.\qquad Nakahara and Ohmi, "Quantum Computing: From Linear Algebra to
Physical Realizations", CRC Press (2008).

120.\qquad A. Nesterov and G. Berman, \textquotedblleft Quantum search using
non-Hermitian adiabatic evolution\textquotedblright,arXiv:1208.4642 [quant-ph]

121.\qquad M. Nielsen and I. Chuang, "Quantum Computation and Quantum
Information", Cambridge University Press (2000).

122.\qquad R. Orus, J. Latorre and M. Martin-Delgado,\textquotedblright
Natural majorisation of the quantum Fourier transform in phase-estimation
algorithms\textquotedblright, arXiv:quant-ph/0206134

123.\qquad S. Parasa and K. Eswaran, \textquotedblleft Quantum pseudo
fractional Fourier transform and its application to quntum phase
estimation\textquotedblright,arXiv:0906.1033

124.\qquad A. Patel, \textquotedblleft Quantum Algorithms: Database Search and
its Variations\textquotedblright, arXiv:1102.2058 [quant-ph]

125.\qquad A. Patel, K. Raghunathan and M. Rahaman, \textquotedblleft Search
on a Hypercubic Lattice through a quantum random walk: II.
d=2\textquotedblright, arXiv:1003.5664 [quant-ph]

126.\qquad A. Patel and M. Rahaman, \textquotedblleft Search on a Hypercubic
Lattice through a quantum random walk: I. d%
$>$%
2\textquotedblright, arXiv:1003.0065 [quant-ph]

127.\qquad A. Perez, \textquotedblleft Non adiabatic quantum search
algorithms\textquotedblright,arXiv:0706.1139

128.\qquad A. Pittenger and M. Rubin, \textquotedblleft Complete separability
and Fourier representations of n-qubit states\textquotedblright,arXiv:quant-ph/9912116

129.\qquad A. Pittenger and M. Rubin,\textquotedblright Separability and
Fourier representations of density matrices\textquotedblright,arXiv:quant-ph/0001014

130.\qquad V. Potocek, A. Gabris, T. Kiss and I. Jex, \textquotedblleft
Optimized quantum random-walk search algorithms\textquotedblright, arXiv:0805.4347.

131.\qquad R. Qu, J. Wang, Z. Li, Y. Bao and X. Cao, \textquotedblleft
Multipartite entanglement and Grover's search algorithm\textquotedblright, arXiv:1210.3418

132.\qquad D. Bhaktavatsala Rao and K. Molmer, \textquotedblleft Effect of
qubit losses on Grover's quantum search algorithm\textquotedblright,
arXiv:1209.0637 [quant-ph]

133.\qquad R. Ramos, P. de Sousa, D. Oliveira, \textquotedblleft Solving
mathematical problems with quantum search algorithm\textquotedblright, arXiv:quanth-ph/0605003

134.\qquad O. Regev and L. Schiff, \textquotedblleft Impossibility of a
Quantum Speed-up with a Faulty Oracle\textquotedblright, arXiv:1202.1027v1 [quant-ph]

135.\qquad H. Roehrig, \textquotedblleft Searching an ordered list on a
quantum computer\textquotedblright, arXiv:quant-ph/9812061

136.\qquad M. Rossi, D. Brus and C. Macchiavello, \textquotedblleft Scale
invariance of entanglement dynamics in Grover's quantum search
algorithm\textquotedblright, arXiv:1205.3000 [quant-ph]

137.\qquad P. Rungta, \textquotedblleft The quadratic speedup in Grover's
search algorithm from the entanglement perspective\textquotedblright, arXiv:0707.1410

138.\qquad M. Santha, \textquotedblleft Quantum walk based search
algorithms\textquotedblright, arXiv:0808.0059 [quant-ph]

139.\qquad D. Shapira, S. Mozes and O. Biham, \textquotedblleft The effect of
unitary noise on Grover's quantum search algorithm\textquotedblright, arXiv:quant-ph/0307142

140.\qquad N. Shenvi, K. Brown and K. Whaley, \textquotedblleft Effects of
noisy oracle on search algorithm complexity\textquotedblright. arXiv:quant-ph/0304138

141.\qquad N. Shenvi, J. Kempe and K. Whaley, \textquotedblleft A quantum
random walk search algorithm\textquotedblright, arXiv:quant-ph/0210064

142.\qquad P. Shor, "Polynomial time algorithms for prime factorization and
discrete logarithms on a quantum computer", SIAM J. on Computing, 26(5) (1997)
pp1484-1509 (arXiv:quant-ph/9508027).

143.\qquad P. Shor, \textquotedblright Introduction to Quantum
Algorithms\textquotedblright,

arXiv:quant-ph/0005003.

144.\qquad R. Sufiani and N. Bahari, \textquotedblleft Quantum search in
structured database using local adiabatic evolution and spectral
methods\textquotedblright,arXiv:1208.0262 [quant-ph]

145.\qquad R. Tucci, \textquotedblleft Quantum fast Fourier transform viewed
as a special case of recursive application of cosine-sine
decomposition\textquotedblright, arXiv:quant-ph/0411097

146.\qquad A. Tulsi, \textquotedblleft Optimal quantum searching to find a
common element of two sets\textquotedblright, arXiv:1210.04648 [quant-ph]

147.\qquad A. Tulsi, \textquotedblleft General framework for quantum search
algorithms\textquotedblright,

arXiv:0806.1257 [quant-ph]

148.\qquad A. Tulsi, \textquotedblleft Faster quantum walk algorithm for the
two dimensional spatial search\textquotedblright, arXiv:0801.0497

149.\qquad A. Tulsi, \textquotedblleft Quantum computers can search rapidly
using almost any selective transformation\textquotedblright, arXiv:0711.4299

150.\qquad J. Tyson, \textquotedblleft Operator-Schmidt decomposition of the
quantum Fourier transform on C\symbol{94}N1 tensor C\symbol{94}%
N2\textquotedblright, arXiv:quant-ph/0210100

151.\qquad P.Vrana, D. Reeb, D. Reitzner and M. Wolf, \textquotedblleft
Fault-ignorant Quantum Search\textquotedblright, arXiv:1307.0771v1 [quant-ph]

152.\qquad N. Yanofsky and M. Mannucci, "Quantum Computing For Computer
Scientists", Cambridge University Press (2008).

153.\qquad A. Younes, \textquotedblleft Strength and Weakness in Grover's
Quantum Search Algorithm\textquotedblright, arXiv:0811.4481 [quant-ph]

154.\qquad A. Younes, \textquotedblleft Constant-Time Quantum Search Algorithm
for the Unstructured Search Problem\textquotedblright, arXiv:08114247 [quant-ph]

155.\qquad A. Younes, \textquotedblleft Fixed phase quantum search
algorithms\textquotedblright, arXiv:0704.1585

156.\qquad A. Younes, J. Rowe and J. Miller, \textquotedblleft A hybrid
quantum search algorithm: a fast quantum algorithm for multiple
matches\textquotedblright, arXiv:quant-ph/0311171

157.\qquad C. Zalka, \textquotedblleft A Grover-based quantum search of
optimal order for an unknown number of marked elements\textquotedblright, arXiv:quant-ph/9902049

158.\qquad C. Zalka, \textquotedblleft Grover's quantum searching algorithm is
optimal\textquotedblright,

arXiv:quant-ph/9711070

159.\qquad M. Zakaria, \textquotedblleft Binary Subdivision for Quantum
Search\textquotedblright,

arXiv:1101.4703 [quant-ph]

The following is a list of books:

1.\qquad A. Aezel,\textquotedblright Entanglement\textquotedblright,Plume, 2001

2.\qquad Y. Aharonov and D. Rohrlich,\textquotedblright Quantum
paradoxes\textquotedblright, Wiley,2009

3.\qquad G. Van Assche,\textquotedblright Quantum cryptography and secret-key
distillation\textquotedblright, Cambridge,2006

4.\qquad J. Audretsch (Ed),\textquotedblright Entangled
world\textquotedblright.Wiley, 2002

5.\qquad J. Baggot,\textquotedblright The meaning of quantum
theory\textquotedblright,Oxford,1994

6.\qquad I. Bengtsson and K. Zyczkowski,\textquotedblright\ Geometry of
quantum states\textquotedblright, Cambridge,2008

7.\qquad D. Bernstein, J. Buchmann and E. Dahmen (Eds), \textquotedblleft
Post-quantum cryptography\textquotedblright, Springer, 2009

8.\qquad D. Bouwmeester, A. Ekert, A, Zeilinger (Eds), \textquotedblleft The
physics of quantum information\textquotedblright, Springer,2000

9.\qquad G. Chen, L. Kauffman and S. Lomonaco,\textquotedblright Mathematics
of quantum computation and quantum technology\textquotedblright, Chapman \& Hall,2008

10.\qquad I. Daubechies, \textquotedblleft Ten lectures on
wavelets\textquotedblright,SIAM Vol 61, 2006

11.\qquad L. Debnath,\textquotedblright Wavelet transforms and their
applications\textquotedblright, Birkhauser, 2002

12.\qquad P. Dirac, \textquotedblleft The principles of quantum
mechanics\textquotedblright,Oxford, 1999

13.\qquad R. Feynman,\textquotedblright QED: The strange theory of light and
matter\textquotedblright,Princeton, 1985

14.\qquad G. Gamow, \textquotedblleft Thirty years that shook physics: the
story of quantum theory\textquotedblright,Dover, 1985

15.\qquad J. Gruska, \textquotedblleft Quantum computing\textquotedblright,
McGraw Hill, 1999

16.\qquad M. Hirvensalo, \textquotedblleft Quantum computing\textquotedblright%
, Springer, 2004

17.\qquad P. Kaye, R. LaFlamme and M. Mosca, \textquotedblleft An Introduction
to Quantum Computing\textquotedblright, OUP 2007.

18.\qquad M. Kumar,\textquotedblright Quantum: Einstein, Bohr and the great
debate about the nature of reality\textquotedblright, Norton, 2008

19.\qquad F. Laloe, \textquotedblleft Do we really understand quantum
mechanics\textquotedblright,Cambridge, 2013

20.\qquad S. Lomonaco (Ed.), \textquotedblleft Quantum Computation: A grand
mathematical challenge for the twenty first century and the
millennium\textquotedblright, AMS Vol 58, 2000

21.\qquad S. Lomonaco (Ed.),\textquotedblright Quantum Computation and
Information\textquotedblright,AMS 305

22.\qquad D. McMahon, \textquotedblleft Quantum computing
explained\textquotedblright,Wiley, 2008

23.\qquad N. Mermin, \textquotedblleft Quantum computer science -- An
introduction\textquotedblright,

Cambridge,2007

24.\qquad M. Nakahara and T. Ohmi,\textquotedblright Quantum computing: from
linear algebra to physical realizations\textquotedblright,CRC Press, 2008

25.\qquad M. Neilsen and I. Chuang,\textquotedblright Quantum computation and
quantum information\textquotedblright, Cambridge 2000

26.\qquad Y. Nievergelt,\textquotedblright Wavelets made
easy\textquotedblright,Birkhauser\textquotedblright,2001

27.\qquad J. Polkinghorne, \textquotedblleft The quantum
world\textquotedblright,Princeton, 1989

28.\qquad R. Portugal, \textquotedblleft Quantum walks and search
algorithms\textquotedblright, Springer, 2013

29.\qquad A. Rae, \textquotedblleft Quantum Physics: Illusion or
Reality\textquotedblright, Canto,2005

30.\qquad E. Rieffel and W. Polak, \textquotedblleft Quantum computing: A
gentle introduction\textquotedblright,MIT, 2011

31.\qquad A. Terras,\textquotedblright Fourier Analysis on Finite Groups and
Applications\textquotedblright,LMS vol 43, 2001

32.\qquad P. Van Fleet,\textquotedblright Discrete Wavelet
Transformations\textquotedblright,Wiley,2008

33.\qquad J. Walker, \textquotedblleft A primer on Wavelets and their
scientific applications\textquotedblright, Chapman \& Hall, 2008

34.\qquad D. Walnut,\textquotedblright An introduction to wavelet
analysis\textquotedblright, Birkhauser, 2004

35.\qquad R. Wang,\textquotedblright Introduction to orthogonal
transforms\textquotedblright,Cambridge, 2012

36.\qquad A. Whitaker, \textquotedblleft The new quantum age\textquotedblright%
,Oxford, 2012

37.\qquad A. Whitaker,\textquotedblright Einstein, Bohr and the quantum
dilemma\textquotedblright,

Cambridge,2006

38.\qquad C. Williams, \textquotedblleft Explorations in Quantum
Computing\textquotedblright, Springer 2011

39.\qquad N. Yanofsky and M. Mannucci, \textquotedblleft Quantum computing for
computer scientists\textquotedblright,Cambridge, 2008

\newpage\ \thispagestyle{empty}
\end{document}